%% file: sn-article.tex
\let\LaTeXcline\cline
\documentclass[pdflatex,sn-mathphys-num]{sn-jnl}
\let\cline\LaTeXcline

\usepackage{fullpage} 

\usepackage{ifthen}
\newboolean{originalformat} 
\setboolean{originalformat}{true} 

\usepackage{graphicx}%
\usepackage{multirow}%
\usepackage{amsmath,amssymb,amsfonts}%
\usepackage{amsthm}%
\usepackage{mathrsfs}%
\usepackage[title]{appendix}%
\usepackage[dvipsnames,table]{xcolor}%
\usepackage{textcomp}%
\usepackage{manyfoot}%
\usepackage{booktabs}%
\usepackage{algorithm}%
\usepackage{algorithmicx}%
\usepackage{algpseudocode}%
\usepackage{listings}%
\usepackage{hyperref}
\usepackage{physics}
\usepackage{tikz}
\usepackage{mathtools}
\usepackage{thmtools}

\newtheoremstyle{slanted}
  {0.4\topsep}
  {\topsep}
  {\slshape}
  {}
  {\bfseries}
  {.}
  {0.5em}
  {}

\theoremstyle{slanted}
\newtheorem{theorem}{Theorem}[section]
\newtheorem*{theorem*}{Theorem}
\newtheorem{proposition}[theorem]{Proposition}%
\newtheorem{question}{Question}

\newtheorem{definition}{Definition}[section]%
\newtheorem*{property*}{Property}%
\newtheorem{lemma}{Lemma}[section]%
\newtheorem{corollary}[theorem]{Corollary}

\declaretheoremstyle[%
  spaceabove=-6pt,%
  spacebelow=6pt,%
  headfont=\normalfont\itshape,%
  postheadspace=1em,%
  qed=\qedsymbol%
]{mystyle} 
\declaretheorem[name={Proof},style=mystyle,unnumbered,
]{prf}

\begin{document}

\title[Article Title]{Quantum Wave Atom Transforms}


\author[1,3]{\fnm{Marianna} \sur{Podzorova}}

\author[2,3]{\fnm{Yi-Kai} \sur{Liu}}


\affil[1]{\orgname{\small Department of Computer Science, University of Maryland},
\city{College Park}, 
\state{MD}}

\affil[2]{\orgname{\small Applied and Computational Mathematics Division, National Institute of Standards and Technology (NIST)},
\city{Gaithersburg}, 
\state{MD}}

\affil[3]{\orgname{\small Joint Center for Quantum Information and Computer Science (QuICS), NIST/University of
Maryland},
\city{College Park}, 
\state{MD}}

\renewcommand{\i}{\mathbf{i}}


\abstract{This paper constructs the first {efficient implementation of a quantum} wavelet packet transform with a ``parabolic scaling'' tree structure, sometimes called {a quantum} wave atom transform. Classically, wave atom transforms are used to construct sparse representations of differential operators, which enable fast {classical} algorithms for {solving wave equations}. 
Compared to previous work {on quantum wavelet transforms}, our quantum algorithm can implement a larger class of wavelet and wave atom transforms, by using an efficient representation for a larger class of possible tree structures. Our quantum implementation has $O(\mathrm{poly}(n))$ gate complexity for applying a transform of dimension $2^n$, while classical implementations use $O(n 2^n)$ floating point operations. {This is potentially useful for designing} quantum algorithms for solving wave equations {that achieve an exponential speedup over classical algorithms}.}

\maketitle

\input{Introduction/01-Introduction}

\input{Preliminaries/02-Preliminaries}

\input{Encoding/03-Encoding}

\input{Shannon/04-Shannon}

\input{Wave-atoms/05-Wave-atoms}

\input{06-Conclusion}


\bibliography{sn-bibliography}%


\newboolean{showappendix}
\setboolean{showappendix}{true} 

\ifthenelse{\boolean{showappendix}}{

\ifthenelse{\boolean{originalformat}}{
 \begin{appendices}
 \input{Appendix/00}
 \end{appendices}
}{
 \pagebreak
 \setcounter{page}{1}
 \begin{center}
  {\LARGE Supplementary Material for\\ ``Quantum Wave Atom Transforms''}
  \vskip 12pt
  {\large Marianna Podzorova$^{1,3}$ and Yi-Kai Liu$^{2,3}$}
  \vskip 12pt
  $^1$Department of Computer Science, University of Maryland, College Park, MD
  \vskip 6pt
  $^2$Applied and Computational Mathematics Division, National Institute of Standards and Technology
(NIST), Gaithersburg, MD
  \vskip 6pt
  $^3$Joint Center for Quantum Information and Computer Science (QuICS), NIST/University of Maryland,
College Park, MD
  \vskip 6pt
 \end{center}
 \setcounter{section}{0}
 \renewcommand\thesection{\Alph{section}}
 \input{Appendix/00}
}

}{
}


\end{document}

%% file: Introduction/01-Introduction.tex
\section{Introduction}\label{section:Introduction}
The quantum Fourier transform plays an important role in many quantum algorithms that achieve an exponential speedup over classical algorithms \cite{Jozsa_1998, childs2010quantum}. {Thus it is natural to ask whether quantum wavelet transforms could be similarly useful? Here, a \textit{quantum wavelet transform} is a unitary operation that can be performed on a quantum computer with $n$ qubits.\footnote{{In principle, one can also consider more general kinds of operations, such as isometries or quantum channels. These can be implemented using a quantum computer with ancilla qubits, or an open quantum system.}} It acts on quantum superposition states:
\begin{equation}\label{eqn-intro-qwt}
    \sum_{x=0}^{2^n-1} f(x) \ket{x} \mapsto \sum_{j,m,n\in\Gamma} c_{j,m,n} \ket{j,m,n},
\end{equation}
i.e., given a superposition state whose complex amplitudes encode some complex-valued function $f$ on the spatial domain $\lbrace 0,1,\ldots,2^n-1 \rbrace \subset \mathbb{R}$, the quantum wavelet transform returns a superposition state whose complex amplitudes encode the coefficients $c_{j,m,n} \in \mathbb{C}$ that appear when $f$ is expanded in some wavelet basis. Typically these coefficients $c_{j,m,n}$ are indexed by variables $(j,m,n)$ that range over some set $\Gamma$, and have an interpretation in terms of scale, frequency and position. Here we have assumed that $f$ is defined on a 1-dimensional spatial domain, but generalizations to more dimensions are also possible.}

{There are two natural questions that arise, when using wavelets to design quantum algorithms:
\begin{question}\label{question-1}
Can a quantum wavelet transform be implemented efficiently on a quantum computer, that is, can it be implemented using a family of quantum circuits whose size grows at most polynomially with the number of qubits $n$?
\end{question}
\noindent
This is different from the implementation of classical wavelet transforms on classical computers, in two respects: the number of gates in the quantum circuit is only allowed to grow polylogarithmically (not polynomially) with the dimension of the vector $(f(0),f(1),\ldots,f(2^n-1))$, and the individual gates in the quantum circuit can act simultaneously on the entire vector ``in superposition'' (though the action of each gate can only depend on a constant-sized subset of the qubits).
\begin{question}\label{question-2}
Can one use a quantum wavelet transform to design a quantum algorithm that solves a useful computational problem, exponentially faster than any classical algorithm?
\end{question}
\noindent
Here, the adjective \textit{useful} expresses the idea that one ultimately wants to solve a computational problem that has a broader motivation, beyond the mathematical theory of wavelets. This requires combining the quantum wavelet transform in Eq.~(\ref{eqn-intro-qwt}) with other techniques that are relevant to the problem being solved.
}

{The answers to Questions \ref{question-1} and \ref{question-2} both depend crucially on the choice of a particular wavelet transform. In this paper, we consider a family of wavelet transforms, known as \textit{wave atom} transforms \cite{DemanetThesis, DEMANET2007368, Villemoes2002WaveletPW}, which are used in classical algorithms for solving wave equations and simulating wave propagation \cite{Demanet2009-zt, Demanet2012}, but (to the best of our knowledge) have not been applied to the design of quantum algorithms.}

{There is a particularly strong motivation for studying \textit{quantum} wave atom transforms: if such an operation can be implemented efficiently on a quantum computer, then there is a promising (albeit heuristic) approach for using this operation to obtain fast quantum algorithms for solving wave equations. (We will describe this in more detail below.) If these steps could be made rigorous, they would provide positive answers to both Questions \ref{question-1} and \ref{question-2}, which would have practical applications in optics, seismology and other fields.}

{Yet for wave atoms, both of the above questions have remained open, due to the complexity of the wave atom construction. For example, wave atoms are constructed using a tree structure, which is not handled in previous works on efficient implementation of quantum wavelet transforms \cite{hoyer1997efficientquantumtransforms, fijany1998quantumwavelettransformsfast, Klappenecker1999-ol, Arguello, Li2019-mp, Zhang2022-bt, Bagherimehrab_2024, ni2024quantumwavepackettransforms}. Furthermore, in order to use a wave atom transform to solve a wave equation, one constructs sparse approximations of the solution operator, which are quite delicate. These sparse approximations have been studied in the context of classical algorithms \cite{candes2004curveletrepresentationwavepropagators, DEMANET2007368, DemanetThesis, Candes2003-vd}, but not in the context of quantum algorithms, which depend on different notions of sparsity and approximation error.}

{Our main contribution in this paper is to overcome the first difficulty (involving the tree structure), leading to the first efficient implementation of a quantum wave atom transform. This provides a positive answer to Question \ref{question-1}, and is a step towards addressing the technical issues involved in Question \ref{question-2}. In the following sections, we will describe the background and motivation for wave atoms, the strengths and weaknesses of our results, and how this relates to previous work, in more detail.}

{\subsection{Wavelets and wave atoms}\label{sec-intro-wavelets}}

In classical signal processing, wavelets are a powerful tool for constructing \textit{sparse representations} of complicated signals, using specially-tailored basis functions that are well-concentrated in both time and frequency. This has motivated the study of a vast variety of different wavelet families, characterized by various choices of vanishing moments, size of wave packet support, symmetry, smoothness, and other properties: Shannon, Meyer, Haar, Battle-Lemarié, or Daubechies compactly supported wavelets \cite{ 10lectures,WaveletTourofSignalProcessing}, Coiflets \cite{Daubechies1993-jm}, Ricker hats \cite{Ricker}, Gabor atoms \cite{Gabor1946},  wave atoms \cite{DemanetThesis, DEMANET2007368, Demanet2009-zt, Villemoes2002WaveletPW}, and curvelets \cite{candes2004curveletrepresentationwavepropagators, DemanetThesis}. 

Among these families of wavelets, wave atoms and curvelets are known to provide especially sparse representations for certain differential operators that arise in the study of wave equations \cite{Candes2003-vd,candes2004curveletrepresentationwavepropagators, DEMANET2007368, DemanetThesis,Demanet2009-zt, Villemoes2002WaveletPW}. This allows their application as a potentially effective preconditioner for partial differential equations \cite{DemanetThesis, Demanet2012}. 
So, some wave packet bases can be valuable for creating sparse representations of linear operators. 

{These sparse representations depend on a specific geometric property of the wave atom basis functions: the basis functions obey a so-called \textit{parabolic scaling law}, 
\begin{equation}\label{eqn-parabolic}
    (\text{wavelength}) \propto (\text{diameter})^2,
\end{equation}
where the wavelength (of each wave packet) is proportional to the square of the diameter (of the set on which the wave packet is supported). This is in contrast to conventional wavelets, which obey a linear scaling law (wavelength $\propto$ diameter), and Gabor wavelets, where the diameter is fixed while the wavelength is allowed to vary \cite{DEMANET2007368}.}

{In order to satisfy this scaling law, the wave atom basis functions are constructed by a recursive procedure, which can be viewed as building a \textit{wave packet tree} \cite{WaveletTourofSignalProcessing}, where the branching pattern of the tree is determined by the scaling law. This branching pattern is one of the main features that distinguish wave atoms from other wavelet bases. For most wavelet bases, the branching pattern is simple, e.g., it corresponds to a complete binary tree, or a dyadic tree, shown in Fig.~\ref{fig:Simple-tree-example}. Wave atoms require a branching pattern that is less regular, and lies somewhere in between a complete binary tree and a dyadic tree, as shown in Fig.~\ref{fig-wave-atom-tree-example}.} 

\subsection{Our results}

{In this paper, we address Question \ref{question-1}, by constructing efficient quantum circuits for implementing a large family of quantum wavelet transforms that are characterized by wave packet trees (Theorems \ref{theorem:algorithms-shannon} and \ref{theorem:complexity-general}). This construction is quite general, subject to some mild technical conditions (e.g., the existence of efficiently-computable encodings of the branching patterns of the trees, and the existence of efficient quantum circuits for performing certain 2-dimensional unitary rotations that are related to the shape of the mother wavelet).}

{By making appropriate choices for the mother wavelet and the branching pattern of the tree, we obtain the first efficient implementation a quantum wave atom transform that satisfies a parabolic scaling law (Corollary \ref{theorem:complexity-specific}). As mentioned above, this parabolic scaling law is crucial for constructing sparse representations of linear operators that appear, e.g., in solving wave equations \cite{DEMANET2007368}.}

{Note that there is an issue of approximation error, when running the quantum circuits described in this paper. These quantum circuits provide an \textit{exact} implementation of the desired quantum wave atom transform, assuming the ability to perform arbitrary 1- and 2-qubit gates. In practice, however, one only has the ability to perform 1- and 2-qubit gates that are chosen from a fixed, finite set $G$, which generates a dense subgroup of SU(2). Thus, one must use the gates in $G$ to approximate the arbitrary 1- and 2-qubit gates that appear in our quantum circuits for the quantum wave atom transform. This will increase the size of the quantum circuits, and produce some nonzero error in their output. However, these effects can be controlled, while preserving the efficiency of the implementation, using standard techniques. For example, using the Solovay-Kitaev theorem \cite{Nielsen_Chuang_2010}, for any $\epsilon>0$, one can ensure that the error in the output of the quantum circuit is at most $\epsilon$ in operator norm, while the size of the quantum circuit is increased by at most a multiplicative factor that is polylogarithmic in $n$ and $1/\epsilon$.}

{\subsection{Related work}}

{There is a large literature on quantum wavelet transforms \cite{hoyer1997efficientquantumtransforms, fijany1998quantumwavelettransformsfast, Klappenecker1999-ol, Arguello, Li2019-mp, Zhang2022-bt, Bagherimehrab_2024, ni2024quantumwavepackettransforms}. Many of these works are complementary to ours: they solve technical problems related to generating quantum superposition states that have more complicated shapes (compared to the ones used in our work), but they do not handle wave packet trees with complicated branching patterns (which is the main focus of our construction). As discussed in Section \ref{sec-intro-wavelets}, the use of wave packet trees with complex branching patterns is crucial for applications that involve solving wave equations.}

{Among these previous works, the work of Ni et al \cite{ni2024quantumwavepackettransforms} is particularly relevant, because it provides efficient implementations of quantum transforms using wavelets and Gabor atoms, using a ``blending'' technique, which we adapt to our setting in Section \ref{section:Decomposition}. Our main contribution, compared to \cite{ni2024quantumwavepackettransforms}, is to handle a much larger class of wave packet trees, namely ``monotonic trees,'' defined in Section \ref{section: Wave packet trees}; these include parabolic-scaling trees as a special case.}

{Our results are a first demonstration that a quantum wave atom transform can be implemented efficiently. This raises an interesting question: can one extend these techniques to work with different choices of the mother wavelet? Recent work by Bagherimehrab et al \cite{Bagherimehrab_2024}, on implementing quantum wavelet transforms using a powerful technique called linear combinations of unitaries, may be helpful here. (Note that the constructions in \cite{Bagherimehrab_2024} are quite flexible, and can already handle some wavelet transforms that have simple tree structures. The hope is that these techniques can be extended to handle wave atom transforms, where the tree structure is more elaborate.)}

{Finally, we expect that our results can be extended to handle wave atoms defined on $\mathbb{R}^2$, by adapting ideas from the existing literature on classical wave atom transforms \cite{DEMANET2007368, DemanetThesis}.}

{\subsection{Potential applications of our results}}

{Turning to Question \ref{question-2}, we now describe how our results are potentially useful for quantum algorithms for solving wave equations. At a high level, we argue that the special properties of wave atom bases (namely, the existence of sparse representations of wave propagators using wave atom bases) can be used to improve the performance of quantum algorithms for solving wave equations and differential equations, such as those developed in Costa et al \cite{costa2019quantum} and Bagherimehrab et al \cite{bagherimehrab2023fastquantumalgorithmdifferential}.}

{Suppose one wants to solve a linear wave equation, describing wave propagation in an inhomogenous medium in some finite region of space $\Omega \subset \mathbb{R}^d$, with appropriate boundary conditions. One can use standard techniques to encode this problem into a linear system of equations $Ax=b$ in dimension $N$, where $N$ is proportional to the number of points used to discretize the domain $\Omega$. (One can also use similar techniques to encode the problem into the dynamics of a quantum Hamiltonian \cite{costa2019quantum}.)} 

One can then run a quantum algorithm, such as the HHL (Harrow-Hassidim-Lloyd) algorithm \cite{Harrow_2009, morales2024quantum}, to solve this system of equations, with running time that scales polylogarithmically with $N$ --- an exponential speedup over classical algorithms. (A large body of work has shown similar speedups for solving differential equations, and simulating Hamiltonian time evolution; see, e.g., \cite{berry2014high, childs2018toward, jennings2024cost}.)

However, in order to use the HHL algorithm, certain conditions must be satisfied: the linear system must be \textit{well-conditioned}, and it must have an extremely compact \textit{description}, which is efficiently computable, using space and time polylogarithmic in $N$. Furthermore, there must be ways to \textit{prepare} the quantum state that encodes $b$, and \textit{read-out} interesting properties of the solution $x$, which are similarly efficient (i.e., running in time polylogarithmic in $N$). While these are conditions are quite restrictive, they can be satisfied in certain cases, particularly using techniques that involve sparsity \cite{clader2013preconditioned, berry2014exponential}. 

{The hope is that the above conditions can be satisfied by transforming the system of equations $Ax=b$ into a wave atom basis, using a quantum wave atom transform. That is, one computes matrix elements of $A$ with respect a wave atom basis, and one performs state preparation and read-out in a wave atom basis. Note that this necessarily requires a \textit{quantum} wave atom transform, rather than a classical transform, since we are seeking an algorithm that runs in time polylogarithmic in $N$.}

{There are heuristic reasons why this approach might satisfy the conditions needed by the HHL algorithm. First, in cases where the solution to the wave equation satisfies a conservation law (e.g., conservation of energy), one may expect the system of equations $Ax=b$ to be well-conditioned. (A similar intuition was used in the quantum algorithm of \cite{costa2019quantum}.)}

{Second, in light of classical results on sparse representations of wave propagators \cite{Candes2003-vd, candes2004curveletrepresentationwavepropagators, DEMANET2007368, DemanetThesis}, one may hope that the system of equations $Ax=b$ will be sufficiently sparse in the wave atom basis, so that the computation of $A$, the preparation of the state encoding $b$, and the read-out of the solution $x$, can all be performed in time polylogarithmic in $N$. (Some aspects of this idea were investigated previously in \cite{bagherimehrab2023fastquantumalgorithmdifferential}, which showed that a quantum algorithm for solving differential equations can be sped up by preconditioning using a wavelet basis. The new elements in our work are the choice of a wave atom basis, and the role of sparsity.)}

{Our construction of an efficient quantum wave atom transform is an important step towards realizing this approach to solving wave equations on a quantum computer. However, significant obstacles remain. One obstacle is that this approach requires sparse approximations of wave propagators, where the approximation error is small in operator norm. It is an interesting question whether these kinds of bounds can be proved by adapting the known techniques for bounding the approximation error, which often use the vector $\ell_\infty$ norm.}

{If this approach is successful, it would provide a new quantum algorithm for solving wave equations, which performs well when the solution of the wave equation is sparse in the wave atom basis. In such cases, this quantum algorithm might achieve an exponential speedup over classical algorithms, and it might also outperform existing quantum algorithms that use discretization on a grid \cite{costa2019quantum} or a wavelet basis \cite{bagherimehrab2023fastquantumalgorithmdifferential} (rather than a wave atom basis).}



{Finally, we mention a few other ways in which our efficient implementation of a quantum wave atom transform might be useful for quantum algorithms. One possibility is to use a quantum wave atom transform to \textit{post-process} the quantum state that is obtained from a quantum differential equation solver, as studied in \cite{Kiani_2022} (using wavelet bases). Another (more speculative) possibility is to use a quantum wave atom transform to solve certain geometric problems in $\mathbb{R}^n$, or to prepare quantum superpositions that are potentially useful for solving high-dimensional lattice problems, as studied in \cite{liu2009quantumalgorithmsusingcurvelet, liu2023uncertaintyprinciplecurvelettransform} (using curvelet bases).}





{\subsection{Organization of this paper}}
 \input{Introduction/01-Organization}
 \vskip 11pt


%% file: Introduction/01-Organization.tex
Section \ref{section:preliminaries} discusses \ifthenelse{\boolean{originalformat}}{preliminaries}{methods}, including notation and background information about wave packets and wavelet packet trees. \ifthenelse{\boolean{originalformat}}{In}{In particular, in} Section \ref{section: Encoding} we discuss the encoding of wave packets and frequencies, along with an important discussion of the reverse procedure, decoding. This section establishes a foundation to build wave packet transform generated by multilevel wavelet packet trees. In Section \ref{section:quantum-shannon-wavelet-transform} we demonstrate the use of encoding/decoding procedures on the example of the Shannon wavelet transform. Section \ref{section: quantum wave atom transform} provides a quantum algorithm for the wave atom transform, which is the main result of the paper.
Finally, Section \ref{section: conclusion} discusses directions for further research. 
Additional details of the algorithms can be found in the \ifthenelse{\boolean{originalformat}}{Appendix}{Supplementary Material}.

%% file: Preliminaries/02-Preliminaries.tex
\ifthenelse{\boolean{originalformat}}{
 \section{Preliminaries}
}{
 \section{Methods}
}
\label{section:preliminaries}
\input{Preliminaries/02-Notation}
\input{Preliminaries/02-Wave-Packets}

\input{Preliminaries/02-Wave-Packet-Trees}

%% file: Preliminaries/02-Notation.tex
\subsection{Notations and assumptions}
A universal gate set, such as Clifford+$T$, enables the approximation of any unitary transformation with arbitrary precision $\varepsilon$ \cite{Harrow_2002, Nielsen_Chuang_2010}. In analyzing gate complexity, this work assumes unrestricted access to all single-qubit and two-qubit gates. Consequently, the gate complexity is expressed in terms of the number of single-qubit and two-qubit gates, while disregarding the $O(1/\varepsilon)$ factor associated with the approximation. This paper implicitly assumes that a given quantum circuit and its controlled version, with a constant number of control qubits, exhibit the same asymptotic gate complexity. However, the gate complexity may increase when accounting for the quantum device topology, particularly in scenarios where qubits are not fully connected \cite{Beals_2013}.

Throughout the paper, we use some standard quantum gates given in Eq.~( \ref{eq:standard-gates}) with circuit representations described in Figure \ref{fig:circuit-notation}. $X$, $Z$ and $H$ represent $X$-Pauli, $Z$-Pauli, and Hadamard gates, respectively. The  $ \mathrm{SWAP}$ gate swaps the state of two qubits. Phase shift gates $P_k$ map basis states as $P_k\ket{0} = \ket{0}$ and $P_k\ket{1} = e^{2\pi \i / 2^k} \ket{1}$, where $\i$ is the imaginary unit and $\ket{\cdot}$ are vectors in the Dirac notation such that $\ket{0} = (1, 0)^T$ and $\ket{1} = (0, 1)^T$, $\bra{\cdot}$ is a conjugate transpose of the vector. $U^\dagger$ is a conjugate transpose of matrix $U$ and $\overline{z}$ is a conjugate of a complex number $z$. $\otimes$ denotes the tensor product of matrices, and $\oplus$ represents a direct sum of subspaces in one case and a modulo $2$ sum in another, with the meaning clear from the context. Finally, $I(\cdot)$ is an indicator function with some condition, for example, $I( x \le 0)=1$ if $x \le 0$, and $0$ otherwise, while $I_M$ is an identity matrix of size $M \times M$.

\begin{equation}\label{eq:standard-gates}
\begin{gathered}
P_k = \begin{pmatrix}
1 & 0\\
0 & e^{2\pi \i / 2^k}
\end{pmatrix}
\quad
X = \begin{pmatrix}
0 & 1 \\ 1 & 0
\end{pmatrix}\\
Z = \begin{pmatrix}
1 & 0 \\ 0 & -1
\end{pmatrix}
\quad
H = \frac{1}{\sqrt{2}} \begin{pmatrix}
1 & 1 \\
1 & -1
\end{pmatrix}
\end{gathered}\quad
\mathrm{SWAP} = \begin{pmatrix}
1 & 0 & 0 & 0 \\
0 & 0 & 1 & 0 \\
0 & 1 & 0 & 0 \\
0 & 0 & 0 & 1
\end{pmatrix}
\end{equation}

\begin{figure}
\begin{minipage}[c]{8.5cm}
\vspace{0.25cm}
\caption{Circuit representations of standard quantum gates.\\Left column (from top to bottom): phase shift $P_k$-, $X$-, $Z$- and $H$-gates.\\ Center-top: $\mathrm{SWAP}$ gate.\\ Center-bottom: controlled-$U$ gate, which is equivalent to $\ket{1}\bra{1}\otimes U + \ket{0}\bra{0}\otimes I_2$\\Right: example of controlled-$U$ gate, when applied to qubits $3$ and $4$ if qubit~$1$~is~$\ket{0}$ and qubit $2$ is $\ket{1}$. In this case controlled-$U$ is equivalent to $\ket{01}\bra{01}\otimes U + (I_4 - \ket{01}\bra{01}) \otimes I_4$.}\label{fig:circuit-notation}
\end{minipage}
\hspace{0.5cm}
    \begin{minipage}[c]{3.5cm}
      \includegraphics[width=3.5cm]{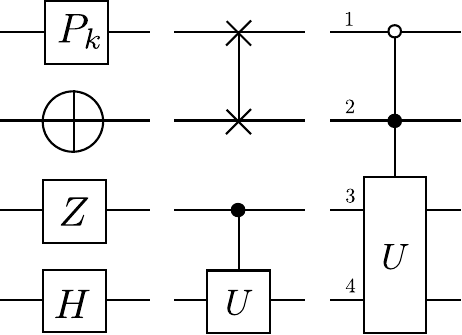}
    \end{minipage}
\end{figure}

\begin{figure}
    \centering
    \includegraphics[scale=0.5,trim= 0 0 540 0,clip]{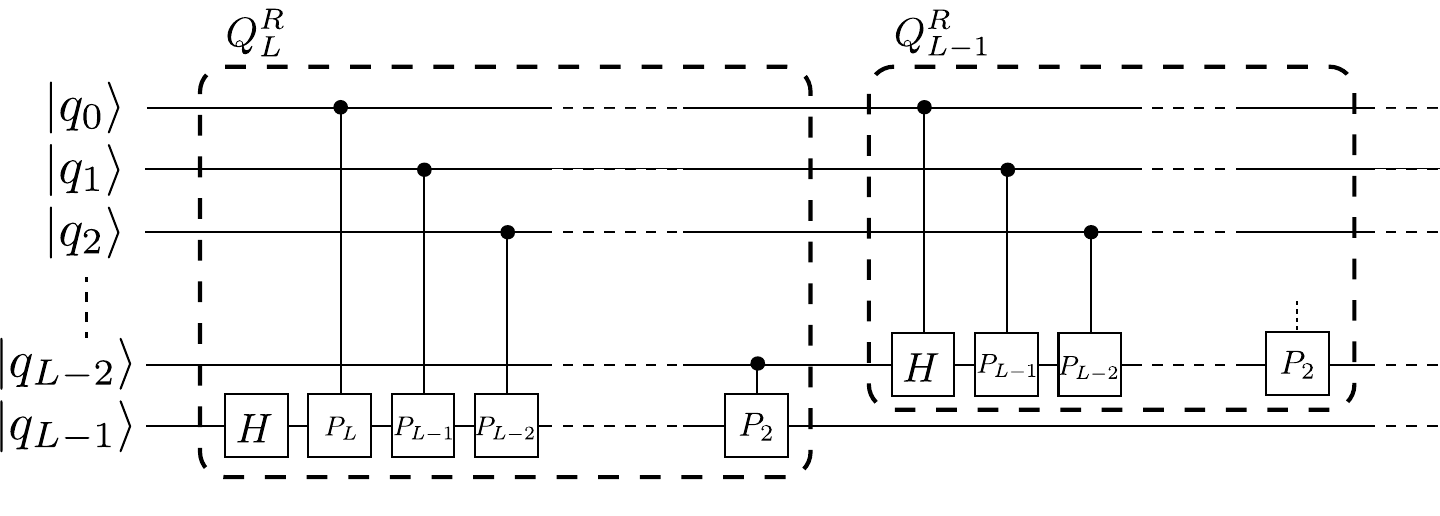}\hspace{1in}
    \\
    \hspace{1in}\includegraphics[scale=0.5,trim= 680 0 0 0,clip]{Preliminaries/Figures/qft.pdf}
    \caption{{The circuit of the QFT} on $L$ qubits, consisting of QFT rotations $Q_i^R$, followed by a QFT swap $Q_L^S$ (see Eq.~(\ref{eq:fourier})).}
    \label{fig:qft-circuit}
\end{figure}

\noindent We denote a subroutine of integer comparison on $j$ qubits to some value $k$ as $\mathrm{Comparator}_j(k)$, formally,
\begin{equation}
\mathrm{Comparator}_j(k)\ket{x}\ket{b} \mapsto \ket{x}\ket{b+I(x \ge k)\bmod 2},
\end{equation}
which can be implemented by adding two's complement of $-k$ with carry bit representing the result, which in turn requires $O(j)$ gates. We denote a subroutine of cyclic shift of $j$ qubits by some value $k$ as $\mathrm{Adder}_j(k)$, formally,
\begin{equation}\label{eq: adder}
\mathrm{Adder}_j(k) \ket{x} \mapsto \ket{x + k \text{ mod } 2^j}.
\end{equation}
There are various implementations of this operation \cite{Vedral_1996, cuccaro2004newquantumripplecarryaddition, draper2000additionquantumcomputer} with various numbers of ancilla qubits required. \cite{Vedral_1996, cuccaro2004newquantumripplecarryaddition} require $O(j)$ gates and $O(j)$ ancilla qubits. \cite{draper2000additionquantumcomputer} implements the cyclic shift with the use of the quantum Fourier transform (QFT) \cite{coppersmith2002approximatefouriertransformuseful} that requires at most $O(j^2)$ gates and no ancilla qubits. 

We extensively use QFT and its subroutines. Recall the transform of size $N=2^L$,
$$
QFT: \ket{x} \mapsto \frac{1}{\sqrt{N}} \sum_{k=0}^{N-1} e^{2\pi \i x k / N} \ket{k},
$$
where $\ket{x}$ and $\ket{k}$ are basis states of $L$ qubits. Moreover,
$$
\ket{k} = \ket{k_{L-1} \ldots k_0} = \ket{k_{L-1}}\otimes \ket{k_{L-2}} \otimes \ldots \otimes \ket{k_0}
$$
where $k=\sum_{j=0}^{L-1} 2^j k_j$ or $k_{L-1}\ldots k_0$ is a binary representation of $k$. This also implies that we follow the big-endian convention.

The implementation of the QFT of $L$ qubits, $Q_L$, has a decomposition:
\begin{equation}\label{eq:fourier}
Q_L = Q^S_L Q^R_1 \ldots Q^R_L = Q^S_L \left(\prod_{i=1}^{L} Q^R_{i}\right),
\end{equation}
where $Q^R_i$ is a QFT rotation of $i$ qubits and $Q^S_L$ reverses the order of $L$ qubits. QFT rotation $Q^R_i$ comprises of $H$-gate applied to qubit $i-1$ followed by phase gates $P_{i-j+1}$ applied to qubit $i-1$ controlled by $j$ for $j=0,\ldots, i-2$. $Q^S_L$ reverses the order by swapping each pair of qubits $i$ and $L-i-1$. Figure \ref{fig:qft-circuit} depicts a circuit of the QFT of $L$ qubits and highlights its subroutines such as QFT rotations $Q_i^R$ and QFT swap $Q_L^S$.

The classical counterpart of the QFT is the discrete Fourier transform (DFT) which has a different sign of the exponent. Specifically, the DFT of size $N$ in Dirac's notation is
$$
DFT: \ket{x} \mapsto \frac{1}{\sqrt{N}} \sum_{k=0}^{N-1} e^{-2\pi \i x k / N} \ket{k}.
$$
Therefore, matrix of the DFT of size $2^L$ is $Q_L^{\dagger}$. The DFT, in turn, is a counterpart of the continuous Fourier transform. The Fourier transform of a measurable function $f \in L_2(\mathbb{R})$ is given by the following equation
\begin{equation}\label{eq:fourier-transform-cont}
\hat{f}(\xi) = \frac{1}{\sqrt{2\pi}}\int\limits_{-\infty}^{\infty} f(x) e^{-\i 2\pi x \xi } dx.
\end{equation}
We refer to $f$ as a function in the spatial domain, and $f(x-x_0)$ is a shift by $x_0$ in the spatial domain by $x_0$. The functions with the circumflex, for example $\hat{f}$, are considered in the (ordinary) frequency domain, and $\hat{f}(\xi-\xi_0)$ is a shift by $\xi_0$ in the frequency domain.

When we consider $f:[0, 1] \mapsto \mathbb{R}$ and uniform discretization in the interval $[0, 1]$, $x_k = k/N$ for $k=0, \ldots, N-1$ where $N=2^L$ for some integer $L \ge 1$. We use brackets to denote the discretization of a function and let $f[k] = f(k/N)$. In the (ordinary) frequency domain, we emphasize that $k\in \mathbb{Z}$ by the square brackets and $\hat{f}[k] = \hat{f}(k)$. Then in the discrete case,
$$
(\hat{f}[0], \ldots, \hat{f}[N-1])^T = Q^{\dagger}_L (f[0], \ldots, f[N-1])^T.
$$
Notice that we use a different range $\{-N/2, \ldots, N/2-1\}$ in the discretization of $\hat{f}$, however, $\hat{f}[-N/2+k] = \hat{f}[N/2 + k]$. 

%% file: Preliminaries/02-Wave-Packets.tex
\subsection{Wave packets}\label{section:wave-packets-background}
In the paper, we design algorithms that implement wave packet transforms for wave packets with compact frequency support. Moreover, we define wave packets exclusively in the frequency domain, although the wave packets in the spatial domain can be obtained via inverse Fourier transform. We index wave packets $\varphi^j_{m,n}$ by three indices: their scale $j \in \mathbb{Z}_+ \setminus \{0\}$, frequency $m \in \mathbb{Z_+}$, and position in space $n \in \mathbb{Z}$. Here, $\mathbb{Z}_+$ denotes nonnegative integer numbers. 

\begin{figure}
    \centering
    \includegraphics[scale=1.1]{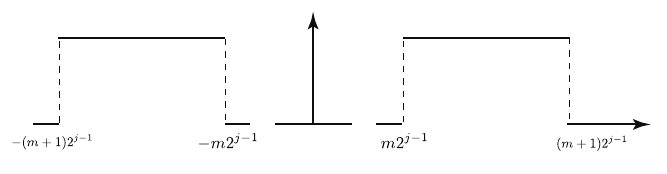}
    \caption{Shannon wavelet $\hat{\varphi}^j_{m}$: one bump per each side of the axis in the frequency domain.}
    \label{fig:Shannon-wave-packet}
\end{figure}

First, we consider wave packets with sharp frequency windows such that generating functions are based on perfect filters:
\begin{equation}
\label{eq:Shannon-definition}
\begin{gathered}
\hat{\varphi}^0_{m}(\xi) = I\left\{2\xi \in [-m-1, -m) \cup [m, m+1)\right\},\\
\hat{\varphi}^j_{m}(\xi) = 2^{-j/2} \hat{\varphi}^0_{m}(2^{-j}\xi),\\
\hat{\varphi}^j_{m, n}(\xi) = e^{-\i 2\pi n 2^{-j} \cdot \xi} \hat{\varphi}^j_m(\xi).
\end{gathered}
\end{equation}
(See Figure \ref{fig:Shannon-wave-packet}.) The support of $\hat{\varphi}^j_{m,n}$ is centered in frequency around $\pm \xi_{j,m} = \pm 2^{j-1} m$ and in space around $x_{j, n} = 2^{-j}n$. The resulting wave packets are \textit{Shannon wavelets}. The classical indexation arising from multipass filters differs from the above and requires a simple permutation (see Chapter 8 of \cite{WaveletTourofSignalProcessing} for details). 

The \textit{1D wave atoms} introduced in \cite{DEMANET2007368} are based on non-standard wavelets of \cite{Villemoes2002WaveletPW}, and are similarly centered in frequency around $\pm \xi_{j,m} = \pm 2^{j-1} m$ and centered in space around $x_{j,n} = 2^{-j}n$:
\begin{equation}\label{def:wave-atoms}
\begin{gathered}
\hat{\psi}_m^j(\xi) = 2^{-j/2} \hat{\psi}_m^0(2^{-j} \xi),\\
\hat{\psi}_{m,n}^j(\xi) = e^{-\i 2\pi n2^{-j} \cdot \xi} \hat{\psi}_m^j(\xi).
\end{gathered}
\end{equation}
Its mother wave packets demonstrate an elegant way to deal with the problem of frequency localization \cite{ Villemoes2002WaveletPW, DEMANET2007368} and involve symmetric pairs of compactly supported bumps in the frequency domain:
\begin{equation}\label{def:psi-hat}
\hat{\psi}_m^0(\xi) = e^{-\i \pi \xi}\big[e^{\i \alpha_m} g((-1)^m (2\pi \xi - 2\alpha_m)) + e^{-\i\alpha_m} g((-1)^{m+1} (2\pi \xi + 2\alpha_m))\big],
\end{equation}
where $\alpha_m = \tfrac{\pi}2(m+\tfrac12)$ and $g$ is a continuous real-valued function compactly supported on an open interval $(-\tfrac76\pi, \tfrac56\pi)$. Furthermore, $g$ must satisfy the following requirements: for any $w \in [-\tfrac{\pi}3, \tfrac{\pi}3]$,
\begin{align}  
& g^2(\tfrac{\pi}2 - w) + g^2(\tfrac{\pi}2 + w) = 1, \label{prop:sum_of_squares}\\
& g(-2w-\tfrac{\pi}2) = g(\tfrac{\pi}2+w) \label{prop:change_of_sign}.
\end{align}

\begin{figure}
    \centering
    \includegraphics[scale=0.75]{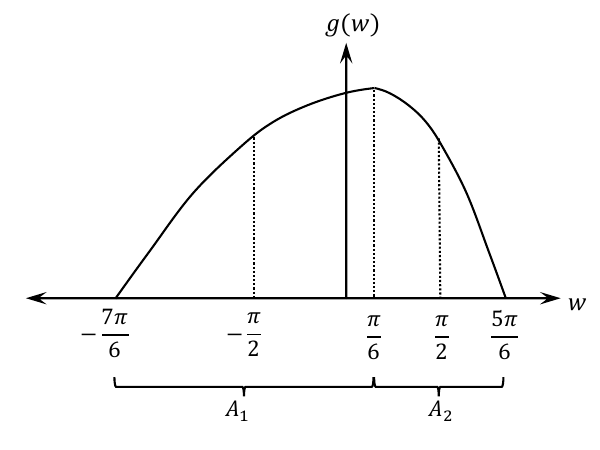}
    \caption{The function $g$, satisfying equations (\ref{prop:sum_of_squares}) and (\ref{prop:change_of_sign})}
    \label{fig:wave-atom-g}
\end{figure}

The above construction can be understood qualitatively in the following way. First, one can write the support of $g$ as the union of two intervals:
\begin{equation}\label{eqn-A1-A2}
A_1 = [-\tfrac{7}{6}\pi, \tfrac{1}{6}\pi],\qquad
A_2 = [\tfrac{1}{6}\pi, \tfrac{5}{6}\pi].
\end{equation}
(See Figure \ref{fig:wave-atom-g}.) Typically $g$ is chosen to be monotonically increasing on $A_1$, and monotonically decreasing on $A_2$. In Eq.~(\ref{prop:sum_of_squares}), $g(\tfrac{\pi}{2} - w)$ and $g(\tfrac{\pi}{2} + w)$ can be understood as mirror images of each other, reflected around the midpoint of the interval $A_2$. In the wave atom transform, these mirror images are associated with neighboring wavepackets, and Eq.~(\ref{prop:sum_of_squares}) says that their squares sum to 1. Eq.~(\ref{prop:change_of_sign}) says that the increasing part of $g$ (on the interval $A_1$, centered at $-\pi/2$) is a mirror image of the decreasing part of $g$ (on the interval $A_2$, centered at $\pi/2$), stretched by a factor of 2. This asymmetric property \eqref{prop:change_of_sign} of profile $g$ is important to preserve the orthogonality of the $\hat{\psi}^j_{m,n}$. In contrast to Shannon wavelets, wave atoms have overlapping support in the frequency domain.

\begin{figure}
    \centering
    \includegraphics[width=\linewidth]{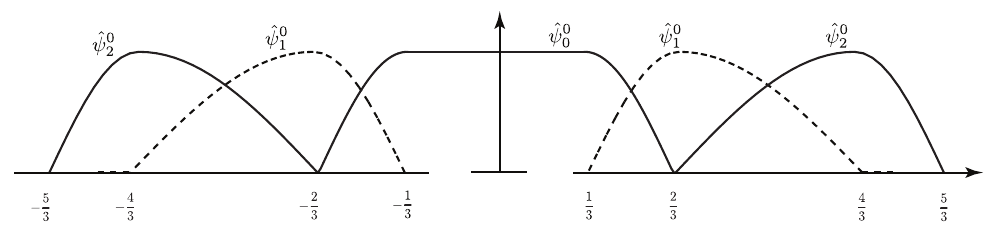}
    \caption{Wave atoms. $\hat{\psi}^0_0$ is unique, $\hat{\psi}^0_1$ is an example of odd $m$'s, $\hat{\psi}^0_2$ represents even $m$'s. }
    \label{fig:Wave-atoms-image}
\end{figure}

One can describe the wavepackets $\hat{\psi}_m^0(\xi)$ qualitatively as follows. (See Figure \ref{fig:Wave-atoms-image}.) In Eq.~(\ref{def:psi-hat}), $\hat{\psi}_m^0(\xi)$ consists of a complex phase factor of $e^{-\i \pi \xi}$, times a linear combination of two bumps, which are supported near $\pm \alpha_m/\pi = \pm \tfrac{1}{2} (m+\frac{1}{2})$. These two bumps are mirror images of each other, with a relative phase difference of $e^{2\i\alpha_m} = \i(-1)^m$. Furthermore, when comparing wavepackets $\hat{\psi}_m^0(\xi)$ and $\hat{\psi}_{m+1}^0(\xi)$ that are adjacent to each other in the frequency domain, one sees that the bumps corresponding to $m$ overlap with, and are mirror images of, the bumps corresponding to $m+1$. In other words, the shapes of the wavepackets $\hat{\psi}_m^0(\xi)$ change according to the parity of $m$.

Because of the above properties, these wavepackets form an orthonormal basis so that for the square-integrable function $f \in L^2(\mathbb{R})$, 
\begin{equation}
f = \sum_{j,m} \sum_n \langle f, \psi^j_{m, n}\rangle \psi^j_{m, n} = \sum_{j,m} \sum_n c_{j, m, n} \psi^j_{m, n}.
\end{equation}
For all $(j, m)$, the coefficients $c_{j, m, n}$ can be interpreted as a convolution at scale $2^{-j}$ and are defined by Plancherel's theorem as follows:
\begin{equation}
\begin{gathered}    
c_{j,m,n} = \langle f, \psi^j_{m, n}\rangle = \int \psi_{m, n}^j(x) f(x) dx \\ = \int \psi_{m}^j(x - 2^{-j}n) f(x) dx = \frac{1}{2\pi} \int e^{\i 2\pi 2^{-j} n \xi} \overline{\hat{\psi}^j_m(\xi)} \hat{f}(\xi) d\xi.
\end{gathered}
\end{equation}

{\subsubsection{Approximation and Discretization}\label{sec-approx-discret}}

{Since we are interested in numerical computations (both classical and quantum), the question naturally arises: is there a way to implement these wavelet transforms approximately, while preserving the orthogonality of the wavepackets? The answer is yes. Intuitively, the reason is that the orthogonality only depends on the overlaps between \textit{pairs} of neighboring wavepackets $\hat{\psi}^j_m$ and $\hat{\psi}^{j'}_{m'}$, i.e., it does not depend on relationships involving more than two wavepackets at a time. Even if both of the wavepackets $\hat{\psi}^j_m$ and $\hat{\psi}^{j'}_{m'}$ are computed approximately on the region where they overlap, one ensure that the approximation errors ``cancel out,'' so that the wavepackets are still orthogonal.}

{The above argument can be made precise, by noticing that the orthogonality of the wavepackets follows from Eqs.~(\ref{prop:sum_of_squares}) and (\ref{prop:change_of_sign}), involving the function $g$. We can ensure that these equations hold exactly, even if $g$ is only computed approximately, in the following way. Recall that the support of $g$ consists of the set $A_1 \cup A_2$ defined in Eq.~(\ref{eqn-A1-A2}). We can split this set into four subsets: the left and right halves of the interval $A_1$, and the left and right halves of the interval $A_2$. Then we can compute the function $g$ approximately on one of these subsets, and then use Eqs.~(\ref{prop:sum_of_squares}) and (\ref{prop:change_of_sign}) to determine the values of $g$ on the rest of the set $A_1 \cup A_2$. In this way, even though $g$ is only computed approximately, Eqs.~(\ref{prop:sum_of_squares}) and (\ref{prop:change_of_sign}) are still satisfied with high precision, hence the wavepackets are still orthogonal.}

{Finally, we introduce the discrete version of the wavelet transform. For simplicity, we make some fairly strong assumptions about the function $f$, in order to motivate the definition of the discrete transform. We consider $f: [0, 1) \rightarrow \mathbb{R}$, and we extend its domain to all of $\mathbb{R}$ while imposing periodic boundary conditions, i.e., $f(x+1) = f(x)$. Then $f$ can be represented by a Fourier series. We assume that $f$ is band-limited, so it can be recovered from the first $N$ terms in its Fourier series, where $N=2^L$ for some integer $L \ge 1$. By the Shannon sampling theorem, $f$ can also be recovered from its uniform discretization on the interval, $f[k] = f(k/N)$ for $k=0, \ldots, N-1$. In this setting, it is natural to define a} discrete version of the wavelet transform:
\begin{equation}\label{ksum1}
c^D_{j,m,n} = \sum_{k=-N/2, \ldots, N/2-1} \overline{\hat{\psi}_{m,n}^j(k)} \hat{f}[k].
\end{equation}
Therefore, the operation of our interest is
\begin{equation}\label{eq:operator-of-interest}
(f[0], \ldots, f[N-1]) \mapsto \{c^D_{j, m, n}\}_{(j, m, n) \in \Gamma},
\end{equation}
where $\{\hat{\psi}^j_{m,n}\}_{(j,m,n)\in \Gamma}$ organize an orthogonal basis for some set of index triplets $\Gamma$. We discuss the selection of orthogonal bases in the next section.

%% file: Preliminaries/02-Wave-Packet-Trees.tex
\subsection{Wavelet packet trees}\label{section: Wave packet trees}
\begin{figure}
    \centering
    \resizebox{0.32\textwidth}{!}{\input{Preliminaries/Figures/uniform-tree}}
    \resizebox{0.27\textwidth}{!}{\input{Preliminaries/Figures/dyadic-tree}}
    \caption{Examples of trees that generate wavelet packet bases. Left: a perfect binary tree, which generates a uniform wavelet packet basis. Right: a dyadic tree, which generates a dyadic wavelet packet basis.}
    \label{fig:Simple-tree-example}
\end{figure}

\begin{figure}
    \centering
\begin{tikzpicture}
    \usetikzlibrary{calc};
    \coordinate (root) at (6,6); 
    \coordinate (l) at ($ (root) + (-3,-0.5) $);
    \coordinate (r) at ($ (root) + (3,-0.5) $);
    \filldraw[black] (root) circle (2pt);
    \draw (l) -- (root) -- (r);
    \coordinate (ll) at ($ (l) + (-1.5,-0.5) $); 
    \coordinate (lr) at ($ (l) + (1.5,-0.5) $);
    \filldraw[black] (l) circle (2pt);
    \draw (ll) -- (l) -- (lr);
    \coordinate (rl) at ($ (r) + (-1.5,-0.5) $);
    \coordinate (rr) at ($ (r) + (1.5,-0.5) $);
    \filldraw[black] (r) circle (2pt);
    \draw (rl) -- (r) -- (rr);
    \coordinate (lll) at ($ (ll) + (-0.75,-0.5) $); 
    \coordinate (llr) at ($ (ll) + (0.75,-0.5) $);
    \filldraw[black] (ll) circle (2pt);
    \draw (lll) -- (ll) -- (llr);
    \coordinate (lrl) at ($ (lr) + (-0.75,-0.5) $);
    \coordinate (lrr) at ($ (lr) + (0.75,-0.5) $);
    \filldraw[black] (lr) circle (2pt);
    \draw (lrl) -- (lr) -- (lrr);
    \coordinate (rll) at ($ (rl) + (-0.75,-0.5) $);
    \coordinate (rlr) at ($ (rl) + (0.75,-0.5) $);
    \filldraw[black] (rl) circle (2pt);
    \draw (rll) -- (rl) -- (rlr);
    \coordinate (rrl) at ($ (rr) + (-0.75,-0.5) $);
    \coordinate (rrr) at ($ (rr) + (0.75,-0.5) $);
    \filldraw[black] (rr) circle (2pt);
    \draw (rrl) -- (rr) -- (rrr);
    \coordinate (llll) at ($ (lll) + (-0.4,-0.5) $); 
    \coordinate (lllr) at ($ (lll) + (0.4,-0.5) $);
    \filldraw[black] (lll) circle (2pt);
    \draw (llll) -- (lll) -- (lllr);
    \coordinate (llrl) at ($ (llr) + (-0.4,-0.5) $);
    \coordinate (llrr) at ($ (llr) + (0.4,-0.5) $);
    \filldraw[black] (llr) circle (2pt);
    \draw (llrl) -- (llr) -- (llrr);
    \coordinate (lrll) at ($ (lrl) + (-0.4,-0.5) $);
    \coordinate (lrlr) at ($ (lrl) + (0.4,-0.5) $);
    \filldraw[black] (lrl) circle (2pt);
    \draw (lrll) -- (lrl) -- (lrlr);
    \coordinate (lrrl) at ($ (lrr) + (-0.4,-0.5) $);
    \coordinate (lrrr) at ($ (lrr) + (0.4,-0.5) $);
    \filldraw[black] (lrr) circle (2pt);
    \draw (lrrl) -- (lrr) -- (lrrr);
    \coordinate (rlll) at ($ (rll) + (-0.4,-0.5) $);
    \coordinate (rllr) at ($ (rll) + (0.4,-0.5) $);
    \filldraw[black] (rll) circle (2pt);
    \draw (rlll) -- (rll) -- (rllr);
    \filldraw[black] (rlr) circle (2pt);
    \filldraw[black] (rrl) circle (2pt);
    \filldraw[black] (rrr) circle (2pt);
    \coordinate (lllll) at ($ (llll) + (-0.2,-0.5) $); 
    \coordinate (llllr) at ($ (llll) + (0.2,-0.5) $);
    \filldraw[black] (llll) circle (2pt);
    \draw (lllll) -- (llll) -- (llllr);
    \coordinate (lllrl) at ($ (lllr) + (-0.2,-0.5) $);
    \coordinate (lllrr) at ($ (lllr) + (0.2,-0.5) $);
    \filldraw[black] (lllr) circle (2pt);
    \draw (lllrl) -- (lllr) -- (lllrr);
    \coordinate (llrll) at ($ (llrl) + (-0.2,-0.5) $);
    \coordinate (llrlr) at ($ (llrl) + (0.2,-0.5) $);
    \filldraw[black] (llrl) circle (2pt);
    \draw (llrll) -- (llrl) -- (llrlr);
    \filldraw[black] (llrr) circle (2pt);
    \filldraw[black] (lrll) circle (2pt);
    \filldraw[black] (lrlr) circle (2pt);
    \filldraw[black] (lrrl) circle (2pt);
    \filldraw[black] (lrrr) circle (2pt);
    \filldraw[black] (rlll) circle (2pt);
    \filldraw[black] (rllr) circle (2pt);
    \coordinate (llllll) at ($ (lllll) + (-0.1,-0.5) $); 
    \coordinate (lllllr) at ($ (lllll) + (0.1,-0.5) $);
    \filldraw[black] (lllll) circle (2pt);
    \draw (llllll) -- (lllll) -- (lllllr);
    \filldraw[black] (llllr) circle (2pt);
    \filldraw[black] (lllrl) circle (2pt);
    \filldraw[black] (lllrr) circle (2pt);
    \filldraw[black] (llrll) circle (2pt);
    \filldraw[black] (llrlr) circle (2pt);
    \filldraw[black] (llllll) circle (2pt); 
    \filldraw[black] (lllllr) circle (2pt);
\end{tikzpicture}
    \caption{{Example of a tree that generates a wave atom basis, due to \cite{DEMANET2007368}. The branching pattern of the tree is chosen so that the wave atom basis functions obey the parabolic scaling law in Eq.~(\ref{eqn-parabolic}). To simplify the picture, the labels on the nodes of the tree are not shown.}}
    \label{fig-wave-atom-tree-example}
\end{figure}

Coifman, Meyer, and Wickerhauser introduced wavelet packets by linking multiresolution approximations (a sequence of embedded vector spaces $(\mathrm{V}_j)_{j \in \mathbb{Z}}$ for approximating $L_2(\mathbb{R})$ functions) with wavelets \cite{Coifman}. They decompose a multiresolution approximation space $\mathrm{V}_j$ into a lower-resolution space $\mathrm{V}_{j+1}$ and a detail space $\mathrm{W}_{j+1}$ by splitting the orthogonal basis into two new orthogonal bases. 
Intuitively, when wavelets are used for image processing, we may lose some details when reducing the resolution.
\cite{Coifman} also show that, instead of just dividing the approximation spaces to create the so-called detail spaces and wavelet bases, we can split the detail spaces to generate new bases. This process of dividing vector spaces can be represented in a structure of a binary tree (see Chapter 8 of \cite{WaveletTourofSignalProcessing}). Any node in this binary tree can be labeled by $(j, m)$, where  $j$ is the level of the tree, and $m$ is the number of nodes on its left at the same level. To each node  we associate a space $\mathrm{W}^j_m$, which admits an orthogonal basis $\{\varphi^j_{m,n}\}_{n \in \mathbb{Z}}$.

We adopt the notion of wavelet packet trees and associate $\mathrm{W}^j_m$ with an orthogonal basis $\{ \hat{\psi}_{m, n}^j \}_{n \in \mathbb{Z}}$ which in the discrete case becomes $\{ \hat{\psi}_{m, n}^j \}_{0 \le n < 2^j}$ for wave atoms or $\{ \hat{\varphi}_{m, n}^j \}_{0\le n < 2^j}$ for the Shannon wavelet. Then, all $\mathrm{W}^j_m$ can be organized into a binary wavelet packet tree. We will abuse the notation and refer to $\mathrm{W}^j_m$ as both the space spanned by the basis $\{\hat{\psi}^j_{m, n}\}_{0 \le n < 2^j}$ (or $\{\hat{\varphi}^j_{m, n}\}_{0 \le n < 2^j}$) and the node in the tree.
\begin{definition}[Admissible wavelet packet binary tree, adapted \cite{WaveletTourofSignalProcessing}] \label{def:wave-packet-admissible-tree}
A binary tree of height $L-1$ with nodes $\{\mathrm{W}^j_m\}$ is \textit{admissible} if
\begin{enumerate}
\item The root of the tree is $\mathrm{W}^L_0$,
\item Every non-leaf node $\mathrm{W}^j_m$ has two children $\mathrm{W}^{j-1}_{2m}$ and $\mathrm{W}^{j-1}_{2m+1}$, and $j \ge 2$.
\end{enumerate}
\end{definition}

Using Definition \ref{def:wave-packet-admissible-tree}, one can construct wavelet bases $\{\varphi^j_{m,n} \;|\; (j,m,n)\in\Gamma_T \}$, where $\Gamma_T$ is the set of index triplets generated by the wave packet admissible tree $T$: 
 \begin{equation}
 \Gamma_T = \{ (j, m, n) \,|\, \mathrm{W}^j_m \in \Lambda_T \text{ and } 0 \le n < 2^j\},
 \end{equation}
\begin{equation}
\Lambda_T = \{\mathrm{W}^j_m\,|\, \mathrm{W}^j_m \text{ is a leaf-node in } T\}.
\end{equation}
Two of the more popular bases generated by wavelet packet trees  \cite{WaveletTourofSignalProcessing} are the following:
\begin{itemize}
    \item \textit{Uniform  $\{\varphi^j_{m,n}\}$,} where $j$'s are equal to some fixed value (call it $j_0$), and $\Gamma = \big\{(j, m, n) \, |\, 0 \le m < 2^{L-j},\, 0 \le n < 2^j,\, j = j_0 \big\}$. This type of basis is generated by perfect binary trees.
    \item \textit{Dyadic  $\{\varphi^j_{m,n}\}$}. In this case, most of the basis function has $m=1$, and $\Gamma = \{ (1, 0, 0), (1, 0, 1) \} \cup \{ (j, 1, n) \,|\, 1 \le j < L,\, 0 \le n < 2^j\}$. This type of basis is generated by dyadic trees (such that every right node is a leaf while every left node is nonleaf unless the node is on the minimal level).
   
\end{itemize}

The examples of wavelet packet trees generating orthogonal bases discussed above can be found in Figure \ref{fig:Simple-tree-example}. One can verify that both approaches render an orthogonal basis for Shannon wavelets and wave atoms. 

In addition, we will consider another class of trees that generate orthogonal bases for the wave atom transform, and have a special property, ``parabolic scaling,'' defined in Eq.~(\ref{eqn-parabolic}), which is desirable for many applications of the wave atom transform \cite{DEMANET2007368}. {We briefly describe how Eq.~(\ref{eqn-parabolic}) constrains the structure of the wavepacket tree. Recall that Eq.~(\ref{eqn-parabolic}) says that, for each wavelet basis function, wavelength $\propto$ (diameter)$^2$. From the discussion in Section \ref{section:wave-packets-background}, the basis function $\psi^j_m$ has wavelength $\propto (1/m) 2^{-j}$ and diameter $\propto 2^{-j}$. Hence Eq.~(\ref{eqn-parabolic}) implies that:
\begin{equation}\label{eqn-parabolic-tree}
m \propto 2^j.
\end{equation}
This condition must be satisfied at every leaf of the wavepacket tree. Hence, the tree resembles a complete binary tree whose lower-right branches have been cut away at the points specified by Eq.~(\ref{eqn-parabolic-tree}). An example of such a tree is shown in Fig.~\ref{fig-wave-atom-tree-example}.}

{Finally, we define a class of ``monotonic'' trees that is easier to work with, and contains the ``parabolic scaling'' trees described above:
\begin{definition}\label{def-monotonic}
An admissible wave packet tree $T$ is called monotonic if, when the tree's leaf nodes $W^j_m$ are visited from left to right, the sequence of $j$ values is nondecreasing.
\end{definition}
\begin{proposition}
If $T$ is an admissible wave packet tree that satisfies the parabolic scaling relation (\ref{eqn-parabolic-tree}), then $T$ is monotonic.
\end{proposition}
\begin{prf}
Suppose $T$ is not monotonic. Then there are leaf nodes $W^j_m$ and $W^{j'}_{m'}$ such that $W^j_m$ appears to the left of $W^{j'}_{m'}$ and $j>j'$. Since both of these leaf nodes satisfy Eq.~(\ref{eqn-parabolic-tree}), we must have $m>m'$. Since $T$ is admissible, we have $m\leq m'$, leading to a contradiction.
\end{prf}
We will develop algorithms for monotonic trees (which will apply to parabolic scaling trees, as a special case) in Sections \ref{section: wave packet encoding} and \ref{section:wave atom admissible tree}.}



We would also like to state a few properties of admissible wave packet trees that will be helpful in the later discussion.
\begin{proposition}[enumeration properties of admissible wave packet
tree]\label{prop:admissible-properties}
For wave packet admissible tree $T$ (by definition \ref{def:wave-packet-admissible-tree}) with a root $\mathrm{W}^L_0$ and $K$ leaf nodes $\{\mathrm{W}^{j_i}_{m_i} \;|\; i=1,\ldots,K\}$ traversed from left to right, 
\begin{enumerate}
\item $m_1 = 0$,
\item For $i=1,\ldots, K-1$, $(m_i+1)2^{j_i} = m_{i+1}2^{j_{i+1}}$,
\item $(m_K+1)2^{j_K} = 2^L$.
\end{enumerate}
\end{proposition}

\input{Preliminaries/APX-Wave-packet-trees}



%% file: Preliminaries/Figures/uniform-tree.tex
\begin{tikzpicture}[scale=7,font=\footnotesize]
\tikzstyle{level 1}=[level distance=1.35mm,sibling distance=3mm]
\tikzstyle{level 2}=[level distance=1.35mm,sibling distance=1.5mm]

\node(00){$\mathrm{W}^3_0$}{
child{node(10){$\mathrm{W}^{2}_0$}
    child{node(20){$\mathrm{W}^{1}_0$}}
    child{node(21){$\mathrm{W}^{1}_1$}}
}
child{node(11){$\mathrm{W}^{2}_1$}
    child{node(22){$\mathrm{W}^{1}_2$}}
    child{node(23){$\mathrm{W}^{1}_3$}}
}
};
\end{tikzpicture}

%% file: Preliminaries/Figures/dyadic-tree.tex
\begin{tikzpicture}[scale=7,font=\footnotesize]
\tikzstyle{level 1}=[level distance=1mm,sibling distance=2mm]
\tikzstyle{level 2}=[level distance=1mm,sibling distance=2mm]
\tikzstyle{level 3}=[level distance=1mm,sibling distance=2mm]

\node(30){$\mathrm{W}^4_0$}{
child{node(20){$\mathrm{W}^3_0$}{
child{node(10){$\mathrm{W}^{2}_0$}
    child{node(00){$\mathrm{W}^{1}_0$}}
    child{node(01){$\mathrm{W}^{1}_1$}}
}
child{node(11){$\mathrm{W}^{2}_1$}
}}
}
child{node(21){$\mathrm{W}^3_1$}}
};
\end{tikzpicture}

%% file: Preliminaries/APX-Wave-packet-trees.tex
\begin{prf}
We traverse the tree as follows:
\begin{enumerate}
\item We start from the root node $W^L_0$ and descend the tree while always moving down to a left child, so at each level $W^{j}_0 \mapsto W^{j-1}_{0}$. Suppose that we stopped at node $W^{j'}_0$, then $W^{j'}_0$ is a leaf node; otherwise, it should have left and right children. Therefore, the first node in the traverse is $W^{j'}_0$, i.e. $m_1=0$.
\item Imagine that we will add nodes to the tree $T$ until it becomes a perfect binary tree of height of $L$, that is, for each leaf node $W^j_m$ we append a perfect binary subtree of height $j$ at that node. Then the leftmost leaf node in the subtree with root $W^{j_i}_{m_i}$ is $W^{0}_{m_i2^{j_i}}$, and the rightmost leaf node is $\mathrm{W}^{0}_{m'}$ where $m' = (m_i+1)2^{j_i}-1$. The leaf node next to $\mathrm{W}^0_{m'}$ is the leftmost leaf node in the subtree with root $\mathrm{W}^{j_{i+1}}_{m_{i+1}}$, which is $\mathrm{W}^{0}_{m_{i+1}2^{j_{i+1}}}$. Therefore, $(m_i+1)2^{j_i} = m_{i+1}2^{j_{i+1}}$.
\item Observe that the rightmost node in the expanded perfect binary tree (constructed in 2) is $\mathrm{W}^0_{2^L-1}$ and is also the rightmost node of the subtree with root $\mathrm{W}^{j_K}_{m_K}$. Therefore, $(m_K+1)2^{j_K} = 2^L$.
\end{enumerate}
\end{prf}

%% file: Encoding/03-Encoding.tex
\ifthenelse{\boolean{originalformat}}{
 \section{Encoding}
}{
 \subsection{Encoding}
}
\label{section: Encoding}
The transform of interest in equation \eqref{eq:operator-of-interest} does not define an order of the coefficients $c_{j, m, n}^D$. In other words, $\hat{\varphi}^j_{m, n}$ and $\hat{\psi}^j_{m, n}$ need to be indexed. Moreover, in equation \eqref{ksum1} the sum goes over negative and positive integer frequencies. We would like to encode both the wave packet indices $(j, m, n)$ and the (integer) frequencies $k$ such that each is represented by $L$ bits.

\input{Encoding/03-Frequency-encoding}
\input{Encoding/03-Wave-packet-encoding}

%% file: Encoding/03-Frequency-encoding.tex
\subsection{Frequency encoding}\label{section:frequency-encoding}
This section introduces frequency encoding and decoding procedures, where encoded frequencies are represented by $L$ bits. It also covers the concept of retaining decoding, allowing a subset of bits to be decoded without affecting others, and involves constructing corresponding quantum circuits.

 Given that both Shannon wavelets and wave atoms are symmetric around zero, we would like to map the frequencies of the same absolute value with different sign in such a way that they would have consecutive indices when encoded. Moreover, it is natural to expect that the nonnegative frequency $k$ would have a smaller encoded index than the frequency $k+1$.

In order to accomplish this, we introduce a \textit{frequency encoding procedure} as a bijection $e:\: \{-N/2, \ldots, N/2-1\} \rightarrow \{0, \ldots, N-1\}$  given by 
\begin{equation} \label{eq:encoding}
    e(k) = 
    \begin{cases}
        2k, \ k \geq 0, \\
        2|k| - 1, \ k <0.
    \end{cases}
\end{equation}
We call the reverse procedure of \eqref{eq:encoding} \textit{frequency decoding procedure} and define it by the inverse mapping $d:\: \{0, \ldots, N-1\} \rightarrow \{-N/2, \ldots, N/2-1\}$ given by
\begin{equation}\label{eq:decoding}
    d(i) = \begin{cases}
        \tfrac12 i, \ i \text{ even }, \\
       - \frac12 (i + 1), \ i \text{ odd }.
    \end{cases}
\end{equation}

We also introduce a \textit{retaining decoding procedure} $\tilde{d}(\cdot, j, m):~\{0,\ldots, 2^j-1\} \rightarrow \{0, \ldots, 2^j-1\}$ such that for $c \in \{0, \ldots, 2^j-1\}$
\begin{equation}\label{eq:retainingdecoding}
\tilde{d}(c, j, m) = d(m2^j+c) \bmod 2^j.
\end{equation} 
Note that this is function $\tilde{d}$ is \textit{not} the inverse of the encoding function introduced above. Rather, this function will be used in our implementations of quantum wavelet transforms, later in this article. The word ``retaining'' refers to the idea of staying in range of size $2^j$ (which will be fully explored in section \ref{section:quantum-shannon-wavelet-transform}) and the argument $c$ corresponds to the indices within the range of desired size. We can compute $\tilde{d}$ as follows
\begin{equation}
\label{local-decoding}
    \tilde{d}(c, j, m) = \begin{cases}
    2^{j-1} I (m \text{ is odd}) + \tfrac12 c, \text{ if } c \text{ is even},\\
    2^{j} - 2^{j-1} I (m \text{ is odd}) - \tfrac12 (c+1), \text{ if } c \text{ is odd}.
    \end{cases}
\end{equation}

\subsubsection{Frequency encoding circuit} \label{subsection:frequency encoding }
Both Shannon wavelet and wave atom transforms include DFT (inverse QFT) as the initial step to transform the spatial domain into the frequency domain. However, the result of DFT is defined on a grid $\{0, \ldots N-1\}$, but our transforms are defined on a grid $\{-N/2, \ldots, N/2-1\}$. Due to the periodicity of Fourier transform $\hat{f}[k] = \hat{f}[k-N]$, we have the initial stage as follows
$$
Q_L^{\dagger} \left(\smashoperator{\sum_{k=0}^{N-1}} f[k] \ket{k} \right) = \smashoperator{\sum_{k=0}^{N-1}} \hat{f}[k]\ket{k} = \smashoperator{\sum_{k=0}^{N/2-1}} \hat{f}[k]\ket{k} + \smashoperator{\sum_{k=-N/2}^{-1}} \hat{f}[k] \ket{k+N}.
$$
We transform $\ket{k}$ into $\ket{e(k)}$ when $k \ge 0$ and $\ket{k+N}$ into $\ket{e(k)}$ when $k < 0$ as described in \eqref{eq:encoding}. The further steps of the transforms are based on
\begin{equation}\label{eqn-encoding-decoding-identity}
\sum_{k=-N/2}^{N/2-1} \hat{f}[k] \ket{e(k)} = \sum_{k=0}^{N-1}\hat{f}[d(k)]\ket{k}.
\end{equation}

\begin{figure}[t]
    \centering
    \includegraphics[scale=0.75]{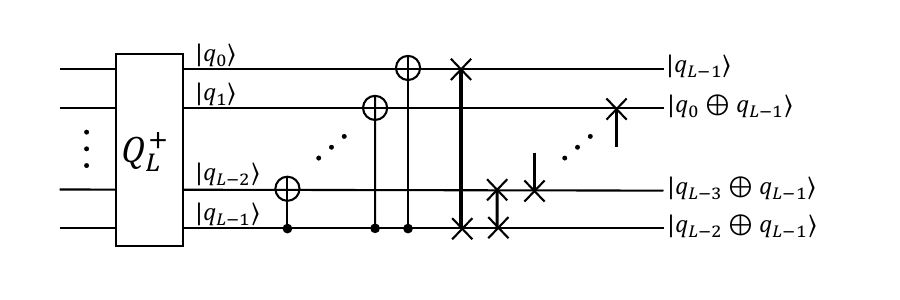}
    \caption{Frequency encoding circuit which implements function $e(\cdot)$ combined with the inverse QFT $Q_L^{\dagger}$}
    \label{fig:column-encoding}
\end{figure}

The two initial steps of the transforms, the inverse QFT and frequency encoding, are carried out by the circuit shown in Figure \ref{fig:column-encoding} (see \ifthenelse{\boolean{originalformat}}{Appendix \ref{section:decoding-circuit-proof}}{Supplementary Material} for details). The gate complexity of the algorithm is $O(L^2)$ due to the QFT step.

\subsubsection{Retaining decoding circuit}\label{section:retaining-decoding-circuit}

\begin{figure}
    \centering
    \includegraphics[scale=0.75]{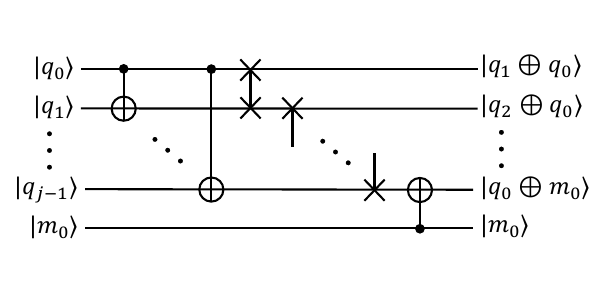}
    \caption{Retaining decoding circuit $\tilde{D}_j$ which implements function $\tilde{d}(\cdot, j, m)$}
    \label{fig:decoding}
\end{figure}

Recall the retaining decoding procedure defined by \eqref{eq:retainingdecoding}. Notice that $\tilde{d}(c, j, m) = \tilde{d}(c, j, m \bmod 2)$ meaning that the retaining decoding procedure depends on the parity of $m$ rather than the value.  Considering the parity of $m$ we may construct the quantum circuit of $\tilde{d}(\cdot, j, m)$ that operates on $j+1$ qubits: $j$ qubits in representation of $c$ and $1$ qubit of $m_0$ where $m_0$ is the least significant qubit of $m$. See Figure \ref{fig:decoding}.  (Additional details are presented in \ifthenelse{\boolean{originalformat}}{Appendix \ref{section:decoding-circuit-proof}}{the Supplementary Material}.) We refer to this circuit as $\tilde{D}_j$. The implementation of $\tilde{D}_j$ has gate complexity of $O(j)$.

In the special case of $j=L$, we do not pass the parity of $m$ to the algorithm and assume that $m_0=0$.



%% file: Encoding/03-Wave-packet-encoding.tex
\subsection{Wave packet encoding}\label{section: wave packet encoding}
The wavelet packet transform \eqref{eq:operator-of-interest} can be written as 
$$
\sum_{k=0}^{N-1}f[k]\ket{k} \mapsto \smashoperator{\sum_{(j, m, n)\in \Gamma_T}} c^D_{j, m, n}\ket{(j, m, n)},
$$
where $N = 2^L$, each $c^D_{j, m, n}  \approx \langle f, \psi^j_{m, n}\rangle$ and $\ket{(j,m,n)}$ is a representation of the triplet $(j, m, n)$. In this section, we discuss the encoding of such triplets using $L$ bits. To do this, we will introduce data structures for representing wave packet trees. 

Consider an admissible wave packet tree $T$ (in the sense of Definition \ref{def:wave-packet-admissible-tree}), described by the set of its leaf nodes $\Lambda_T = \{\mathrm{W}^j_m\}$. By Proposition \ref{prop:admissible-properties}, the mapping
\begin{equation}
(j, m, n) \mapsto p = m2^j + n \label{eq:row-encoding}
\end{equation}
is a bijection between the set of index triplets $\Gamma_T = \{(j, m, n)\,|\, \mathrm{W}^j_m \in \Lambda_T, 0 \le n < 2^j \}$ and the range of $\{0, \ldots, N-1\}$.
Observe that $j$ least significant bits of the encoded index $p$ represent $n$ in binary form, while the remaining bits represent $m$. Therefore, to decode the index, we need to know $j$. This requires information about the structure of the tree $T$.

\subsubsection{Encodings of wave packet trees}

We will construct algorithms for implementing wave packet transforms, given access to information about the wave packet tree $T$ through the following mechanism. Let us define a function $\tilde{h}: \{0, \ldots, N-1\} \mapsto \{1, \ldots, L - 1\}$ such that 
\begin{equation}\label{eq:tilde-h}
\forall (j,m,n) \in \Gamma_T,\quad  \tilde{h}(m2^j + n) = j.
\end{equation}
In other words, the function $\tilde{h}$ maps the encoded index $p$ back to $j$. 

Furthermore, we define a family of \textit{helper decoding boolean functions} $h_j: \{0, \ldots, N-1\} \mapsto \{0, 1\}$ for $1 \le j < L$ such that 
\begin{equation}\label{eq:helper-functions}
h_j(p) = 1 \text{ if and only if } \tilde{h}(p) \ge j.
\end{equation} 
Note that the function $\tilde{h}(\cdot)$ and the family of functions $h_j(\cdot)$ allows us to decode $p$'s. 

In this paper, we assume that there is an efficient implementation of functions $h_j$, specifically 
$$
\ket{p}\ket{b_1 \ldots b_{L-1}} \mapsto \ket{p}\ket{b_1 \oplus h_1(p), \ldots b_{L-1} \oplus h_{L-1}(p)}
$$
can be implemented using $O(L^2)$ gates. We demonstrate that this is possible for specific kinds of wave packet trees, including those described in Section \ref{section: Wave packet trees}:
\begin{itemize}
    \item \textit{Tree generating uniform basis.} All $j$'s of leaf nodes in the wave packet tree are equal, say $j=j'$, so $h_j(p)=I(j \ge j')$, with the implementation using one $X$-gate.
    \item \textit{Dyadic tree}. Observe that $h_{L-1}(p) = p_{L-1}$, $h_{L-2}(p) = h_{L-1}(p) \lor p_{L-2} = p_{L-1} \lor p_{L-2}$, in general,
    $$
    h_{i}(p) = \lor_{i'=i}^{L-1} p_{i'} = h_{i+1}(p) \lor p_i \text{ for } i = 2, \ldots L-2,
    $$ 
    and $h_1(p) = 1$, where $\lor$ is a logical OR operator. The procedure requires $O(L)$ gates.
    \item \textit{Monotonic tree}. 
    Denote by $m_j^*$, index $m$ of the leftmost leaf node on the level $j$. If no leaf nodes are on the level, $m_j^*$ is the index of the node which could be next to the most right node on level $j$, that is, $m_j^*-1$ is the index of the most right node. If there are no nodes on the level, let $m_j^* = 0$. Then
    $$
    h_j(p) = I(p \ge m_j^* 2^j).
    $$
    Given that $p$ is an $L$-bit value, the comparison requires $O(L)$ classical bit operations or quantum gates. The total number of gates for the procedure to compute all $h_j(p)$, $1 \le j < L$, is $O(L^2)$. 
\end{itemize}
Recall from Section \ref{section: Wave packet trees} that the wave atom transform requires 
a parabolic tree, which is covered by the case of monotonic trees. In the paper, we focus on monotonic trees and assume that all trees of interest are encoded by the indices of the leftmost leaf nodes at each level, $m_j^*$' s. This tree representation requires $O(L^2)$ bits and, therefore, does not increase the complexity of circuit construction.

%% file: Shannon/04-Shannon.tex
\ifthenelse{\boolean{originalformat}}{
 \section{Quantum Shannon wavelet transform}
}{
 \section{Results: Quantum Shannon wavelet transform}
}
\label{section:quantum-shannon-wavelet-transform}
In this section, we demonstrate how the retaining decoding function $\tilde{d}$ and the helper decoding functions $h_j$ can be used to construct an efficient implementation of the quantum Shannon wavelet transform for a monotonic tree (see Theorem \ref{theorem:algorithms-shannon}). 

Recall the definition of Shannon wavelets $\hat{\varphi}^j_m$ by equation \eqref{eq:Shannon-definition},
$$
\hat{\varphi}^j_m(\xi) = 2^{-j/2} I\left\{ 2^{-j+1} \xi \in [-m-1, -m) \cup [m, m+1) \right\}.
$$
(See also Figure \ref{fig:Shannon-wave-packet}.) It is straightforward to see that for any wave packet admissible tree $T$ generating a transform of size $N=2^L$, each frequency $k \in [-N/2, N/2)$ is in the support of one Shannon wavelet $\hat{\varphi}^j_m$ corresponding to a leaf node $\mathrm{W}^j_m \in \Lambda_T$. With the inclusion of shifts in the spatial domain,
\begin{equation}
\hat{\varphi}^j_{m, n}(\xi) = e^{-\i 2\pi n2^{-j} \cdot \xi} \hat{\varphi}^j_m(\xi),
\end{equation}
and following the definition of the discrete wave packets transform \eqref{ksum1}, we can formally define the quantum Shannon wavelet transform. 

\begin{definition}\label{def:shannon-wavelet-transform}
Consider an admissible wave packet tree $T$ with a root node $\mathrm{W}^L_0$. Define a matrix $C^S$ as
\begin{equation}\label{eq:shannon-transform-definition}
C^S_{m2^j+n, i} = \overline{\hat{\varphi}^j_{m, n}(d(i))},
\end{equation}
where $\mathrm{W}^j_m \in \Lambda_T$, $d$ is the decoding function in equation \eqref{eq:decoding}, and $n, i \in \mathbb{Z}$ such that $0 \le n < 2^j$ and $0 \le i < N$.
The following transform of size $N=2^L$, where $Q_L$ is QFT,
\begin{equation}
\begin{gathered}
\sum_{k=0}^{N-1} f[k]\ket{k} \mapsto \sum_{\mathrm{W}^j_m \in \Lambda_T}\sum_{n=0}^{2^j-1} \left(\sum_{i=0}^{N-1} C^S_{m2^j+n, i} \hat{f}[d(i)]\right) |m2^j+n\rangle
\\ \text{ where } \hat{f}[k] = \begin{cases} \sum_{i=0}^{N-1} (Q_L^\dagger)_{ki} f[i], \text{ if } k \ge 0,\\
\sum_{i=0}^{N-1} (Q_L^\dagger)_{(k+N)i} f[i], \text{ if } k < 0,
\end{cases}
\end{gathered}
\end{equation}
is called a \textit{quantum Shannon wavelet transform}.
\end{definition}


\subsection{Block structure}

The Shannon wavelet transform has a block structure, which allows it to be implemented efficiently by a quantum circuit, as follows.
When we consider a leaf node $\mathrm{W}^j_m \in \Lambda_T$ and $n=0,\ldots, 2^j-1$, $C^S_{m2^j+n, i}$ can be expressed as follows
\[
\begin{split}
C^S_{m2^j+n, i} &= \overline{\hat{\varphi}^j_{m,n}(d(i))} 
\\
&= 2^{-j/2} e^{\i 2\pi 2^{-j} n\cdot d(i)} I\left\{d(i) \in [-(m+1)2^{j-1}, -m2^{j-1}) \cup [m2^{j-1}, (m+1)2^{j-1})\right\} 
\\
&= 2^{-j/2} e^{\i 2\pi 2^{-j} n\cdot d(i)} I\left\{i \in [m2^j, (m+1)2^j)\right\}.
\end{split}
\]
Denote $c = i - m2^j$, then
$$
C^S_{m2^j+n, m2^j+c} = 2^{-j/2} e^{\i 2\pi 2^{-j} n\cdot d(m2^j+c)} I\left\{ 0 \le c < 2^j  \right\}.
$$
\begin{figure}
    \centering
    \includegraphics[width=0.5\linewidth]{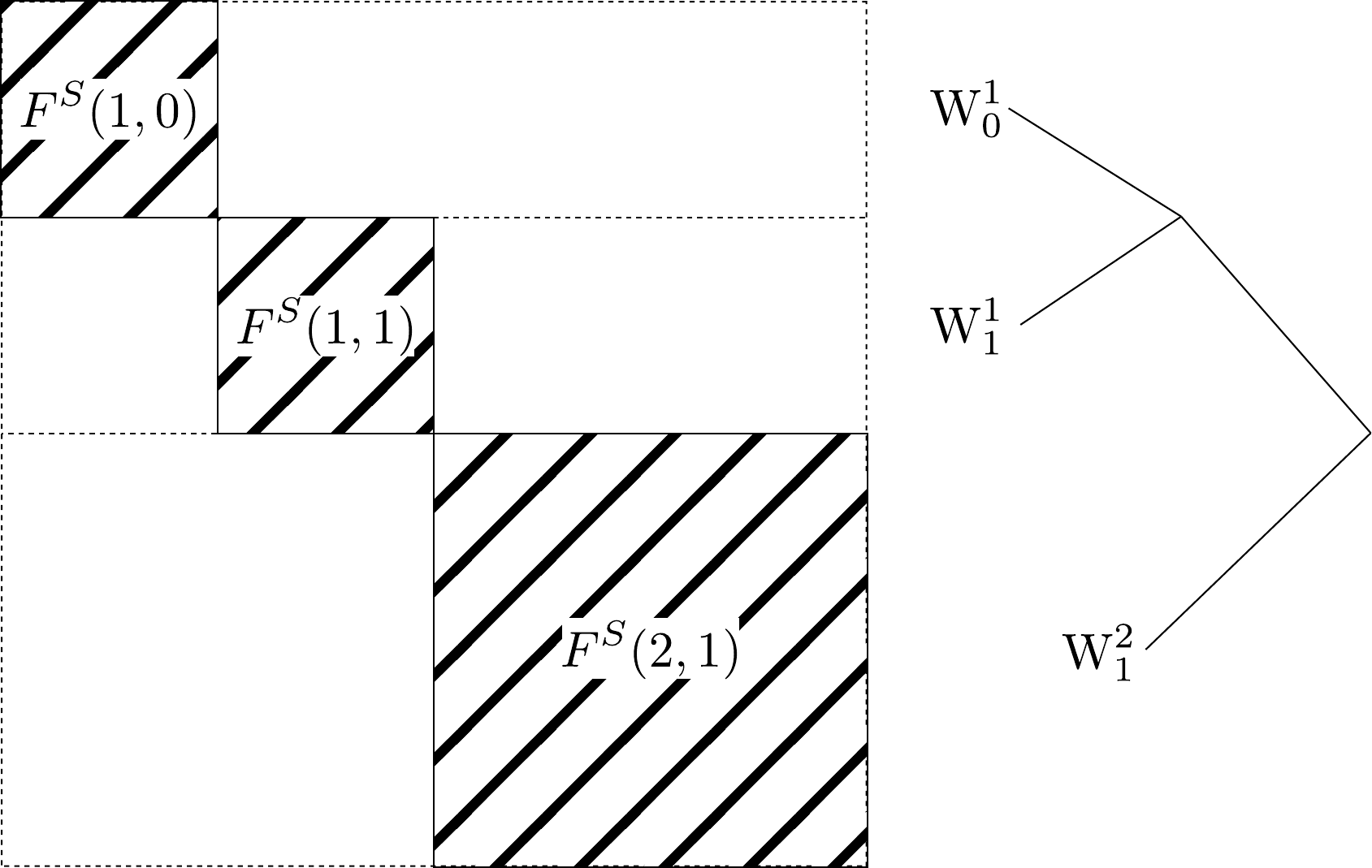}
    \caption{Block-diagonal structure of matrix $C^S$ with the blocks $F^{S}(1, 0), F^{S}(1, 1)$ and $F^{S}(2,1)$ corresponding to leaf nodes of $T = \{\mathrm{W}^1_0, \mathrm{W}^1_1, \mathrm{W}^2_1 \}$.}
    \label{fig:shannon-wavelet-matrix-tree}
\end{figure}

This implies that $C^S$ is a block diagonal matrix with blocks of size $2^j \times 2^j$ for each leaf node $\mathrm{W}^j_m\in \Lambda_T$. Let us denote the block corresponding to $\mathrm{W}^j_m$ by $F^{S} (j, m)$.  See Figure \ref{fig:shannon-wavelet-matrix-tree} for an example of the block-diagonal structure of $C^S$. Furthermore, given the periodicity of the exponent term of $C^S_{m2^j+n, m2^j+c}$, we can replace $d(m2^j+c)$ with $\tilde{d}(j, m, c)$ by definition \eqref{eq:retainingdecoding}. Then for $\mathrm{W}^j_m \in \Lambda_T,  0 \le n, c < 2^j$,
\begin{equation} \label{equation:ShannonAdaptive}
C^S_{m2^j+n, m2^j+c} = (F^S(j, m))_{n, c} = 2^{-j/2} e^{\i 2^{-j}2\pi n\cdot \tilde{d}(c, j, m)}.
\end{equation}
This can also be written more compactly as:
\begin{equation}
C^S = \bigoplus_{\mathrm{W}^j_m \in \Lambda_T} F^S(j,m).
\end{equation}

\subsection{Implementation on a quantum computer}

\begin{figure}[h!]
    \centering
    \includegraphics[scale=0.45,trim= 0 0 690 0,clip]{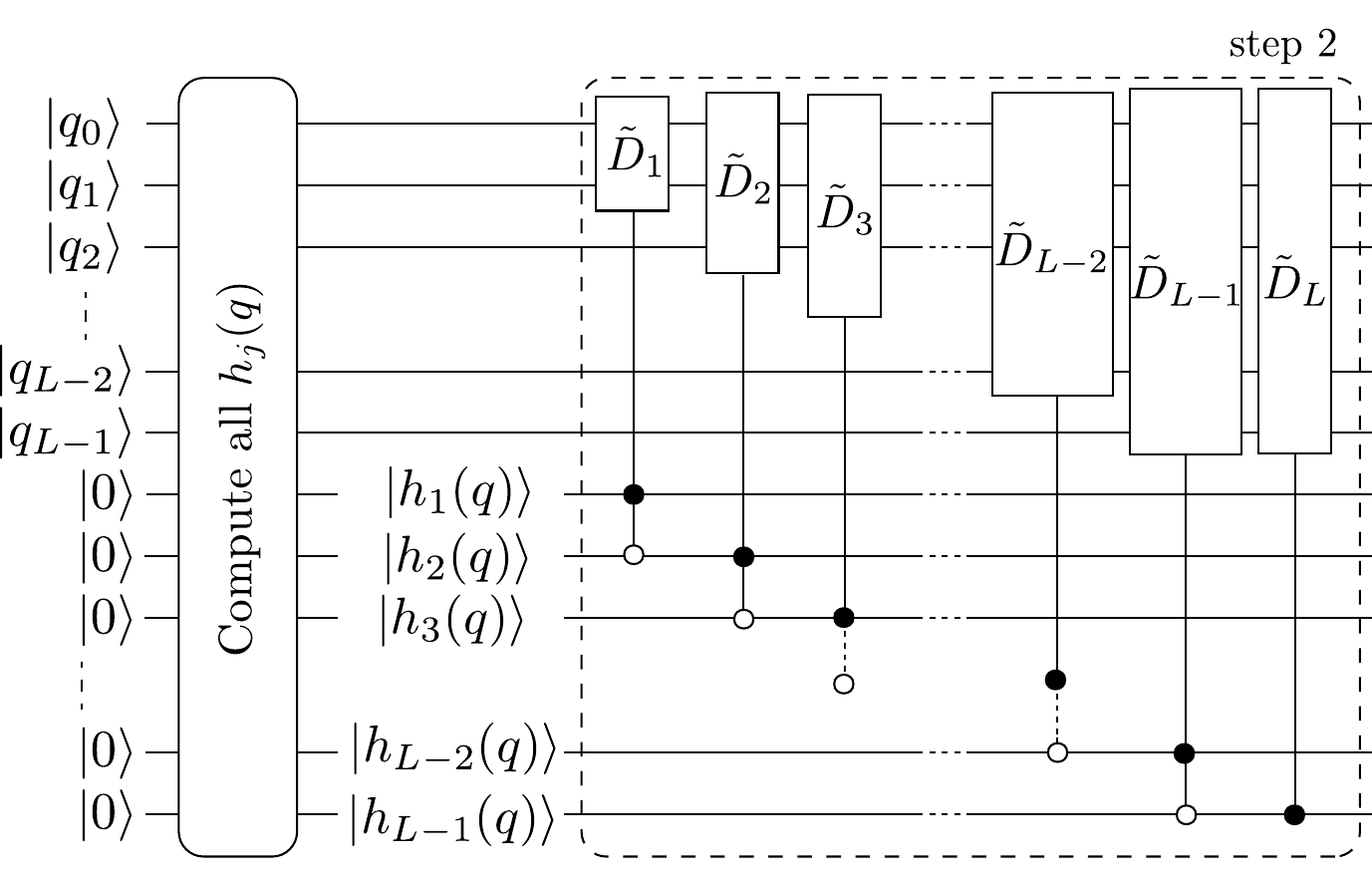}
    \\
    \hspace{1in}\includegraphics[scale=0.45,trim= 670 0 0 0,clip]{Shannon/Figures/shannon-circuit.pdf}
    \caption{{Quantum circuit implementing Shannon wavelet transform $C^S$} of size $2^L$.
    The family of boolean functions $h_j$ is defined in \eqref{eq:helper-functions}, the retaining decoding circuits $\tilde{D}_j$ are shown in Figure \ref{fig:decoding}. $Q_j^R$ and $Q_j^S$ are steps of the quantum Fourier transform, see Figure \ref{fig:qft-circuit}.
    }\label{fig:shannon-circuit}
\end{figure}

As shown above, $F^S(j, m)$ consists of two actions: (1) decoding $\tilde{d}$, and (2) applying a QFT on $j$ qubits. For the decoding $\tilde{d}$ we pass $j+1$ qubits to the retaining decoding circuit $\tilde{D}_j$ (see Figure \ref{fig:decoding}), where the most significant qubit gives parity of $m$. To know $j$ we use the helper decoding functions $h_j$ by definition \eqref{eq:helper-functions}. Finally, we apply QFT on $j$ qubits controlled by ancilla qubits that contain the value of $h_j$. 

Combining these steps, we get the quantum circuit shown in Figure \ref{fig:shannon-circuit}. For a detailed description, see \ifthenelse{\boolean{originalformat}}{Algorithm \ref{algo:shannon-wavelet} in Appendix \ref{appendix:algorithms-shannon}}{the Supplementary Material}.
The total gate complexity of the algorithm is $O(L^2)$. In summary, we have shown:
\begin{theorem}\label{theorem:algorithms-shannon}
Let $T$ be a monotonic wave packet admissible tree defined for discretization of size $O(2^L)$. If $T$ is represented by the leftmost node on each level, then the gate complexity of quantum Shannon wavelet transform generated by the tree $T$ is $O(L^2)$. The complexity of the circuit construction is $O(L^2)$.
\end{theorem}

%% file: Wave-atoms/05-Wave-atoms.tex
\ifthenelse{\boolean{originalformat}}{
 \section{Quantum wave atom transform}
}{
 \section{Results: Quantum wave atom transform}
}
\label{section: quantum wave atom transform}

In this section, we define and implement the quantum wave atom transform in one dimension, and we prove that it is efficient, for a variety of choices of the mother wave atom and the wave atom tree structure (see Theorem \ref{theorem:complexity-general} and Corollary \ref{theorem:complexity-specific}). This is a step towards implementing two-dimensional wave atoms, which have numerous applications \cite{DEMANET2007368}, on a quantum computer. 

To do this, we address two technical issues. First, in contrast to the Shannon wavelet transform, wave atoms have supports that overlap in the frequency domain; therefore, the wave atom transform requires considering the ``blending'' of the elements of~$\hat{f}$. 

Second, for many applications, it is desirable to use wave atoms that obey a ``parabolic'' scaling relationship, i.e., $\text{wavelength} \sim (\text{diameter of support})^2$ \cite{DEMANET2007368}. In order to do this, it is necessary to use more general classes of wave packet trees (introduced in Section \ref{section: wave packet encoding}), combined with a specific construction of wave packets that allows neighboring wave packets to have different scales (due to Villemoes \cite{Villemoes2002WaveletPW}, and described in Section \ref{section:wave-packets-background}). 

These techniques are quite flexible, and can be used to implement quantum wave atom transforms based on a variety of wave packet tree structures (see Theorem \ref{theorem:complexity-general}), including monotonic trees (see Corollary \ref{theorem:complexity-specific}). 


\input{Wave-atoms/05-Wave-atoms-intro}
\input{Wave-atoms/05-Wave-atoms-transform}

\input{Wave-atoms/05-Wave-atoms-decomposition}
\input{Wave-atoms/05-F-A}

\input{Wave-atoms/05-G-A}
\input{Wave-atoms/05-R}

\input{Wave-atoms/05-Complexity}

%% file: Wave-atoms/05-Wave-atoms-intro.tex
\subsection{Wave atoms' support over the frequency domain}\label{section:wave-atoms-form}
Recall the definition of wave atoms $\hat{\psi}^j_m$ by equations \eqref{def:wave-atoms} and \eqref{def:psi-hat},
\begin{small}
$$
\begin{gathered}
\hat{\psi}^j_m(\xi) = 2^{-j/2} \hat{\psi}^0_m(2^{-j} \xi),\\
\hat{\psi}_m^0(\xi) = e^{-\i \pi \xi}\big[e^{\i \alpha_m} g((-1)^m (2\pi \xi - 2\alpha_m)) + e^{-\i\alpha_m} g((-1)^{m+1} (2\pi \xi + 2\alpha_m))\big],
\end{gathered}
$$
\end{small}
where $\alpha_m = \tfrac{\pi}2(m+\tfrac12)$. The structure of wave atom is dictated by the function $g$ and its properties. Given that the support of function $g$ is $(-\tfrac76\pi, \tfrac56\pi)$, we derive the support of the wave atom $\hat{\psi}^j_m$ over the frequency domain:
\begin{equation} \label{eq:waveatomssupport}
\begin{gathered}
2^{j-1} \left(-\tfrac43, \tfrac23 \right) \cup  2^{j-1} \left(-\tfrac23, \tfrac43 \right) = 2^{j-1} \left(-\tfrac43, \tfrac43 \right), \text{ if } m= 0,\\
2^{j-1} \left( -\left(m+\tfrac53\right), -\left(m-\tfrac13\right) \right)  \cup 2^{j-1} \left( \left(m-\tfrac13\right), \left(m+\tfrac53\right) \right), \text{ if }m \text{ is odd},\\
2^{j-1} \left( -\left(m+\tfrac43\right), -\left(m-\tfrac23\right) \right)  \cup 2^{j-1} \left( \left(m-\tfrac23\right), \left(m+\tfrac43\right) \right), \text{ if } m \text{ is even}.
\end{gathered}
\end{equation}
(See Figure \ref{fig:Wave-atoms-image}.) The support of each wave atom consists of two intervals, one for negative and one for positive frequencies. This can be rewritten as a union of four intervals with centers $\pm2^{j-1}m$ and $\pm 2^{j-1}(m+1)$, as follows:
\begin{equation}
\begin{gathered}
\pm 2^{j-1} \left(-\tfrac23, \tfrac23 \right) \cup \pm 2^{j-1} \left(1-\tfrac13, 1+\tfrac13 \right), \text{ if } m= 0,\\
\pm 2^{j-1} \left(m-\tfrac13, m+\tfrac13 \right) \cup \pm 2^{j-1} \left((m+1)-\tfrac23, (m+1)+\tfrac23 \right), \text{ if } m \text{ is odd},\\
\pm 2^{j-1} \left(m-\tfrac23, m+\tfrac23 \right) \cup \pm 2^{j-1} \left((m+1)-\tfrac13, (m+1)+\tfrac13 \right), \text{ if } m \text{ is even}.
\end{gathered}
\end{equation}


Let us denote those intervals centered around $\pm2^{j-1}m$ and $\pm 2^{j-1}(m+1)$ as $\chi_{\pm2^{j-1}m}$ and $\chi_{\pm 2^{j-1}(m+1)}$ correspondingly. Then \eqref{eq:waveatomssupport} can be rewritten in general case as 
\begin{equation}\label{eq:xicup}
   \chi_{-2^{j-1}(m+1)} \cup \chi_{-2^{j-1}m} \cup \chi_{2^{j-1}m} \cup \chi_{2^{j-1}(m+1)}  
\end{equation}

We will implement a discrete wave atom transform, by restricting to integer frequencies only. For this purpose, let us denote the integer part of the half interval lengths of $\chi_{\pm 2^{j-1}m}$ and $\chi_{\pm 2^{j-1}(m+1)}$ by $\mu_0(j, m)$ and $\mu_1(j, m)$ respectively. Observe that they are equal to
\begin{equation}\label{prop:integer support}
\mu_0(j, m) = \lfloor 2^{j-I(m \text{ is odd})} / 3 \rfloor, \ \mu_1(j, m) = \lfloor 2^{j-I(m \text{ is even})} / 3 \rfloor.
\end{equation}
\begin{proposition}[Integer support of wave atoms]\label{prop:integer-support-psi}
An integer frequency $k \in \mathbb{Z}$ is in the support of $\hat{\psi}^j_m$ if and only if
$$
||k|-2^{j-1}m| \le \mu_0(j, m) \text{ or } ||k|-2^{j-1}(m+1)| \le \mu_1(j, m).
$$
\end{proposition}

\input{Wave-atoms/APX-Wave-atoms-form}
\noindent One can notice that the intersection of the intervals consists of one point $\chi_{2^{j-1}(m+1)} \cap \chi_{2^{j-1}m} = \{2^{j-1}m + 2^{j-I(m \text{ is odd})}/3\}$ and this value is always fractional, so there is no $k \in \mathbb{Z}$ such that
\begin{small}
$$
||k|-2^{j-1}m| \le \mu_0(j, m) \text{ \textit{and} } ||k|-2^{j-1}(m+1)| \le \mu_1(j, m).
$$
\end{small}


%% file: Wave-atoms/APX-Wave-atoms-form.tex
\begin{prf}  \label{proof:prop:integer-support-psi}
Given that the support is symmetric around zero. Without loss of generality, we consider only the non-negative frequencies.
For odd $m$, the support of $\hat{\psi}^j_m$ is
\begin{gather*}
\left(2^{j-1}(m-\tfrac13), 2^{j-1}(m+1+\tfrac23)\right) \\ = \left(2^{j-1}(m-\tfrac13), 2^{j-1}(m+\tfrac13)\right) \cup \{2^{j-1}(m+\tfrac13)\} \cup \left(2^{j-1}(m+1-\tfrac23), 2^{j-1}(m+1+\tfrac23) \right).
\end{gather*}
Given that $2^{j-1}(m+\tfrac13)$ is never an integer, we can ignore this point. Further, as $2^{j-1}m$ and $2^{j-1}(m+1)$ are integers, the integer support of $\hat{\psi}^j_m$ becomes
$$
\left[ 2^{j-1}m - \lfloor2^{j-1}/3\rfloor, 2^{j-1}m + \lfloor2^{j-1}/3\rfloor\right] \cup \left[ 2^{j-1}(m+1) - \lfloor2^{j}/3\rfloor, 2^{j-1}(m+1) + \lfloor2^{j}/3\rfloor\right].
$$
Given that for odd $m$, $\mu_0(j, m)=\lfloor2^{j-1}/3\rfloor$ and $\mu_1(j, m)=\lfloor2^{j}/3\rfloor$, this concludes the proof for the case of odd $m$. Similarly for even $m$.
\end{prf}

%% file: Wave-atoms/05-Wave-atoms-transform.tex
\subsection{Wave atom admissible trees}\label{section:wave atom admissible tree}
We would like to construct a (discrete) wave atom transform based on an admissible wave packet tree $T$. Here, we state some conditions on $T$ that are sufficient to ensure that the resulting wave atom basis is orthogonal.

We would like every integer $k \in \{-N/2, \ldots, N/2-1\}$ to lie in the support of two wave atoms, unless $k$ is close to zero (which is covered by a single wave atom) or close to $\pm N$ (which is treated as a special case). As such when we consider two neighboring leaf nodes $\mathrm{W}^j_m, \mathrm{W}^{j'}_{m'} \in \Lambda_T$, where $\mathrm{W}^j_m$ is the left neighbor, we require that the intervals centered around $\pm(m+1)2^j$ related to the two nodes are of equal length,
\begin{equation}\label{eq:mu-condition}
\mu_1(j, m) =\mu_0(j', m'). 
\end{equation}
We define the subset of all wave packet admissible trees that satisfy the condition in equation \eqref{eq:mu-condition}.
\begin{definition}
\label{def:wave-atom-admissible-tree}
Consider an admissible wave packet tree with a root $\mathrm{W}^L_0$ and $K$ leaf nodes $\{\mathrm{W}^{j_i}_{m_i}\;|\;i = 1, \ldots, K\}$ traversed from left to right. If the following conditions hold for all $i \in \{1, \ldots, K-1\}$:
\begin{enumerate}
\item $|j_i - j_{i+1}| \le 1$,
\item if $j_{i+1} = j_i + 1$, then both $m_i$ and $m_{i+1}$ are odd,
\item if $j_{i+1} = j_i - 1$, then both $m_i$ and $m_{i+1}$ are even,
\end{enumerate}
we call such a tree a \emph{wave atom admissible tree}.
\end{definition}

It is straightforward to see that other admissible wave packet trees do not satisfy the condition \eqref{eq:mu-condition} for some pair of neighboring leaf nodes of the tree. Given the Definition \ref{def:wave-atom-admissible-tree} and its reliance on the condition \eqref{eq:mu-condition}, we state an evident property of the wave atom admissible trees.
\begin{proposition}\label{prop:in-support-of-two}
Consider a wave atom admissible tree $T$ with a root node $\mathrm{W}^L_0$. Denote by $\mathrm{W}^{j_l}_0$ and $\mathrm{W}^{j_r}_{m_r}$ the leftmost and rightmost leaf nodes respectively. For any $k \in \{-N/2, \ldots, N/2-1\}$, where $N=2^L$, the following holds:
\begin{enumerate}
\item If $|k| \le \mu_0(j_l, 0)$, then $k$ is in the support of only $\hat{\psi}^{j_l}_0$;
\item If $|k| > m_r2^{j_r-1} + \mu_0(j_r, m_r)$, then $k$ is in the support of only $\hat{\psi}^{j_r}_{m_r}$;
\item Otherwise, $k$ is in the support of exactly two wave atoms.
\end{enumerate}
\end{proposition}
\begin{prf}
Follows directly from condition \eqref{eq:mu-condition}, observation that $\mu_0(j, m)+\mu_1(j,m)=2^{j-1}-1$ and $\chi_{\pm 2^{j-1}(m+1)} \cap \chi_{\pm 2^{j-1}m}$ has no integer points. 
\end{prf}

In addition, we would like the discrete wave atom transform to be a unitary operator. This requires special treatment of the lowest and highest frequencies, corresponding to the left-most and right-most wave packets in the tree $T$. 
So we define a modified basis $\lbrace \tilde{\psi}^j_{m} \rbrace$ as follows: we straighten the rightmost wave packet corresponding to $\mathrm{W}^{j_r}_{m_r} \in \Lambda_T$ as its support extends beyond the interval of interest $\{-N/2, \ldots, N/2-1\}$.  Moreover, we also modify the wave packet corresponding to the leftmost leaf node $\mathrm{W}^{j_l}_0 \in \Lambda_T$ as its support approaches the origin, as follows:
\begin{small}
\begin{equation}\label{eq:waveatomstilda}
\tilde{\psi}^j_{m}(k) = \begin{cases}
2^{-j/2} e^{-\i \pi 2^{-j} k} e^{(-1)^{k< 0} \i \alpha_0} \text{ if } |k| \le \mu_0(j_l, 0), \mathrm{W}^j_m = \mathrm{W}^{j_l}_0,\\
2^{-j/2} e^{-\i \pi 2^{-j} k} e^{(-1)^{k< 0} \i \alpha_{m_r}} \text{ if } |k| > m_r 2^{j_r} + \mu_0(j_r, m_r), \mathrm{W}^j_m = \mathrm{W}^{j_r}_{m_r},\\
\hat{\psi}^j_{m}(k) \text{ otherwise}.
\end{cases}
\end{equation}
\end{small}
Similarly to equation (\ref{def:wave-atoms}), we then define wave atoms with the shift of $2^{-j}n$ in the spatial domain to be
\begin{equation}\label{eq:waveatomstilda-shift}
\tilde{\psi}_{m,n}^j(k) = e^{-\i 2\pi n2^{-j} k} \tilde{\psi}_m^j(k).
\end{equation}
\begin{figure}
    \centering
    \includegraphics[scale=0.6]{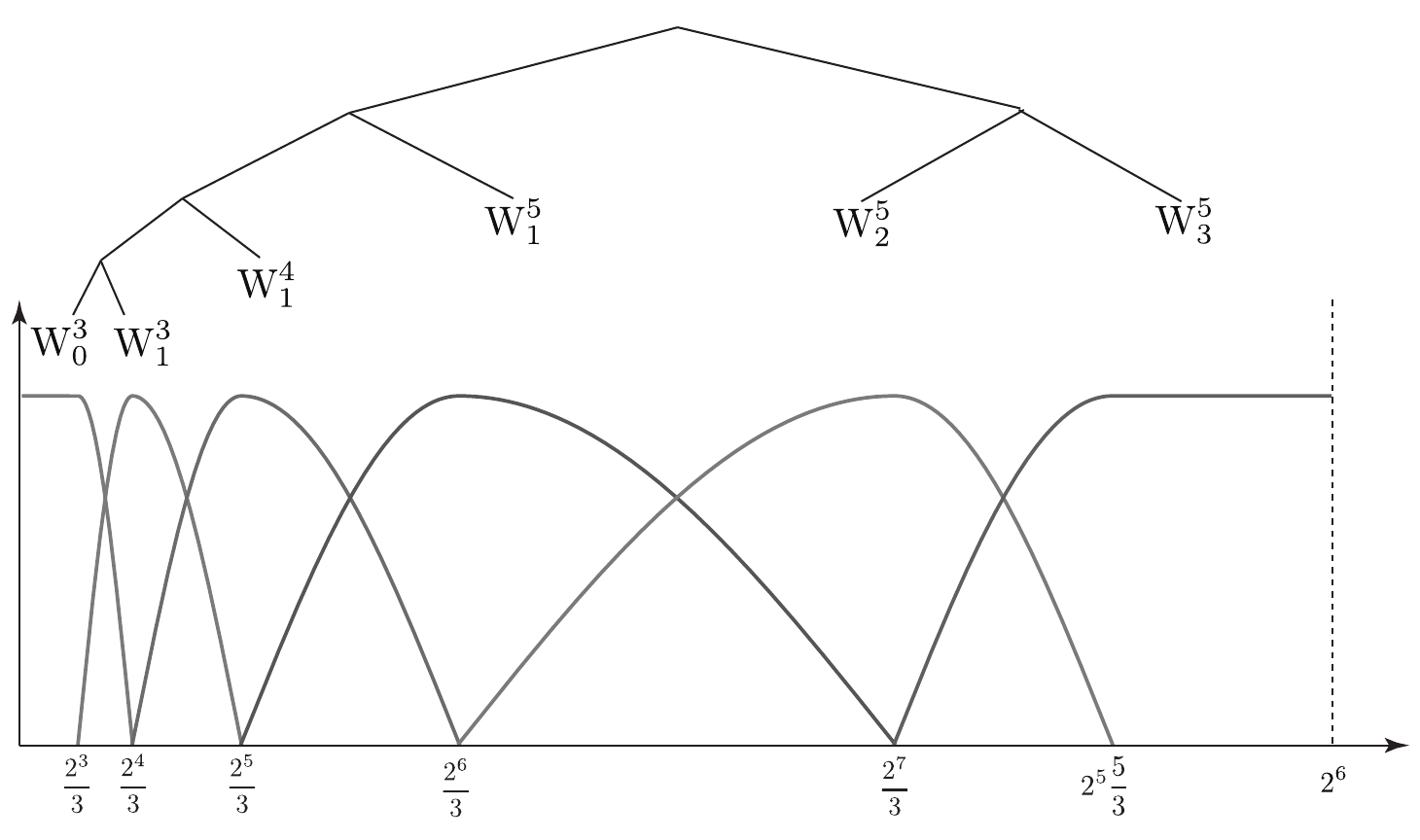}
    \caption{Wave atoms $\tilde{\psi}^j_m$ corresponding to the wave atom admissible tree on the top (absolute values). The figure shows only the positive side of the frequency domain. The negative side of the frequency domain is symmetric.}
    \label{fig:wave-packets-with-tree}
\end{figure}
An example of this construction is shown on Figure \ref{fig:wave-packets-with-tree}. 
Now, we can formally define the quantum wave atom transform.
\begin{definition}\label{def:wave-atom-transform}
Consider a wave atom admissible tree $T$ with a root node $\mathrm{W}^L_0$. Define a matrix $C^A$ as
\begin{equation}\label{eq:wave-atom-transform-definition}
C^A_{m2^j+n, i} = \overline{\tilde{\psi}^j_{m, n}(d(i))},
\end{equation}
where $\mathrm{W}^j_m \in \Lambda_T$, $d$ is the decoding function in equation \eqref{eq:decoding}, and $n, i \in \mathbb{Z}$ such that $0 \le n < 2^j$ and $0 \le i < N$.
The following transform of size $N=2^L$, where $Q_L$ is QFT,
\begin{equation}
\begin{gathered}
\sum_{k=0}^{N-1} f[k]\ket{k} \mapsto \sum_{\mathrm{W}^j_m \in \Lambda_T}\sum_{n=0}^{2^j-1} \left(\sum_{i=0}^{N-1} C^A_{m2^j+n, i} \hat{f}[d(i)]\right) |m2^j+n\rangle 
\\ \text{ where } \hat{f}[k] = \begin{cases} \sum_{i=0}^{N-1} (Q_L^\dagger)_{ki} f[i], \text{ if } k \ge 0,\\
\sum_{i=0}^{N-1} (Q_L^\dagger)_{(k+N)i} f[i], \text{ if } k < 0,
\end{cases}
\end{gathered}
\end{equation}
is called a \textit{quantum wave atom transform}.
\end{definition}

%% file: Wave-atoms/05-Wave-atoms-decomposition.tex
\subsection{Decomposition of $C^A$}\label{section:Decomposition}
\begin{figure}
    \centering
    \begin{tabular}{c} 
        $C^A |\hat{f}\rangle$ \\
    \begin{tabular}{|c|l}
        \hline \cellcolor{lightgray} \phantom{$\hat{f}$} \\
        \hline \cellcolor{lightgray} \phantom{$\hat{f}$} \\
        \hline \cellcolor{lightgray} \phantom{$\hat{f}$} \\
        \hline \cellcolor{lightgray} \phantom{$\hat{f}$} \\
        \hline \cellcolor{lightgray} \phantom{$\hat{f}$} \\
        \hline \cellcolor{olive} \phantom{$\hat{f}$} \\
        \hline \cellcolor{olive} \phantom{$\hat{f}$} \\
        \hline \cellcolor{olive} \phantom{$\hat{f}$} \\
        \hline \cellcolor{olive} \phantom{$\hat{f}$} \\
        \hline \cellcolor{olive} \phantom{$\hat{f}$} \\
        \hline
    \end{tabular}
    \end{tabular}
    \begin{tabular}{c} 
        $F^A$ \\
    \begin{tikzpicture}
        \draw (0,4.45) -- (0.01,4.45); 
        \draw[ultra thick, <-] (0,3.25) -- (0.5,3.25); 
        \draw[dashed] (0,2.12) -- (0.5,2.12); 
        \draw[ultra thick, <-] (0,1) -- (0.5,1); 
        \draw (0,0) -- (0.01,0); 
    \end{tikzpicture}
    \end{tabular}
    \begin{tabular}{c} 
        \phantom{$\hat{f}$} \\
    \begin{tabular}{|c|l}
        \hline \phantom{$\hat{f}$} \\
        \hline \cellcolor{YellowGreen} \phantom{$\hat{f}$} \\
        \hline \cellcolor{Goldenrod} \phantom{$\hat{f}$} \\
        \hline \phantom{$\hat{f}$} \\
        \hline \cellcolor{Orchid} \phantom{$\hat{f}$} \\
        \hline \cellcolor{Orchid} \phantom{$\hat{f}$} \\
        \hline \phantom{$\hat{f}$} \\
        \hline \cellcolor{Goldenrod} \phantom{$\hat{f}$} \\
        \hline \cellcolor{YellowGreen} \phantom{$\hat{f}$} \\
        \hline \phantom{$\hat{f}$} \\
        \hline
    \end{tabular}
    \end{tabular}
    \begin{tabular}{c} 
        $R^\dagger$ \\
    \begin{tikzpicture}
        \draw (0,4.45) -- (0.01,4.45); 
        \draw[<-] (0,1) -- (0.2,1) -- (0.4,3.8) -- (0.5,3.8); 
        \draw[<-] (0,3.8) -- (0.2,3.8) -- (0.4,1) -- (0.5,1);
        \draw (0,0) -- (0.01,0); 
    \end{tikzpicture}
    \end{tabular}
    \begin{tabular}{c} 
        \phantom{$\hat{f}$} \\
    \begin{tabular}{|c|l}
        \hline \phantom{$\hat{f}$} \\
        \hline \cellcolor{Goldenrod} \phantom{$\hat{f}$} \\
        \hline \cellcolor{Goldenrod} \phantom{$\hat{f}$} \\
        \hline \phantom{$\hat{f}$} \\
        \hline \cellcolor{Orchid} \phantom{$\hat{f}$} \\
        \hline \cellcolor{Orchid} \phantom{$\hat{f}$} \\
        \hline \phantom{$\hat{f}$} \\
        \hline \cellcolor{YellowGreen} \phantom{$\hat{f}$} \\
        \hline \cellcolor{YellowGreen} \phantom{$\hat{f}$} \\
        \hline \phantom{$\hat{f}$} \\
        \hline
    \end{tabular}
    \end{tabular}
    \begin{tabular}{c} 
        $G^A$ \\
    \begin{tikzpicture}
        \draw (0,4.45) -- (0.01,4.45); 
        \draw[<-] (0,3.85) -- (0.5,3.85); 
        \draw[<-] (0,3.75) -- (0.5,3.45);
        \draw[<-] (0,3.45) -- (0.5,3.75);
        \draw[<-] (0,3.35) -- (0.5,3.35);
        \draw[<-] (0,2.45) -- (0.5,2.45); 
        \draw[<-] (0,2.35) -- (0.5,2.05);
        \draw[<-] (0,2.05) -- (0.5,2.35);
        \draw[<-] (0,1.95) -- (0.5,1.95);
        \draw[<-] (0,1.05) -- (0.5,1.05); 
        \draw[<-] (0,0.95) -- (0.5,0.65);
        \draw[<-] (0,0.65) -- (0.5,0.95);
        \draw[<-] (0,0.55) -- (0.5,0.55);
        \draw (0,0) -- (0.01,0); 
    \end{tikzpicture}
    \end{tabular}
    \begin{tabular}{c} 
        \phantom{$\hat{f}$} \\
    \begin{tabular}{|c|l}
        \hline \phantom{$\hat{f}$} \\
        \hline \cellcolor{Salmon} \phantom{$\hat{f}$} \\
        \hline \cellcolor[HTML]{AAACED} \phantom{$\hat{f}$} \\
        \hline \phantom{$\hat{f}$} \\
        \hline \cellcolor[HTML]{AAACED} \phantom{$\hat{f}$} \\
        \hline \cellcolor{Salmon} \phantom{$\hat{f}$} \\
        \hline \phantom{$\hat{f}$} \\
        \hline \cellcolor[HTML]{AAACED} \phantom{$\hat{f}$} \\
        \hline \cellcolor{Salmon} \phantom{$\hat{f}$} \\
        \hline \phantom{$\hat{f}$} \\
        \hline
    \end{tabular}
    \end{tabular}
    \begin{tabular}{c} 
        $R$ \\
    \begin{tikzpicture}
        \draw (0,4.45) -- (0.01,4.45); 
        \draw[<-] (0,1) -- (0.2,1) -- (0.4,3.8) -- (0.5,3.8); 
        \draw[<-] (0,3.8) -- (0.2,3.8) -- (0.4,1) -- (0.5,1);
        \draw (0,0) -- (0.01,0); 
    \end{tikzpicture}
    \end{tabular}
    \begin{tabular}{c} 
        $|\hat{f}\rangle$ \hspace{1in} \\
    \begin{tabular}{|c|l}
        \cline{1-1} \phantom{$\hat{f}$} $\cdots$ & $\cdots$ \\
        \cline{1-1} \cellcolor[HTML]{AAACED} $\hat{f}(-m2^{j-1}+s)$ & $\ket{m2^j-2s-1}$ \\
        \cline{1-1} \cellcolor[HTML]{AAACED} $\hat{f}(m2^{j-1}-s)$ & $\ket{m2^j-2s}$ \\
        \cline{1-1} \phantom{$\hat{f}$} $\cdots$ & $\cdots$ \\
        \cline{1-1} \cellcolor[HTML]{AAACED} $\hat{f}(-m2^{j-1})$ & $\ket{m2^j-1}$ \\
        \cline{1-1} \cellcolor{Salmon} $\hat{f}(m2^{j-1})$ & $\ket{m2^j}$ \\
        \cline{1-1} \phantom{$\hat{f}$} $\cdots$ & $\cdots$ \\
        \cline{1-1} \cellcolor{Salmon} $\hat{f}(-m2^{j-1}-s)$ & $\ket{m2^j+2s-1}$ \\
        \cline{1-1} \cellcolor{Salmon} $\hat{f}(m2^{j-1}+s)$ & $\ket{m2^j+2s}$ \\
        \cline{1-1} \phantom{$\hat{f}$} $\cdots$ & $\cdots$ \\
        \cline{1-1}
    \end{tabular}
    \end{tabular}
    \caption{{The implementation of the operation $C^A$,} defined in Eq.~(\ref{eq:wave-atom-transform-definition}). This represents a wave atom transform, omitting the initial step of applying a discrete Fourier transform (DFT). In other words, one starts from the state $|\hat{f}\rangle = \sum_{k=-N/2}^{N/2-1} \hat{f}[k] \ket{e(k)} = \sum_{k=0}^{N-1} \hat{f}[d(k)]\ket{k}$ from Eq.~(\ref{eqn-encoding-decoding-identity}). Each box in the diagram represents one component of the state-vector, {when the quantum state is written as a superposition of standard basis vectors (denoted by $\ket{\cdots}$ symbols on the right side of the diagram).} First the permutation $\rho$ from Eq.~(\ref{eq:perm-rho}) is applied, which rearranges the components of the state-vector. This is followed by the ``blending matrix'' $G^A$ from Eq.~(\ref{eq:block-g-orig}), which forms linear combinations of pairs of components of the state-vector, with coefficients given by $g_+[j,m,s]$ and $g_-[j,m,s]$ in Eq.~(\ref{eq:g-plus-minus}). To illustrate this, the diagram shows the intervals $\chi_{\pm m2^{j}}$ (see Eq.~(\ref{eq:xicup})) that represent the intersection of supports of two wave atoms: $\mathrm{W}^{j}_{m'}$ on the left side and $\mathrm{W}^{j}_{m}$ on the right side. Finally the inverse permutation $\rho^{-1}$ is applied, followed by the operation $F^A$, which implements a Fourier transform adapted to the tree structure $T$, defined in Eq.~(\ref{equation:f-definition}). {On the parts of the state-vector that are shown here, the action of $F^A$ consists of two blocks: $F^A(j,m')$ (acting on the top half of the vector) and $F^A(j,m)$ (acting on the bottom half of the vector).} This produces the desired state $C^A|\hat{f}\rangle$.}
    \label{fig:decomposition-schema}
\end{figure}

Recall that the Shannon wavelet transform was implemented using a Fourier-like transform that had a block-diagonal structure determined by the wave packet tree $T$ (see Section \ref{section:quantum-shannon-wavelet-transform}). To implement the wave atom transform, one can use the same idea, combined with a new technique: redistribution of weights between elements of the vector $\hat{f}$, sometimes called ``blending'' \cite{ni2024quantumwavepackettransforms}. This is needed because of the overlaps between the supports of neighboring wave atoms over the frequency domain (which did not occur in the case of Shannon wavelets).

In this section, we will decompose the matrix $C^A$ (from Definition \ref{def:wave-atom-transform}) into a product of a blending operation, and a Fourier-like transform.

Classically, for a given node $\mathrm{W}^j_m$ in the tree $T$, the wave atom transform blends $2^j$ pairs of elements of the vector $\hat{f}$ into a vector of size $2^j$ and applies the Fourier transform of size $2^j$ to obtain the vector of wave atom coefficients $\{c^D_{j,m,n}\}_{n=0}^{2^j-1}$ \cite{DEMANET2007368}. Unless we consider border cases of $\mathrm{W}^j_m$, the transform blends $\hat{f}[k]$ and $\hat{f}[k']$, for all pairs $(k,k')$ of the form:
\begin{gather*}
(-m2^{j-1} + s, m2^{j-1} + s), \quad  \forall |s| \le \mu_0(j, m),\\
(-(m+1)2^{j-1} + s, (m+1)2^{j-1} + s), \quad  \forall |s| \le \mu_1(j, m).
\end{gather*}
(This can be seen by recalling Definition \ref{def:wave-atom-transform}, equations (\ref{eq:waveatomstilda})-(\ref{eq:waveatomstilda-shift}), and the fact that $\hat{\psi}^j_m$ is supported on two intervals in the frequency domain, as in equation (\ref{eq:waveatomssupport}).) 


In order to implement the blending efficiently, we permute elements of the vector $\hat{f}$ such that the pairs involved in the blending become adjacent to each other. 
Given that the encoding procedure \eqref{eq:encoding} orders frequencies by increasing absolute values, for some $\mathrm{W}^j_m$ and $|s| \le \mu_0(j, m)$, the pair of frequencies $-m2^{j-1} + s$ and $m2^{j-1} + s$ are not neighboring indices when encoded unless $s=0$. Therefore, the blending procedure becomes a three-stage process.
\begin{enumerate}
\item Permute elements of the vector $\hat{f}$ such that $-m2^{j-1}+s$ and $m2^{j-1}+s$ become neighbors by swapping $-m2^{j-1}+s$ with $-m2^{j-1}-s$.
\item Redistribute the weights between adjacent elements according to the function $g$.
\item Reverse the permutation in Step 1.
\end{enumerate}
Figure \ref{fig:decomposition-schema} schematically illustrates the decomposition of $C^A$ by tracking the actions on pairs $\pm m2^{j-1}\pm s$. In the following, we define the building blocks that are used to construct $C^A$.

\begin{definition}\label{def:wave atoms fourier}
For each $\mathrm{W}^j_m \in \Lambda_T$, consider matrix $F^A(j, m) \in \mathbb{C}^{2^j \times 2^j}$  defined by elements
\begin{equation}\label{equation:f-definition}
(F^A(j, m))_{n, c} = 2^{-j/2} e^{\i 2\pi 2^{-j} (n + \tfrac12) d(m2^j + c)} e^{(-1)^{c + 1} \i \alpha_m}, \text{ for } 0 \le n, c < 2^j.
\end{equation}
Let us define matrix $F^A \in \mathbb{C}^{N \times N}$ of \textit{Fourier transform adapted to tree structure} of $T$ as a composition of matrices $F^A(j, m)$,
\begin{equation}
F^A = \bigoplus_{\mathrm{W}^j_m \in \Lambda_T} F^A(j,m),
\end{equation}
that is to say,
\begin{equation}
F^A_{m2^j+n, m2^j+c} = (F^A(j,m))_{n, c}, \text{ for } \mathrm{W}^j_m \in \Lambda_T,\, 0 \le n, c < 2^j.
\end{equation}
\end{definition}

\begin{definition}\label{def:permutation-rho}
Let us define \textit{permutation of indices} $\rho: \{0, \ldots, N-1\}\mapsto \{0, \ldots, N-1\}$ by its non-trivial actions as
\begin{equation}\label{eq:perm-rho}
\forall \mathrm{W}^j_m \in \Lambda_T \text{ s.t. } \mathrm{W}^j_m \ne \mathrm{W}^{j_l}_0, \;\forall |s| \le \mu_0(j, m), \; \rho(m2^j+2s-1) = m2^j-2s-1.
\end{equation}
Define its permutation matrix $R \in \mathbb{C}^{N \times N}$ such that nonzero elements of the matrix are
$$
R_{\rho(i) i} = 1 \text{ for } i \in  \{0, \ldots, N-1\}.
$$
\end{definition}
Let us use the following shorthand for function $g$. For $j \ge 1, m \ge 0$ and $s \in \mathbb{Z}$,
\begin{equation}\label{eq:g-plus-minus}
\begin{aligned}
g_-[j, m, s] &= g((-1)^{m}(2^{-j+1}\pi |s| - \tfrac{\pi}2)),\\
g_+[j, m, s] &= \i (-1)^{I(s < 0)} g((-1)^{m+1}(2^{-j+1}\pi |s| + \tfrac{\pi}2)).
\end{aligned}
\end{equation}
\begin{definition}\label{def:matrix-g}
Let us define \textit{blending matrix} $G^A \in \mathbb{C}^{N \times N}$ as follows: for any $\mathrm{W}^j_m \in \Lambda_T$ such that $\mathrm{W}^j_m \ne \mathrm{W}^{j_l}_0$, and for any $s\in\mathbb{Z}$ such that $|s| \le \mu_0(j, m)$,
\begin{equation}\label{eq:block-g-orig}
\begin{pmatrix}
G^A_{m2^j+2s-1, m2^j+2s-1} & G^A_{m2^j+2s-1, m2^j+2s}\\
G^A_{m2^j+2s, m2^j+2s-1} & G^A_{m2^j+2s, m2^j+2s}
\end{pmatrix} = \begin{pmatrix}
g_-[j, m, s] & g_+[j, m, s] \\
g_+[j, m, s]  & g_-[j, m, s] 
\end{pmatrix},
\end{equation}
and $G^A$ is the identity on other diagonal blocks.
\end{definition}

\begin{theorem}\label{theorem:decomposition}
The matrix $C^A$ (from Definition \ref{def:wave-atom-transform}) is a product of matrices given by Definitions \ref{def:wave atoms fourier}, \ref{def:permutation-rho}, \ref{def:matrix-g},
\begin{equation}\label{eq:decomposition}
C^A = F^A R^{\dagger} G^A R.
\end{equation}
\end{theorem}

Before proving Theorem \ref{theorem:decomposition}, we note that it implies the following:

\begin{corollary}\label{theorem:unitary}
The matrix $C^A$ is unitary.
\end{corollary}
\begin{prf}
This follows from the observation that each matrix in the product of equation \eqref{eq:decomposition} is unitary. The blocks of $G^A$ are unitary due to properties \eqref{prop:sum_of_squares} and \eqref{prop:change_of_sign} of the function $g$. {Note that $G^A$ will still be unitary, even if $g$ is only computed approximately (not exactly), when $g$ is computed using the methods in Section \ref{sec-approx-discret}.}
\end{prf}

\input{Wave-atoms/APX-Decomposition}

In the following subsections, we discuss the implementation of each of the operators, $F^A$, $G^A$ and $R$.

%% file: Wave-atoms/APX-Decomposition.tex
We now prove Theorem \ref{theorem:decomposition}. First, we state two identities due to periodicity of the complex exponential function and the definition of $\alpha_m=\tfrac{\pi}2(m+\tfrac12)$:
\begin{align}
e^{\i 2\pi 2^{-j} (n+\tfrac12) (-m2^{j-1}+s)} e^{\i \alpha_m} &= \i e^{\i 2\pi 2^{-j} (n+\tfrac12) (m2^{j-1}+s)} e^{-\i \alpha_m}\label{eq:periodicity-m}
\\
e^{\i 2\pi 2^{-j}(n+\tfrac12)((m+1)2^{j-1} + s)} e^{-\i \alpha_m} &= \i e^{\i 2\pi 2^{-j}(n+\tfrac12)(-(m+1)2^{j-1} + s)} e^{\i \alpha_m}\label{eq:periodicity-m+1}
\end{align}
Also we state a simple identity: 
\begin{equation}\label{eq:g-s-sign-switch}
g((-1)^m(2^{-j+1}\pi (-s) - \tfrac{\pi}2)) = g((-1)^{m+1}(2^{-j+1}\pi s + \tfrac{\pi}2))
\end{equation}

Our goal is to implement the operation described in Definition \ref{def:wave-atom-transform}:
\begin{equation}
\sum_{i=0}^{N-1} C^A_{m2^j+n, i} \hat{f}[d(i)] = 
\sum_{i=0}^{N-1} \overline{\tilde{\psi}^j_{m, n}(d(i))} \hat{f}[d(i)].
\end{equation}
Recall that $\tilde{\psi}^j_m$ is related to $\hat{\psi}^j_m$ via equations (\ref{eq:waveatomstilda})-(\ref{eq:waveatomstilda-shift}), and $\hat{\psi}^j_m$ is supported on two intervals in the frequency domain, described by equation (\ref{eq:waveatomssupport}) and proposition \ref{prop:integer-support-psi}. This motivates us to state and prove the following two lemmas.
\begin{lemma}\label{prop:left-interval-sum}
For any $\mathrm{W}^j_m \ne \mathrm{W}^{j_l}_0$, 
\begin{equation}\label{eq:lemma-left}
\begin{gathered}
\sum_{|s| \le \mu_0(j, m)} \left( \overline{\hat{\psi}^j_{m, n}(m2^{j-1}+s)} \hat{f}[m2^{j-1}+s] + \overline{\hat{\psi}^j_{m, n}(-m2^{j-1}+s)} \hat{f}[-m2^{j-1}+s]\right) \\ = \sum_{i_k=m2^j}^{m2^j+2\mu_0(j, m)} F^A_{m2^j+n, i_k} \sum_{c = 0}^{N-1} G^A_{\rho(i_k),c} \hat{f}[d(\rho(c))].
\end{gathered}
\end{equation}
\end{lemma}
\vspace{-1.5em}
\begin{prf}
For $s \ge 0$ we use identity \eqref{eq:periodicity-m} for points on negative side.
\begin{align*}
\overline{\hat{\psi}^j_{m, n}(m2^{j-1}+s)} &= 2^{-j/2} e^{\i 2\pi 2^{-j} (n+\tfrac12) (m2^{j-1}+s)} e^{-\i \alpha_m} g((-1)^m(2^{-j+1}\pi |s|- \tfrac{\pi}2))\\
&= F^A_{m2^j+n, m2^j+2s} G^A_{\rho(m2^j+2s), m2^j+2s},\\
\overline{\hat{\psi}^j_{m, n}(-m2^{j-1}+s)} &= 2^{-j/2} e^{\i 2\pi 2^{-j} (n+\tfrac12) (m2^{j-1}+s)} e^{-\i \alpha_m} \i g((-1)^{m+1}(2^{-j+1}\pi |s| + \tfrac{\pi}2)) \\
&= F^A_{m2^j+n, m2^j+2s} G^A_{\rho(m2^j+2s), m2^j+2s-1}.
\end{align*}
For $s < 0$ we use the identity \eqref{eq:periodicity-m} for points on positive side. Moreover, we use the identity \eqref{eq:g-s-sign-switch} for points on both sides.
\begin{align*}
\overline{\hat{\psi}^j_{m, n}(m2^{j-1}+s)} &= 2^{-j/2}e^{\i 2\pi 2^{-j} (n+\tfrac12) (-m2^{j-1}+s)} e^{\i \alpha_m} (-\i) g((-1)^{m+1}(2^{-j+1}\pi |s| + \tfrac{\pi}2))\\
&= F^A_{m2^j+n, m2^j+2|s|-1} G^A_{\rho(m2^j+2|s|-1), m2^j-2|s|},\\
\overline{\hat{\psi}^j_{m, n}(-m2^{j-1}+s)} &= 2^{-j/2} e^{\i 2\pi 2^{-j} (n+\tfrac12) (-m2^{j-1}+s)} e^{\i \alpha_m} g((-1)^m(2^{-j+1}\pi |s| - \tfrac{\pi}2)) \\
&= F^A_{m2^j+n, m2^j+2|s|-1} G^A_{\rho(m2^j+2|s|-1), m2^j-2|s|-1}.
\end{align*}
Therefore, the left-hand side of equation \eqref{eq:lemma-left} becomes
\begin{align*}
\sum_{|s| \le \mu_0(j, m)} &\left( \overline{\hat{\psi}^j_{m, n}(m2^{j-1}+s)} \hat{f}[m2^{j-1}+s] + \overline{\hat{\psi}^j_{m, n}(-m2^{j-1}+s)} \hat{f}[-m2^{j-1}+s]\right) \\
&= \sum_{s=0}^{\mu_0(j, m)} F^A_{m2^j+n, m2^j+2s} \Big( G^A_{\rho(m2^j+2s), m2^j+s} \hat{f}[d(\rho(m2^j+2s))] \\ 
&+ G^A_{\rho(m2^j+2s), m2^j+2s-1} \hat{f}[d(\rho(m2^j+2s-1))] \Big)
\\
&+ \sum_{s=1}^{\mu_0(j, m)}  F^A_{m2^j+n, m2^j+2|s|-1} \Big( G^A_{\rho(m2^j+2|s|-1), m2^j-2|s|} \hat{f}[d(\rho(m2^j-2|s|))]  \\
&+ G^A_{\rho(m2^j+2|s|-1), m2^j-2|s|-1} \hat{f}[d(\rho(m2^j-2|s|-1))]\Big) \\
&= \sum_{s=0}^{2\mu_0(j, m)} F^A_{m2^j+n, m2^j+s} \sum_{c=0} G^A_{\rho(m2^j+s), c} \hat{f}[d(\rho(c))].
\end{align*}
The latter step is due to the block diagonal structure of matrix $G^A$ with blocks of size $2\times 2$.
\end{prf}

\begin{lemma}\label{prop:right-interval-sum}
For any $\mathrm{W}^j_m \ne \mathrm{W}^{j_r}_{m_r}$, 
\begin{equation}
\begin{gathered}\label{eq:lemma-right}
\begin{aligned}
\sum_{|s| \le \mu_1(j, m)} \Big( &\overline{\hat{\psi}^j_{m, n}((m+1)2^{j-1}+s)} \hat{f}[(m+1)2^{j-1}+s] \\ + &\overline{\hat{\psi}^j_{m, n}(-(m+1)2^{j-1}+s)} \hat{f}[-(m+1)2^{j-1}+s]\Big) 
\end{aligned}
\\ = \sum_{i_k=m2^j+2\mu_0(j, m)+1}^{(m+1)2^j-1} F^A_{m2^j+n, i_k} \sum_{c = 0}^{N-1} G^A_{\rho(i_k),c} \hat{f}[d(\rho(c))].
\end{gathered}
\end{equation}
\end{lemma}
\vspace{-1.5em}
\begin{prf}
Before reviewing the sum, we show that for any $\mathrm{W}^j_m \ne \mathrm{W}^{j_r}_{m_r}$ and $s\in\mathbb{Z}$ such that $|s| \le \mu_1(j, m)$
\begin{equation}\label{eq:g-m+1}
\begin{gathered}
\begin{pmatrix}
G^A_{(m+1)2^j+2s-1, (m+1)2^j+2s-1} & G^A_{(m+1)2^j+2s-1, (m+1)2^j+2s}\\
G^A_{(m+1)2^j+2s, (m+1)2^j+2s-1} & G^A_{(m+1)2^j+2s, (m+1)2^j+2s}
\end{pmatrix} \\ = \begin{pmatrix}
g_-[j, m+1, s] & g_+[j, m+1, s] \\
g_+[j, m+1, s]  & g_-[j, m+1, s] 
\end{pmatrix}.
\end{gathered}
\end{equation}
Denote the right neighbor of $\mathrm{W}^j_m$ by $\mathrm{W}^{j'}_{m'}$, then 
\begin{itemize}
    \item If $j'=j$, then $m'=m+1$ and \eqref{eq:g-m+1} is simply the definition the block of matrix $G^A$;
    \item If $j'=j+1$, then both $m$ and $m'$ are odd and by properties of the function $g$.
\begin{equation}\label{eq:g-j+1-}
\begin{split}
 g&((-1)^{m'}(2^{-j'+1}\pi s - \tfrac{\pi}2)) = g(- 2^{-j'+1}\pi s+\tfrac{\pi}2) 
 \\ &= g(-\tfrac{\pi}2+2^{-j+1}\pi s) = g((-1)^{m+1}(2^{-j+1}\pi s - \tfrac{\pi}2)),
\end{split}
\end{equation}
and
\begin{equation}\label{eq:g-j+1+}
\begin{split}
 g&((-1)^{m'+1}(2^{-j'+1}\pi s+\tfrac{\pi}2))= g(2^{-j'+1}\pi s+\tfrac{\pi}2) 
 \\ &= g(-\tfrac{\pi}2-2^{-j+1}\pi s) = g((-1)^m(2^{-j+1}\pi s + \tfrac{\pi}2)).
\end{split}
\end{equation}
    \item For $j'=j-1$, then $m$ and $m'$ are even and the derivation is analogous to case of $j'=j+1$.
\end{itemize}
For $s \ge 0$ we use the identity \eqref{eq:periodicity-m+1} for points on the positive side.
\begin{align*}
&\overline{\hat{\psi}^j_{m, n}((m+1)2^{j-1}+s)} \\ 
&\qquad= 2^{-j/2} e^{\i 2\pi 2^{-j} (n+\tfrac{1}{2}) (-(m+1)2^{j-1}+s)} e^{\i \alpha_m} \i g((-1)^m(2^{-j+1}\pi |s| + \tfrac{\pi}{2})) \\
&\qquad= F^A_{m2^j+n, (m+1)2^j-2s-1} G^A_{\rho((m+1)2^j-2s-1), (m+1)2^j+2s},\\
&\overline{\hat{\psi}^j_{m, n}(-(m+1)2^{j-1}+s)}\\
&\qquad= 2^{-j/2} e^{\i 2\pi 2^{-j} (n+\tfrac{1}{2}) (-(m+1)2^{j-1}+s)} e^{\i \alpha_m} g((-1)^{m+1}(2^{-j+1}\pi |s| - \tfrac{\pi}{2})) \\
&\qquad= F^A_{m2^j+n, (m+1)2^j-2s-1} G^A_{\rho((m+1)2^j-2s-1), (m+1)2^j+2s-1}.
\end{align*}
For $s < 0$ we use the identity \eqref{eq:periodicity-m+1} for points on the negative side. Moreover, we use the identity \eqref{eq:g-s-sign-switch} for points on both sides.
\begin{align*}
&\overline{\hat{\psi}^j_{m, n}((m+1)2^{j-1}+s)} \\ 
&\qquad= 2^{-j/2} e^{\i 2\pi 2^{-j} (n+\tfrac12) ((m+1)2^{j-1}+s)} e^{-\i \alpha_m} g((-1)^{m+1}(2^{-j+1}\pi |s| - \tfrac{\pi}2)) \\
&\qquad= F^A_{m2^j+n, (m+1)2^j-2|s|} G^A_{\rho((m+1)2^j-2|s|), (m+1)2^j-2|s|},\\
&\overline{\hat{\psi}^j_{m, n}(-(m+1)2^{j-1}+s)}\\
&\qquad= 2^{-j/2} e^{\i 2\pi 2^{-j} (n+\tfrac12) ((m+1)2^{j-1}+s)} e^{-\i \alpha_m} (-\i) g((-1)^m(2^{-j+1}\pi |s| + \tfrac{\pi}2)) \\
&\qquad= F^A_{m2^j+n, (m+1)2^j-2|s|} G^A_{\rho((m+1)2^j-2|s|), (m+1)2^j-2|s|-1}.
\end{align*}
Therefore, the left-hand side of equation \eqref{eq:lemma-right} becomes
\begin{align*}
&\sum_{|s| \le \mu_1(j, m)} \Big( \overline{\hat{\psi}^j_{m, n}((m+1)2^{j-1}+s)} \hat{f}[(m+1)2^{j-1}+s] \\ 
&+ \overline{\hat{\psi}^j_{m, n}(-(m+1)2^{j-1}+s)} \hat{f}[-(m+1)2^{j-1}+s]\Big) \\
&\qquad= \sum_{s=0}^{\mu_1(j, m)} F^A_{m2^j+n, (m+1)2^j-2s-1} \times \Big( G^A_{\rho((m+1)2^j-2s-1), (m+1)2^j+2s} \hat{f}[d(\rho((m+1)2^j+2s))] \\ 
&\qquad+ G^A_{\rho((m+1)2^j-2s-1), (m+1)2^j+2s-1} \hat{f}[d(\rho((m+1)2^j+2s-1))] \Big) \\
&\qquad+ \sum_{s=1}^{\mu_1(j, m)} F^A_{m2^j+n, (m+1)2^j-2s} \Big(  G^A_{\rho((m+1)2^j-2s), (m+1)2^j-2s} \hat{f}[d(\rho((m+1)2^j-2s))] \\ 
&\qquad+ G^A_{\rho((m+1)2^j-2s), (m+1)2^j-2s-1} \hat{f}[d(\rho((m+1)2^j-2s-1))] \Big)\\
&\qquad= \sum_{s=1}^{2\mu_1(j, m)+1} F^A_{m2^j+n, (m+1)2^j-s} \sum_{c = 0}^{N-1} G^A_{\rho((m+1)2^j-s,c} \hat{f}[d(\rho(c))]
\end{align*}
The latter step is due to the block diagonal structure of matrix $G^A$ with blocks of size $2\times 2$.
\end{prf}

\vskip 11pt

\begin{prf}[of Theorem \ref{theorem:decomposition}]
First, we consider two special cases. Observe that for $\mathrm{W}^{j_l}_0$ and $i_k \le 2\mu_0(j_l, 0)$ and $n\in \{0, \ldots, 2^{j_l}-1\}$ we have
$$
C^A_{n, i_k} = F^A_{n, i_k}  \text{ and } R_{i_k, i_k} = G^A_{i_k, i_k} = 1.
$$
Also for $\mathrm{W}^{j_r}_{m_r}$ and $i_k > m_r 2^{j_r} + 2\mu_0(j_r, m_r)$ and $n\in \{0, \ldots, 2^{j_r}-1\}$ we have
$$
C^A_{m_r2^{j_r}+n, i_k} = F^A_{m_r2^{j_r}+n,i_k} \text{ and } R_{i_k, i_k} = G^A_{i_k, i_k} = 1.
$$

The general case can be handled using Lemmas \ref{prop:left-interval-sum} and \ref{prop:right-interval-sum}, and Proposition \ref{prop:integer-support-psi}. Combining all of the above, we get
\begin{equation}\label{eq:sum-c-frgr-f}
C^A \sum_{c}^{N-1} \hat{f}[d(c)] \ket{c} = F^A R^{\dagger} G^A R \sum_{c}^{N-1} \hat{f}[d(c)] \ket{c}.
\end{equation}
Observe that \eqref{eq:sum-c-frgr-f} holds for any $\hat{f}$ and conclude the proof.
\end{prf}

%% file: Wave-atoms/05-F-A.tex
\subsection{Quantum circuits for $F^A$}\label{section:block-diagonal-fourier}
\begin{figure}
    \centering
    \includegraphics[scale=0.5,trim= 0 0 866 0,clip]{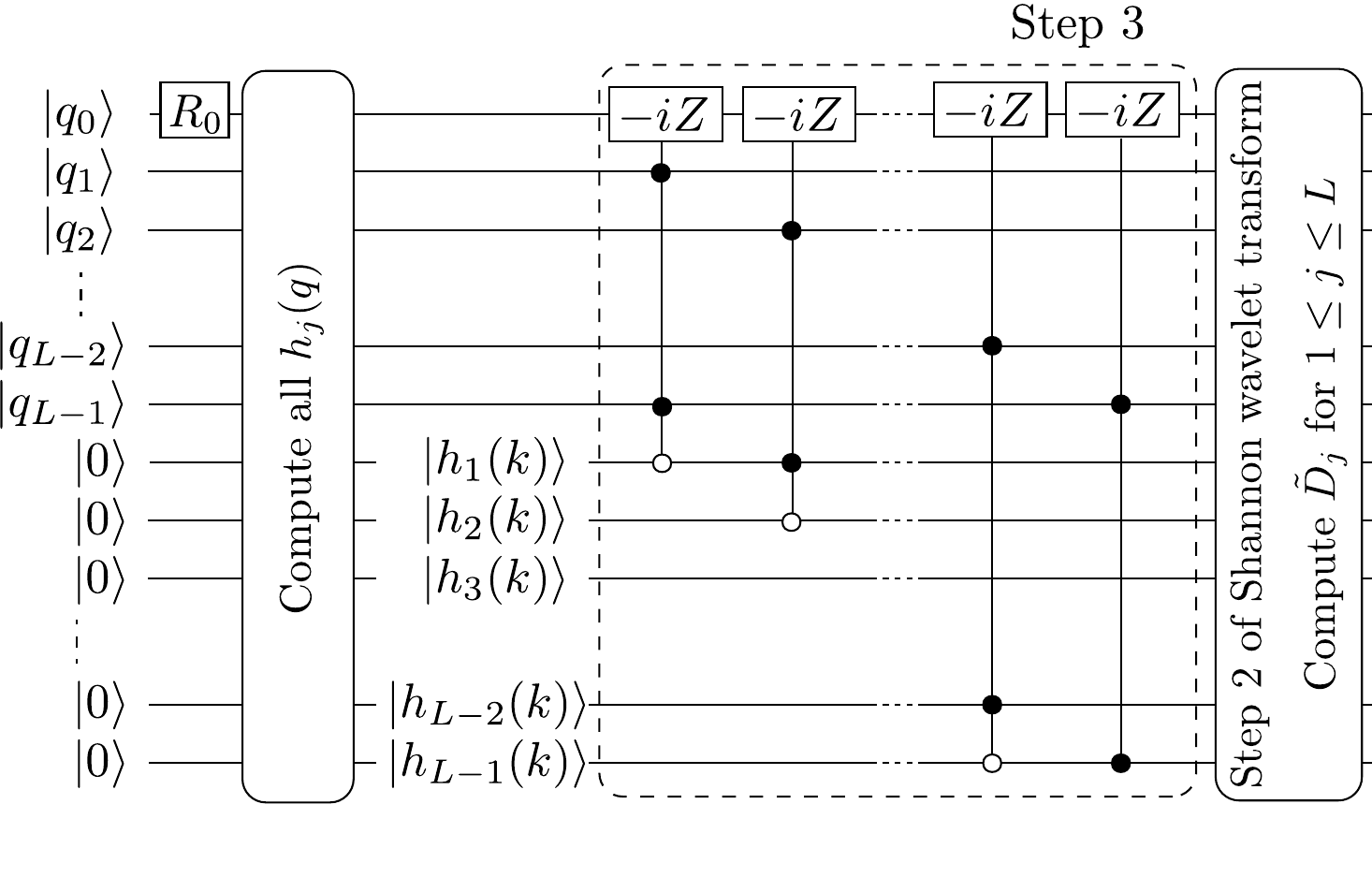}
    \\
    \includegraphics[scale=0.5,trim= 700 0 0 0,clip]{Wave-atoms/Figures/wavatom-circuit-f.pdf}
    \caption{{Quantum circuit implementation of the operation $F^A$} of size $2^L$. {The boxes labeled Step 3 and Step 5 apply the first two phase factors in Eq.~(\ref{F-block-decomp}). For more details, see the Supplementary Material. The third factor in Eq.~(\ref{F-block-decomp}) is handled using the same techniques as in the Shannon wavelet transform.} Steps 2-4 of the Shannon wavelet transform can be found in Figure \ref{fig:shannon-circuit}.}
    \label{fig:waveatoms-f-circuit}
\end{figure}
In this section we will implement the operation $F^A$ defined in (\ref{equation:f-definition}), for a variety of wave packet trees $T$. This is similar to the implementation of the Shannon wavelet transform in (\ref{equation:ShannonAdaptive}), modulo some technical details.

First, equation (\ref{equation:f-definition}) involves a complex phase factor that has a periodicity of $2^{j+1}$ (rather than $2^j$, in the case of (\ref{equation:ShannonAdaptive})):
$$
e^{\i 2^{-j} 2 \pi (n+\tfrac12) d(m2^j+c)} = e^{\i 2^{-j} 2 \pi (n+\tfrac12) (d(m2^j+c)+p2^{j+1})}.
$$ 
Recall that in the wave packet encoding \eqref{eq:row-encoding}, the $j$ least significant bits represent $c$ in binary form while the remaining bits represent $m$. So the $(j+1)$-th bit represents the parity of $m$, which determines the shape of the wavepacket (see Section \ref{section:wave-packets-background}). 

Recall functions $d$ and $\tilde{d}$ defined by \eqref{eq:decoding}, \eqref{local-decoding} correspondingly and observe that
\begin{equation} \label{eq:waveatomsdecoding}
d(m2^j+c) = \tilde{d}(c, j, m) + (-1)^{c} \lfloor \tfrac{m}2 \rfloor \cdot 2^{j} - I(c \text{ is odd})\cdot 2^{j}.
\end{equation}
Like in the case of Shannon wavelets, the decoding $\tilde{d}$ keeps us within the same block of size $2^j$. However, \eqref{eq:waveatomsdecoding} shows that some adjustment is required due to the change in periodicity. 
Then, each block of the matrix $F^A$ can be expressed as the product of three factors:
\begin{equation}
\label{F-block-decomp}
\begin{gathered}
(F^A(j,m))_{n, c} = \underbrace{e^{(-1)^c \i \pi \lfloor \tfrac{m}2 \rfloor} e^{-I(c\text{ is odd}) \i \pi} e^{(-1)^{c+1} \i \alpha_m}}_{r_m} \underbrace{e^{\i\cdot \pi 2^{-j} \tilde{d}(c, j, m)}}_{\text{Phase}} \underbrace{2^{-j/2} e^{\i \cdot 2 \pi 2^{-j} n \cdot \tilde{d}(c, j, m)}}_{\text{QFT}}.
\end{gathered}
\end{equation}
{The first two factors in Eq.~(\ref{F-block-decomp}) are new, compared to the analogous formula (\ref{equation:ShannonAdaptive}) for the Shannon wavelet transform.}

The first factor of \eqref{F-block-decomp} represents the necessary adjustment. 
Consider combination of values of $c$ and $m$: 
\begin{equation}
r_m(c) = \begin{cases}
e^{-\i \pi/4} \text{ if } m = c \bmod 2,\\
-e^{\i \pi/4} \text{ if } m \ne c \bmod 2.
\end{cases}
\end{equation}
For quantum realization, we would like to have a unitary corresponding different cases of $m$ and $c$. We may apply the following operators to the qubit corresponding to the least significant bit of $c$ which defines its parity.
\begin{equation}
R_0 = \begin{pmatrix}
e^{-\i\pi/4} & 0 \\
0 & -e^{\i \pi/4}
\end{pmatrix}, \qquad R_1 = R_0 \cdot (-i Z).
\end{equation}
We apply $R_0$ if $m$ is even and $R_1$ if $m$ is odd.
The complete procedure is 
implemented by the quantum circuit shown in Figure \ref{fig:waveatoms-f-circuit}. 
Additional details are presented in \ifthenelse{\boolean{originalformat}}{Algorithm \ref{algo:wave-atoms-fa} in Appendix \ref{appendix:alg-fa}}{the Supplemental Material}. 
The gate complexity of the algorithm is $O(L^2)$.

%% file: Wave-atoms/05-G-A.tex
\subsection{Quantum circuits for $G^A$}\label{section:blending}
Recall Definition \ref{def:matrix-g}, where each block of matrix $G^A$ is defined by equation \eqref{eq:block-g-orig}. 
It will be convenient to shift the indices of this matrix by 1. In other words, let us define the block-diagonal matrix $\tilde{G}^A \in \mathbb{C}^{N \times N}$ as follows: for $\mathrm{W}^j_m \ne \mathrm{W}^{j_l}_0$ and $|s| \le \mu_0(j, m)$, 
\begin{equation}\label{eq:block-g}
\begin{pmatrix}
\tilde{G}^A_{m2^j+2s, m2^j+2s} & \tilde{G}^A_{m2^j+2s, m2^j+2s+1}\\
\tilde{G}^A_{m2^j+2s+1, m2^j+2s} & \tilde{G}^A_{m2^j+2s+1, m2^j+2s+1}
\end{pmatrix} = \begin{pmatrix}
g_-[j, m, s] & g_+[j, m, s] \\
g_+[j, m, s]  & g_-[j, m, s] 
\end{pmatrix}.
\end{equation}
This satisfies
$$
G^A = \mathrm{Adder}_L(1) \tilde{G}^A \mathrm{Adder}_L(-1),
$$
where $\mathrm{Adder}_L(\pm 1)$ is the unitary operation defined in equation \eqref{eq: adder}. 

Furthermore, the definition of $\tilde{G}^A$ involves both positive and negative shifts $s$ from $m2^j$. The challenge with this approach is that when we work with wave packets in $\mathrm{W}^{j}_{m}$, it requires knowing information about the left neighbor of $\mathrm{W}^{j}_{m}$ in the tree. It would be more practical if we could define the operation $\tilde{G}^A$ in a way that involves only positive shifts (which we denote by $n$), and only one leaf node $\mathrm{W}^{j}_{m}$ at a time. This can be done in the following way.

We can re-write the formula \eqref{eq:block-g} for $\tilde{G}^A$ in terms of $n \in \{0, \ldots, 2^{j-1} - 1\}$ for $\mathrm{W}^j_m \in \Lambda_T$. Then for $\mathrm{W}^j_m \ne \mathrm{W}^{j_l}_0$ and $0~\le~n~\le~\mu_0(j, m)$, we have the same expressions as in \eqref{eq:block-g}, but substituting $s \mapsto n$:
\begin{equation}
\begin{gathered}
\tilde{g}^{(0)}_-[j, m, n] = g\left((-1)^m \pi (2^{-j+1}n - \tfrac12) \right),\\
\tilde{g}^{(0)}_+[j, m, n] = \i g\left((-1)^{m+1} \pi (2^{-j+1}n + \tfrac12) \right),\\
\begin{pmatrix}
\tilde{G}^A_{m 2^j + 2n, m 2^j + 2n} & \tilde{G}^A_{m 2^j + 2n, m 2^j + 2n + 1}\\
\tilde{G}^A_{m 2^j + 2n+1, m 2^j + 2n} & \tilde{G}^A_{m 2^j + 2n+1, m 2^j + 2n + 1}
\end{pmatrix} = \begin{pmatrix}
\tilde{g}^{(0)}_-[j, m, n] & \tilde{g}^{(0)}_+[j, m, n]\\
\tilde{g}^{(0)}_+[j, m, n] & \tilde{g}^{(0)}_-[j, m, n]
\end{pmatrix},
\end{gathered}
\end{equation}
and for $\mathrm{W}^j_m \ne \mathrm{W}^{j_r}_{m_r}$ and $\mu_0(j, m) < n < 2^{j-1}$, we have the same expressions as in \eqref{eq:block-g}, but substituting $(j,m) \mapsto (j,m+1)$ and $s \mapsto n - 2^{j-1}$:
\begin{equation}
\begin{gathered}
\tilde{g}^{(1)}_-[j, m, n] = g\left((-1)^m \pi (2^{-j+1}n - \tfrac12) \right),\\
\tilde{g}^{(1)}_+[j, m, n] =  -\i g\left((-1)^{m+1}\pi (2^{-j+1}n - \tfrac32) \right),\\
\begin{pmatrix}
\tilde{G}^A_{m 2^j + 2n, m 2^j + 2n} & \tilde{G}^A_{m 2^j + 2n, m 2^j + 2n + 1}\\
\tilde{G}^A_{m 2^j + 2n+1, m 2^j + 2n} & \tilde{G}^A_{m 2^j + 2n+1, m 2^j + 2n + 1}
\end{pmatrix} = \begin{pmatrix}
\tilde{g}^{(1)}_-[j, m, n] & \tilde{g}^{(1)}_+[j, m, n]\\
\tilde{g}^{(1)}_+[j, m, n] & \tilde{g}^{(1)}_-[j, m, n]
\end{pmatrix}.
\end{gathered}
\end{equation}

So far, the discussion of wave atoms did not rely on the exact form of function $g$. Each block of matrix $\tilde{G}^A$ depends on the value of $2^{-j+1}n$, and there are $O(2^{j_m})$ distinct such values, where $j_m$ is the maximum level of leaf nodes in the tree. Therefore there are $O(2^{j_m})$ different types of blocks in $\tilde{G}^A$. So, in the general case, the number of gates in the implementation of $\tilde{G}^A$ without any approximation is $O(2^{j_m})$. However, the complexity can be reduced for certain choices of the function $g$.

\subsubsection{Efficient quantum circuits for $\tilde{G}^A$}
\label{section:efficient-GA}

We give an example of a function $g$ for which $\tilde{G}^A$ can be implemented exactly and efficiently. Consider the following function $g$:
\begin{equation}\label{superfunctiong}
g(w) = \begin{cases}
\cos(\tfrac38 w - \tfrac{\pi}{16}), \text{ if } -\tfrac76 \pi \le w \le \tfrac{\pi}6,\\
\cos(\tfrac34 w - \tfrac{\pi}{8}), \text{ if } \tfrac\pi6 < w \le \tfrac{5}{6} \pi,\\
0, \text{ otherwise}.
\end{cases}
\end{equation}
This is a simple example of a function that resembles an asymmetric ``bump'' and satisfies the requirements in equation \eqref{prop:sum_of_squares} and equation \eqref{prop:change_of_sign}. Other types of functions satisfying these requirements also can be considered. 

For this choice of $g$, let
\begin{equation}\label{eq:vjmn}
v(j, m, n) = \begin{cases}
\dfrac{\pi}{2^{j+2}}\left( 3n - 2^{j-1} \right) \cdot 2^{I(n \le \mu_0(j, m))}, \text{ if } m \text{ is odd},\\
\dfrac{\pi}{2^{j+2}}\left( 3n - 2^j \right) \cdot 2^{I(n > \mu_0(j, m))}, \text{ if } m \text{ is even},
\end{cases}
\end{equation}
then
\begin{equation}\label{eq:s-pauli-s}
\begin{split}
\begin{pmatrix}
\tilde{G}^A_{m 2^j + 2n, m 2^j + 2n} & \tilde{G}^A_{m 2^j + 2n, m 2^j + 2n + 1}\\
\tilde{G}^A_{m 2^j + 2n+1, m 2^j + 2n} & \tilde{G}^A_{m 2^j + 2n+1, m 2^j + 2n + 1}
\end{pmatrix}  
&= \begin{pmatrix}
\cos(v(j, m, n)) & \i \sin(v(j, m, n))\\
\i \sin(v(j, m, n)) & \cos(v(j, m, n))\\
\end{pmatrix} \\
&= S^{\dagger} \underbrace{\begin{pmatrix}
\cos(v(j, m, n)) & \sin(v(j, m, n))\\
-\sin(v(j, m, n)) & \cos(v(j, m, n))\\
\end{pmatrix}}_{\text{Linear Pauli Rotations}} S,
\end{split}
\end{equation}
where $S$ denotes the single-qubit gate 
$S = \begin{pmatrix}
    1 & 0\\0&\i
\end{pmatrix}$.

The middle factor in (\ref{eq:s-pauli-s}) can be implemented, by using the formulas (\ref{eq:vjmn}), and applying a controlled single-qubit gate of the form 
\begin{equation}\label{eqn:linear-pauli-rotation}
\ket{x}\ket{y} \mapsto \ket{x} \Bigl( \cos(ax+b)\ket{y} + (-1)^y \sin(ax+b)\ket{y \oplus 1} \Bigr), \text{ for } y \in \{0, 1\},
\end{equation}
which we call a ``linear Pauli rotation'' with slope $2a$ and offset $2b$. This can be implemented using the quantum circuit shown in Figure \ref{fig:linear-pauli}. This procedure is described in detail in \ifthenelse{\boolean{originalformat}}{Appendix \ref{appendix:alg-ga} (Algorithm \ref{algo:wave-atoms-tilde-g-old})}{the Supplementary Material}, and has gate complexity of $O(L^2)$. 
\begin{figure}
    \centering
    \begin{minipage}[c]{8.5cm}
    \includegraphics[width=0.95\linewidth]{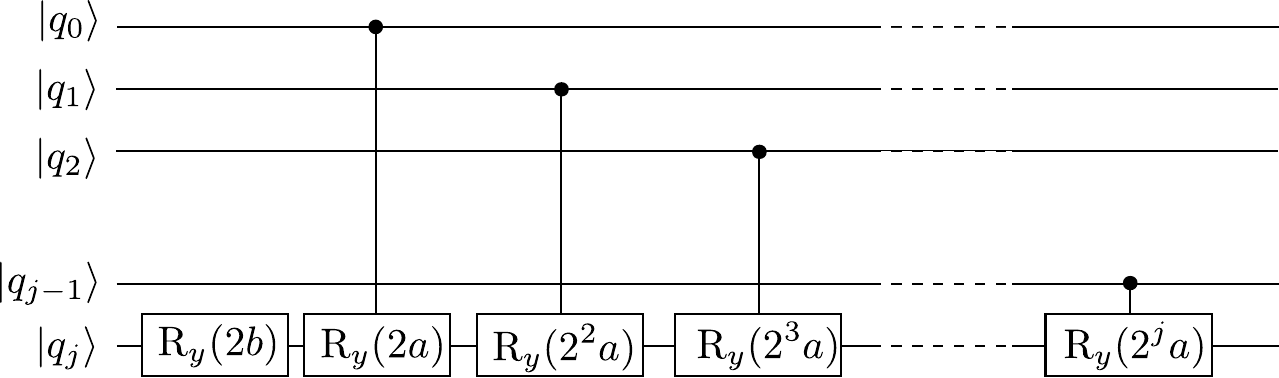}
    \end{minipage}
    \begin{minipage}[c]{4.0cm}
    $$
    \mathrm{R}_y(\theta) = \begin{pmatrix}
    \cos\left(\tfrac{\theta}2\right) & -\sin\left(\tfrac{\theta}2\right)\\
    \sin\left(\tfrac{\theta}2\right) & \cos\left(\tfrac{\theta}2\right)
    \end{pmatrix}
    $$
    \end{minipage}
    
    \caption{Quantum circuit that implements a linear Pauli rotation with slope $2a$ and offset $2b$, defined in equation (\ref{eqn:linear-pauli-rotation}).}
    \label{fig:linear-pauli}
\end{figure}


Finally, to implement $\tilde{G}^A$, we need to call the above procedure controlled by the predicates $I(q > 2\mu_0(j_l,0))$ and $I(q \le m_r2^{j_r}+2\mu_(j_r, m_r))$.
 

%% file: Wave-atoms/05-R.tex
\subsection{Quantum circuits for $R$}\label{section:permutations}
Recall that, by Definition \ref{def:permutation-rho}, the action of $R$ is described by a permutation $\rho$, and $\rho$  permutes odd indices as follows. For any $\mathrm{W}^j_m \ne \mathrm{W}^{j_l}_0$ and $s\in \mathbb{Z}$ such that $|s| \le \mu_0(j, m)$,
$$
\rho(m2^j+2s-1) = m2^j-2s-1.
$$
Let us exclude even indices from consideration. We identify the set of odd indices $\{1,3,5,\ldots,N-1\}$ with the set of all indices up to $N/2-1$, that is $\{0,1,2,\ldots,N/2-1\}$. We let $\dot{\rho}$ denote the permutation acting on $\{0,1,2,\ldots,N/2-1\}$ that corresponds to $\rho$. Note that $\dot{\rho}$ is related to $\rho$ as follows:
\begin{equation}
\rho(2k-1) = 2\dot{\rho}(k) - 1. \label{eq:rho-acute-rho}
\end{equation}

Given an implementation of $\dot{\rho}$, one can implement the operation $R$ (i.e., the permutation $\rho$) as follows: given that
$$
\rho(2k+1) = \rho(2(k+1)-1) = 2\dot{\rho}(k+1)-1,
$$
the implementation of $R$ consists of: a shift by 2; the circuit that implements $\dot{\rho}$ on qubits $1,\ldots,L-1$, controlled by the parity of qubit 0; and a shift by $-2$. This transforms the initial state of $\ket{2k+1}$ into $\ket{2(k+1)+1}$; then the state after applying $\dot{\rho}$ is $\ket{2\dot{\rho}(k+1)+1}$, which we need to transform into $\ket{2\dot{\rho}(k+1)-1}$.

\subsubsection{Applying the permutation $\dot{\rho}$}

It remains for us to show how to apply the permutation $\dot{\rho}$. For a fixed leaf node $\mathrm{W}^j_m$ such that $\mathrm{W}^j_m\ne\mathrm{W}^{j_l}_0$, let us denote the permutation associated with leaf node $\mathrm{W}^j_m$ by $\dot{\rho}^j_m$. Then $\dot{\rho}^j_m$ has non-trivial action as 
\begin{equation}    
\dot{\rho}^j_m(m2^{j-1}+s) = m2^{j-1}-s \text{ for } s \in \mathbb{Z} \text{ s.t. } |s| \le \mu_0(j,m). 
\end{equation}

Combining all permutations $\dot{\rho}^j_m$ together we get the permutation $\dot{\rho}$:
\begin{gather}
\dot{\rho} = \prod_{\mathrm{W}^j_m \ne \mathrm{W}^{j_l}_0} \dot{\rho}^j_m.
\end{gather}
The order in the product is not important since all permutations $\dot{\rho}^j_m$ generated by wave atom admissible tree $T$ pair-wise commute. This follows from the observation that no two such permutations have non-trivial action on the same index. One can notice that equation \eqref{eq:rho-acute-rho} does not work in case of $k=N/2$ or $k=0$. However, $\rho$ and $\dot{\rho}$ act trivially on the boundaries, so we do not consider such $k$'s.

Each $\dot{\rho}^j_m$ acts non-trivially on $2\mu_0(j, m)$ indices. Let us combine permutations that share the same value of $2\mu_0(j, m)$. Suppose this common value is $2\mu_0(j_*, 0)$, for some $j_*$. Then let $\dot{\rho}^{j_*}$ denote the combined permutation. By the definition of $\mu_0$ in equation~\eqref{prop:integer support},
\begin{equation}\label{eq:acute-rho-j}
\dot{\rho}^{j_*} = \prod_{\substack{\mathrm{W}^j_m \ne \mathrm{W}^{j_l}_0 \\ \mu_0(j, m) = \mu_0(j_*, 0)}} \dot{\rho}^j_m = \prod_{\substack{\mathrm{W}^j_m \ne \mathrm{W}^{j_l}_0 \\ j-I(m\text{ is odd}) = j_*}} \dot{\rho}^j_m.
\end{equation}

For permutation ${\rho}$, let us define a helper function $\tilde{h}^\rho: \{0, \ldots, N/2-1\} \mapsto \{0, \ldots, L-2\}$ that says, for each index, which permutation $\dot{\rho}^{j_*}$  acts nontrivially on it. Formally, we have: 
\begin{equation}\label{eq:tilde-h-rho}
\begin{gathered}
\forall \mathrm{W}^j_m \ne \mathrm{W}^{j_l}_0 \text{ and } |s| \le \mu_0(j, m),\ \tilde{h}^{\rho}(m2^{j-1}+s) = j - I(m \text{ is odd}),\\
\text{and}\\
\tilde{h}^\rho(k) = 0 \text{ if } \Bigl( k \le \mu_0(j_l, 0) \text{ or }k > m_r2^{j_r-1}+\mu_0(j_r,m_r) \Bigr).
\end{gathered}
\end{equation}
Furthermore, for $j \in \{1, \ldots, L-1\}$, define a family of boolean functions  $h^{\rho}_j: \{0, \ldots, N/2-1\} \mapsto \{0, 1\}$ such that 
\begin{equation}
h^{\rho}_j(k) = I(h^{\rho}(k)=j).
\end{equation}
\begin{algorithm}[t!]
\caption{Applying permutation $\dot{\rho}^{j_*}$, given values of helper function $h^\rho_{j_*}(\cdot)$}\label{algo:acute-rho}
\noindent Input: $\sum_{q=0}^{2^{L-1}-1} \alpha_q \ket{q} \ket{h_{j_*}^\rho(q)}$, that is, a superposition over indices $\ket{q}$, with labels $\ket{h_{j_*}^\rho(q)}$ computed by the helper function $h^\rho_{j_*}(\cdot)$.\\
Result: $\sum_{q=0}^{2^{L-1}-1} \alpha_q \ket{\dot{\rho}^{j_*}(q)}\ket{h_{j_*}^\rho(q)}$.\\
Notation:\\
\indent\hspace{0.25cm}Quantum register $\mathcal{R}_q$ of $L-1$ qubits; $\mathcal{R}_q[i]$ denotes the $i$-th qubit, $i=0, \ldots, L-1$;\\
\indent\hspace{0.25cm}Quantum register $\mathcal{R}_h$ of $1$ qubit.
\begin{enumerate}
\item Apply $\mathrm{Adder}_{L-1}(\mu_0(j_*, 0))$ to $\mathcal{R}_q$ controlled on $\mathcal{R}_h=\ket{1}$.
\item For each $i=0,\ldots,j_*-2$ apply $X$-gate to $\mathcal{R}_q[i]$ controlled on $\mathcal{R}_h=\ket{1}$.
\item Apply $\mathrm{Adder}_{L-1}(-\mu_1(j_*, 0))$ to $\mathcal{R}_q$ controlled on $\mathcal{R}_h=\ket{1}$.
\end{enumerate}
\end{algorithm}

With such boolean helper functions, the quantum implementation of $\dot{\rho}^{j_*}$ in equation \eqref{eq:acute-rho-j} becomes a three-stage process listed in Algorithm \ref{algo:acute-rho}. One can notice that $\dot{\rho}$ keeps indices within their respective intervals centered around $m2^{j-1}$ for some $\mathrm{W}^j_m$. This means that
\begin{equation}
\tilde{h}^\rho(k)=\tilde{h}^\rho(\dot{\rho}(k)) \text{ for any } k \in \{0, \ldots, N/2-1\},
\end{equation}
which implies that the calculation of $h_j^\rho$'s can be reversed after $\dot{\rho}$ is applied. 

Finally, we show a way to compute the helper functions $h^\rho_j(\cdot)$, using the helper functions $h_j(\cdot)$ defined in \eqref{eq:helper-functions}. We start by examining a leaf node $\mathrm{W}^j_m \not\in \{\mathrm{W}^{j_l}_0, \mathrm{W}^{j_r}_{m_r}\}$. Consider an index $m2^{j-1}+n$ where $n\in\{0, \ldots, N/2-1\}$. There are two cases:  
\begin{itemize}
    \item If $n \le \mu_0(j, m)$, then $\tilde{h}^\rho(m2^{j-1}+n) = j - I(m \text{ is odd})$ as $m2^{j-1}+n$ is affected by the permutation $\dot{\rho}^j_m$.
    \item If $n > \mu_0(j, m)$, then $\tilde{h}^\rho(m2^{j-1}+n) = j - I(m \text{ is even})$ as $m2^{j-1}+n$ is affected by the permutation $\dot{\rho}^{j'}_{m'}$ where $\mathrm{W}^{j'}_{m'}$ is the right neighbor of $\mathrm{W}^j_m$. 
\end{itemize}
For $\mathrm{W}^j_m \in \{\mathrm{W}^{j_l}_0, \mathrm{W}^{j_r}_{m_r}\}$, in one of the above cases $\tilde{h}^\rho(m2^{j-1}+n)=0$ according to the definition in equation \eqref{eq:tilde-h-rho}. Finally, one can evaluate these expressions, by obtaining the value of $j$ from the functions $h_j(\cdot)$. The above procedure is described in detail in \ifthenelse{\boolean{originalformat}}{Appendix \ref{appendix:alg-r} (Algorithm \ref{algo:h-rho})}{the Supplementary Material}. Its gate complexity is $O(L^2)$. 

Altogether, the complexity of implementation of the permutation $\dot{\rho}$ is $O(L^2)$, given that we need to call the above procedure two times to compute and uncompute $h^\rho_j$'s, and call $O(L)$ times Algorithm \ref{algo:acute-rho} (we care about $2 \le j < L$).

%% file: Wave-atoms/05-Complexity.tex
\subsection{Computational Complexity}\label{section: complexity}
The classical complexity of the 1-D wave atom transform is the same as the complexity of the fast Fourier transform $O(N \log N)$ where $N$ is the number of points in discretization \cite{DEMANET2007368}. 

The quantum implementation of the wave atom transform that combines all the steps in decomposition
$$
C^A = F^A R^{\dagger} G^A R,
$$
depends on
\begin{enumerate}
    \item Efficient computation of the family of boolean functions $h_j(\cdot), j=\{1, \ldots, L-1\}$ for tree $T$,
    \item Efficient computation of blocks of $G^A$.
\end{enumerate}
Here, efficient computation means polynomial time in the number of qubits, in other words, polylogarithmic in $N$. The results are summarized in the following theorem.
\begin{theorem}[Quantum complexity of 1D wave atom transform]\label{theorem:complexity-general}
Take a wave atom admissible tree $T$ defined for discretization of size $O(2^L)$. If 
\begin{itemize}
\item function $g$ of the wave atom transform allows the implementation of matrix $G^A$ with gate complexity of $O(L^2)$,
\item the family of functions $h_j(\cdot)$ has a quantum implementation that requires $O(L^2)$ gates,
\end{itemize}
then the gate complexity of quantum wave atom transform generated by tree $T$ and function $g$ is $O(L^2)$.
\end{theorem}

In Section \ref{section:blending}, we provided a concrete implementation of the matrix $G^A$ for a specific function $g$. In addition, recall from Section \ref{section: wave packet encoding} that the family of functions $h_j$ can be efficiently implemented for monotonic trees. This implies the following Corollary.
\begin{corollary}[Quantum complexity of 1D wave atom transform for monotonic trees]\label{theorem:complexity-specific}
Let $T$ be a monotonic wave atom admissible tree defined for discretization of size $O(2^L)$. If $T$ is represented by the leftmost node on each level, then the gate complexity of the quantum wave atom transform generated by the tree $T$ and the function $g$ defined in \eqref{superfunctiong} is $O(L^2)$. The complexity of the construction of the quantum circuit is $O(L^2)$.
\end{corollary}
{Finally, note that the above result applies to parabolic scaling trees as a special case, as shown in Section \ref{section: Wave packet trees}.}

%% file: 06-Conclusion.tex
\ifthenelse{\boolean{originalformat}}{
 \section{Discussion}
}{
 \section{Discussion and Conclusion}
}
\label{section: conclusion}
The first result of this paper is a quantum implementation of Shannon wavelet transform with a hierarchical tree structure of shifts and scales (Theorem \ref{theorem:algorithms-shannon}). This construction made use of helper decoding functions $h_j(\cdot)$ (equation (\ref{eq:helper-functions})) that map indices to their respective levels in the wave packet tree. With such functions, the hierarchical quantum Shannon wavelet transform becomes a product of Fourier transforms of different scales. The type of wave packet trees dictates the circuit implementation of functions $h_j$. We demonstrated that for monotonic trees, the implementation of helper decoding functions is efficient and requires $O(L^2)$ gates, when implementing a wavelet transform of dimension $2^L$. Together, this allows us to achieve exponential speedup with the quantum gate complexity of $O(L^2)$ for the entire procedure of Shannon wavelet transform.

The wave atom transform with a hierarchical tree structure is more intricate as it requires ``blending,'' similar to \cite{ni2024quantumwavepackettransforms}, due to the overlaps between the supports of neighboring wave packets in frequency space \cite{DEMANET2007368, Villemoes2002WaveletPW}. A crucial contribution of this paper is to show how to implement blending that \textit{also follows the tree structure}, on a quantum computer, i.e., acting on quantum superposition states. Using this technique, we showed that for monotonic admissible trees with wave atoms satisfying some technical conditions, the quantum wave atom transform also achieves an exponential speedup with the gate complexity of $O(L^2)$ (Theorem \ref{theorem:complexity-general}). We demonstrated this exponential speedup for wave atoms constructed from a particular function $g(\cdot)$ (Corollary \ref{theorem:complexity-specific}). In particular, this family of wave atoms has ``parabolic scaling,'' which is useful for constructing sparse representations of solutions to wave equations \cite{DEMANET2007368}. 

Generalizing this construction is an interesting topic for future work.
In particular, can this construction be carried out efficiently for different choices of the function $g(\cdot)$ (besides the choice described in Section \ref{section:efficient-GA}), perhaps by adapting the techniques introduced recently by \cite{Bagherimehrab_2024}?

Using tree-structured wavelet bases, one has additional flexibility to construct wavelet bases with desirable properties. This suggests considering other wave packet trees that satisfy our admissibility conditions (Definition \ref{def:wave-atom-admissible-tree}). 
However, one is unlikely to find a single ``universal'' quantum algorithm that is efficient and applicable to all admissible trees. The general tree representation requires at least $O(2^L)$ bits as there are more than $2^{2^{L-1}}$ different admissible wave packet trees \cite{WaveletTourofSignalProcessing}. If there is a classical step, the decoding of the tree representation is already $\Omega(2^L)$.

Another potential future direction involves higher-dimensional wave packet transforms, which could help solve wave equations and systems of hyperbolic differential equations in general, as its basis offers a sparse representation of the solution operator \cite{Candes2003-vd,candes2004curveletrepresentationwavepropagators, DemanetThesis,Demanet2009-zt}. This paper shows a 1-dimensional wave atom transform, which can be extended to construct a multi-dimensional transform. For example, taking products of 1-D wavelet packets allows for constructing two-dimensional orthonormal basis functions with four bumps in the frequency plane. Although the construction is not a straightforward tensor product since both dimensions share the same scale \cite{DemanetThesis}, the results in 1-D are easily extendable to the 2-D case.

\vskip 11pt

\input{Introduction/01-Acknowledgements}

%% file: Introduction/01-Acknowledgements.tex
\textbf{Acknowledgements:} 
{We thank the anonymous referees for several suggestions that improved the presentation of our results.} MP is grateful to Radu Balan for the discussion related to this paper, which restored belief in the best. 

This work is supported by a collaboration between the US DOE and other Agencies. This material is based upon work supported by the U.S. Department of Energy, Office of Science, Accelerated Research in Quantum Computing, Fundamental Algorithmic Research toward Quantum Utility (FAR-Qu). Additional support is acknowledged from NSF and NIST. 

MP also acknowledges funding and support from a NIST/UMD Joint Center for Quantum Information and Computer Science
(QuICS) Lanczos Graduate Fellowship, and a MathQuantum Graduate Fellowship (NSF award DMS-2231533).

Contributions by NIST employees are not subject to US copyright protection.

%% file: Appendix/00.tex



\input{Appendix/Decoding}


\ifthenelse{\boolean{originalformat}}{
 \section{Shannon wavelets (see section \ref{section:quantum-shannon-wavelet-transform})}
}{
 \section{Shannon wavelets}
}
\label{appendix:algorithms-shannon}
\input{Algorithms/Shannon-wavelet}

\section{Wave atoms}
\label{appendix:algorithms-wave-atoms}
\input{Algorithms/F-A}
\input{Algorithms/Linear-Pauli}
\input{Algorithms/R}

%% file: Appendix/Decoding.tex
\ifthenelse{\boolean{originalformat}}{
 \section{Encoding and decoding circuits (see section \ref{section:frequency-encoding})}
}{
 \section{Encoding and decoding circuits}
}
\label{section:decoding-circuit-proof}

\begin{algorithm}
\caption{QFT and frequency encoding}\label{algo:encoding}
\noindent Input: $\sum_{q=0}^{2^L-1} f[q] \ket{q}$.\\
\noindent Result: $\sum_{q=0}^{2^L-1} \hat{f}[d(q)]\ket{q}$.
\begin{enumerate}
    \item Apply $Q_L^{\dagger}$ (equivalent to DFT). This produces:
    $\sum_{q=0}^{2^L-1} \hat{f}[q] \ket{q_{L-1}} \ket{q_{L-2}} \ldots \ket{q_0}$
    \item For each $i=0, \ldots, L-2$, apply controlled-$X$ gate to qubit $i$ controlled on qubit $L-1$. This produces the state:
    $$
    \sum_{q=0}^{2^L-1} \hat{f}[q] \ket{q_{L-1}} \ket{q_{L-2} \oplus q_{L-1}} \ldots \ket{q_0 \oplus q_{L-1}}.
    $$
    \item Cyclic shift of $L$ qubits implemented by SWAP gates. This produces the state:
    $$
    \sum_{q=0}^{2^L-1} \hat{f}[q] \ket{q_{L-2} \oplus q_{L-1}} \ldots \ket{q_0 \oplus q_{L-1}} \ket{q_{L-1}}
    $$
\end{enumerate}
\end{algorithm}

\begin{algorithm}
\caption{Retaining decoding circuit $\tilde{D}_j$}\label{algo:decode-j}
\noindent Input: $\ket{m_0} \sum_{q=0}^{2^j-1} \alpha_q \ket{q}$ where $m_0\in\{0, 1\}$.\\
\noindent Result: $\ket{m_0} \sum_{q=0}^{2^j-1} \alpha_q |\tilde{d}(q, j, m_0)\rangle$.
\begin{enumerate}
\item For each qubit $i=1, \ldots, j-1$, apply $X$ gate to qubit $i$ controlled on qubit $0$. This produces the state:
    $$
    \ket{m_0} \sum_{q=0}^{2^j-1} \alpha_q \ket{q_{j-1} \oplus q_0 } \ldots \ket{q_1  \oplus q_0} \ket{q_0}.
    $$
\item Cyclic shift of $j$ qubits implemented by SWAP gates. This produces the state:
    $$
    \ket{m_0} \sum_{q=0}^{2^j-1} \alpha_q \ket{q_0} \ket{q_{j-1} \oplus q_0 } \ldots \ket{q_1  \oplus q_0}.
    $$
\item Apply $X$ gate to qubit $j-1$ controlled on qubit $j$. This produces the state:
    $$
    \ket{m_0} \sum_{q=0}^{2^j-1} \alpha_q \ket{q_0 \oplus m_0} \ket{q_{j-1} \oplus q_0 } \ldots \ket{q_1  \oplus q_0}.
    $$
\end{enumerate}
\end{algorithm}

Frequency encoding (after applying the DFT) is carried out by Algorithm \ref{algo:encoding}; decoding is performed by reversing the steps. The retaining decoding procedure $\tilde{D}_j$ is shown in Algorithm \ref{algo:decode-j}.

\begin{proof}[Proof of correctness of circuit $\tilde{D}_j$ (Algorithm \ref{algo:decode-j}):]\, \\
Recall the definition of $\tilde{d}$ in equation \ref{local-decoding} and write it in terms of $q$
\begin{equation*}
    \tilde{d}(q, j, m) = \begin{cases}
    2^{j-1} \cdot I (m \text{ is odd}) + \tfrac{q}2, \text{ if } q \text{ is even},\\
    2^{j} - 2^{j-1} \cdot I (m \text{ is odd}) - \tfrac{q+1}{2}, \text{ if } q \text{ is odd}.
    \end{cases}
\end{equation*}
In binary representation $q=q_{j-1}, q_{j-2}, \ldots, q_0$. For even $q
$ ($q_0 = 0$), 
$$
q/2 = 0, q_{j-1}, \ldots, q_1.
$$
For odd $q$ ($q_0 = 1$),
$$
2^j - (q+1)/2 = 1, (q_{j-1} \oplus 1), \ldots, (q_1 \oplus 1).
$$
Combined:
$$
\tilde{d}(q, j, m) = \tilde{d}(q_{j-1}, q_{j-2}, \ldots, q_0; j; m) = (q_0 \oplus m), (q_{j-1} \oplus q_0), \ldots, (q_1 \oplus q_0).
$$
The description of the algorithm has the evolution of each basis state. The final state is 

$\quad\ 
\sum_{q=0}^{2^j-1}\alpha_q \ket{m}\ket{q_0 \oplus m}\ket{q_{j-1} \oplus q_0} \ldots \ket{q_1 \oplus q_0} = \sum_{q=0}^{2^j-1}\alpha_q \ket{m} \ket{\tilde{d}(q, j, m)}.
$
\end{proof}

%% file: Algorithms/Shannon-wavelet.tex
\begin{algorithm}[H]
\caption{Quantum Shannon wavelet transform}\label{algo:shannon-wavelet}
\noindent Input: $|\hat{f}\rangle = \sum_{q=0}^{2^L-1} \hat{f}[d(q)] \ket{q}$.\\
Result: $\ket{c^D} = \sum_{W^j_m\in \Lambda_T} \sum_{n=0}^{2^j-1} \left( \sum_{q=0}^{2^L - 1} \overline{\hat{\varphi}^j_{m, n}(d(q))} \hat{f}[d(q)] \right) \ket{m2^j + n}$.\\
Notation:
\begin{itemize}
\item Quantum register $\mathcal{R}_q$ of $L$ qubits contains the input state, $\mathcal{R}_q[i]$ corresponds to $i$-th qubit in register $\mathcal{R}_q$, $i=0, \ldots, L-1$
\item Quantum register $\mathcal{R}_h$ of $L-1$ qubits contains ancilla qubits in zero initial state with indices starting from~$1$, i.e., $\mathcal{R}_h[i]$ is initialized to the state $\ket{0}$, $i=1, \ldots, L-1$
\end{itemize}
\begin{enumerate}
    \item Compute $h_j$, $1 \le j < L$, by calling subroutines such that for basis state $\ket{q}$ 
    
    \hspace{1cm} $\underbrace{\ket{q}}_{\mathcal{R}_q}\underbrace{\ket{0}}_{\mathcal{R}_h} \mapsto \ket{q}\ket{h_1(q) h_2(q) \ldots h_{L-1}(q)},$
    
    \noindent in other words, qubit $\mathcal{R}_h[i]$ is in state $\ket{h_i(q)}$.
    \item For each $j=1, \ldots, L-1$ apply decoding circuit $\tilde{D}_j$ to qubits $\mathcal{R}_q[0], \ldots, \mathcal{R}_q[j]$ controlled on $\mathcal{R}_h[j]=\ket{1}$ and $\mathcal{R}_h[j+1]=\ket{0}$ (for $j=L-1$, the latter condition is omitted and $\tilde{D}_j$ is applied to $\mathcal{R}_q[0], \ldots, \mathcal{R}_q[j-1]$).
    \item For each integer $j$ starting from $j=L$ till $j=1$, apply QFT rotation $Q_j^R$ to qubits $\mathcal{R}_q[0], \ldots, \mathcal{R}_q[j-1]$ controlled on $\mathcal{R}_h[j]=\ket{1}$.
    \item For each $j=2, \ldots, L$ apply QFT swap $Q_j^S$ to qubits $\mathcal{R}_q[0],\ldots, \mathcal{R}_q[j-1]$ controlled on $\mathcal{R}_h[j]=\ket{1}$  and $\mathcal{R}_h[j+1]=\ket{0}$ (the latter condition is omitted for $j=L-1$).
    \item Undo the computation of step 1, returning the register $\mathcal{R}_h$ to the state $\ket{0}$. This is feasible because the action of steps 2-4 on the register $\mathcal{R}_q$ is block-diagonal, i.e., it preserves each subspace $\text{span}\{\ket{q} \;|\; \tilde{h}(q)=j\}$.
\end{enumerate}
\end{algorithm}

\begin{proof}
Let us consider the evolution of $|\hat{f}\rangle$ within the subspace $\text{span}\{\ket{q} \;|\; \tilde{h}(q)=j\}$, in other words, the evolution of the following linear combination of the basis states
\begin{small}
$$
\sum_{\tilde{h}(q)=j} \hat{f}[d(q)]\ket{q_{L-1}}\ldots\ket{q_0}\ket{0}.
$$
\end{small}
After the first step, the state is 
\begin{small}
$$
\sum_{\tilde{h}(q)=j} \hat{f}[d(q)]\ket{q_{L-1}}\ldots\ket{q_0}|\underbrace{1 \ldots 1}_{j}0\ldots 0\rangle.
$$
\end{small}
After the second step, the state is 
\begin{small}
$$
\sum_{\tilde{h}(q)=j} \hat{f}[d(q)]\ket{q_{L-1}\ldots q_{j}} \ket{\tilde{d}(q_{j-1}\ldots q_0, j, q_j)}|1 \ldots 10\ldots 0\rangle.
$$
\end{small}
Steps 3-4 apply QFT of $j$-qubits to the state, which results in
\begin{small}
$$
\sum_{\tilde{h}(q)=j} \ket{q_{L-1}\ldots q_j} 
\left(2^{-j/2}\sum_{n=0}^{2^j-1}\hat{f}[d(q)] e^{\i2^{-j} n \cdot \tilde{d}(q_{j-1}\ldots q_0, j, q_j)} \ket{n} \right)
|1 \ldots 10\ldots 0\rangle.
$$
\end{small}
After the last step
\begin{small}
\begin{gather*}
\sum_{\tilde{h}(q)=j} \ket{q_{L-1}\ldots q_{j+1}} 
\left(2^{-j/2}\sum_{n=0}^{2^j-1}\hat{f}[d(q)] e^{\i2\pi 2^{-j} n \cdot \tilde{d}(q_{j-1}\ldots q_0, j, q_j)} \ket{n} \right)
|0\rangle = (\star_1)
\end{gather*}
\end{small}
if we denote $q_{L-1}\ldots q_{j}$ as $m$, and $q_{j-1} \ldots q_0$ as $\tilde{q}$ then
\begin{small}
$$
(\star_1) = \sum_{m\,s.t.\, \mathrm{W}^j_m\in \Lambda_T} \ket{m} \sum_{\tilde{q}=0}^{2^j-1} \left(2^{-j/2}\sum_{n=0}^{2^j-1}\hat{f}[d(m2^j+\tilde{q})] e^{\i2\pi 2^{-j} n \cdot \tilde{d}(\tilde{q}, j, q_{j+1})} \ket{n} \right) \ket{0} = (\star_2)
$$
\end{small}
However, for fixed $j$ and $n$, $\tilde{d}(\tilde{q}, j, q_j) = d(m2^j+\tilde{q})$, and given the support of $\hat{\varphi}^j_{m, n}$,
\begin{small}
$$
\begin{gathered}
\sum_{q=0}^{2^L-1} \overline{\hat{\varphi}^j_{m, n}(d(q))}\hat{f}[d(q)] = \sum_{\tilde{q}=0}^{2^j-1} \overline{\hat{\varphi}^j_{m, n}(2\pi d(m2^j+\tilde{q}))}\hat{f}[d(m2^j+\tilde{q})] 
\\= 2^{-j/2}\sum_{\tilde{q}=0}^{2^j-1}\hat{f}[d(m2^j+\tilde{q})] e^{\i2\pi 2^{-j} n \cdot \tilde{d}(\tilde{q}, j, q_{j+1})}
\end{gathered}
$$
\end{small}
Together,
\begin{small}
$$
(\star_2) = \sum_{m\,s.t.\, \mathrm{W}^j_m\in \Lambda_T} \sum_{n=0}^{2^j-1} \left(\sum_{q=0}^{2^L-1} \overline{\hat{\varphi}^j_{m, n}(d(q))}\hat{f}[d(q)]\right) |m2^j+n\rangle. 
$$
\end{small}
The sum over all possible $j$'s gives the evolution claimed by the algorithm.
\end{proof}

%% file: Algorithms/F-A.tex
\ifthenelse{\boolean{originalformat}}{
 \subsection{Quantum circuits for $F^A$ (see section \ref{section:block-diagonal-fourier})}
}{
 \subsection{Quantum circuits for $F^A$}
}
\label{appendix:alg-fa}

\begin{algorithm}[H]
\caption{Quantum wave atoms: Operation $F^A$}\label{algo:wave-atoms-fa}
\noindent Input: $|\hat{f}\rangle = \sum_{q=0}^{2^L-1} \hat{f}[d(q)] \ket{q}$.\\
Result: $\ket{f^D} = \sum_{W^j_m\in \Lambda_T} \sum_{n=0}^{2^j-1} \left( \sum_{q=0}^{2^L - 1} F^A_{m2^j + n, q} \hat{f}[d(q)] \right) \ket{m2^j + n}$.\\
Notation:
\begin{itemize}
\item Quantum register $\mathcal{R}_q$ of $L$ qubits contains the input state; $\mathcal{R}_q[i]$ corresponds to $i$-th qubit in register $\mathcal{R}_q$, $i=0, \ldots, L-1$;
\item Quantum register $\mathcal{R}_h$ of $L-1$ qubits contains ancilla qubits in zero initial state with indices starting from~$1$, i.e. $\mathcal{R}_h[i]$ is initialized in the state $\ket{0}$, $i=1, \ldots, L-1$.
\end{itemize}
\begin{enumerate}
    \item Apply gate $R_0$ to qubit $\mathcal{R}_q[0]$.
    \item Compute $h_j$, $1 \le j < L$, by calling subroutines such that for basis state $\ket{q}$ 
    
    \hspace{1cm}$
    \underbrace{\ket{q}}_{\mathcal{R}_q}\underbrace{\ket{0}}_{\mathcal{R}_h} \mapsto \ket{q}\ket{h_1(q) h_2(q) \ldots h_{L-1}(q)},
    $
    
    \noindent in other words, qubit $\mathcal{R}_h[i]$ is in state $\ket{h_i(q)}$.
    \item For each $j=1, \ldots, L-1$ apply $(-iZ)$-gate to qubit $\mathcal{R}_q[0]$ controlled on $\mathcal{R}_q[j]=\ket{1}$, $\mathcal{R}_h[j]=\ket{1}$ and $\mathcal{R}_h[j+1]=\ket{0}$ (for $j=L-1$, the latter condition is omitted).
    \item For each $j=1, \ldots, L-1$ apply decoding circuit $\tilde{D}_j$ to qubits $\mathcal{R}_q[0], \ldots, \mathcal{R}_q[j]$ controlled on $\mathcal{R}_h[j]=\ket{1}$ and $\mathcal{R}_h[j+1]=\ket{0}$ (for $j=L-1$, the latter condition is omitted and $\tilde{D}_j$ is applied to $\mathcal{R}_q[0], \ldots, \mathcal{R}_q[j-1]$).
    \item For each $j=1, \ldots, L-1$, and each $i=0, \ldots, j-1$, apply $P_{j-i+1}$ gate to $\mathcal{R}_q[i]$ controlled on $\mathcal{R}_h[j]=\ket{1}$ and $\mathcal{R}_h[j+1]=\ket{0}$.
    \item For each integer $j$ starting from $j=L$ to $j=1$, apply QFT rotation $Q_j^R$ to qubits $\mathcal{R}_q[0], \ldots, \mathcal{R}_q[j-1]$ controlled on qubit $\mathcal{R}_h[j]=\ket{1}$.
    \item For each $j=2, \ldots, L$ apply QFT swap $Q_j^S$ to qubits $\mathcal{R}_q[0],\ldots, \mathcal{R}_q[j-1]$ controlled on $\mathcal{R}_h[j]=\ket{1}$ and $\mathcal{R}_h[j+1]=\ket{0}$ (the latter condition is omitted for $j=L-1$).
    \item Undo the computation of step 2, returning the register $\mathcal{R}_h$ to the state $\ket{0}$. This is feasible because the action of steps 3-7 on the register $\mathcal{R}_q$ is block-diagonal, i.e., it preserves each subspace $\text{span}\{\ket{q} \;|\; \tilde{h}(q)=j\}$.
\end{enumerate}
\end{algorithm}

\begin{proof}[Proof of correctness of Algorithm \ref{algo:wave-atoms-fa}]
Let's consider the evolution of $|\hat{f}\rangle$ within the subspace $\text{span}\{\ket{q} \;|\; \tilde{h}(q)=j\}$, 
$$
\sum_{\tilde{h}(q) = j} \hat{f}[d(q)]\ket{q}\ket{0} = \sum_{\tilde{h}(q)=j} \hat{f}[d(q)] \ket{q_{L-1}}\ldots \ket{q_0} \ket{0}.
$$
After the first step, the state is 
$$
\sum_{\tilde{h}(q)=j} (-1)^{q_0} e^{(-1)^{q_0+1}\i\pi/4}\hat{f}[d(q)]\ket{q}\ket{0}.
$$
After the second step, the state is 
$$
\sum_{\tilde{h}(q)=j} (-1)^{q_0} e^{(-1)^{q_0+1}\i\pi/4}\hat{f}[d(q)]\ket{q}|\underbrace{1\ldots1}_j0\ldots \rangle.
$$
After the third step, the state is 
$$
\sum_{\tilde{h}(q)=j} (-1)^{q_0 \lor q_j} \i^{q_j} e^{(-1)^{q_0+1}\i\pi/4}\hat{f}[d(q)]\ket{q}\ket{1\ldots10\ldots 0}.
$$
After the fourth step, the state is 
$$
\begin{gathered}
\sum_{\tilde{h}(q)=j} (-1)^{q_0 \lor q_j} \i^{q_j} e^{(-1)^{q_0+1}\i\pi/4}\hat{f}[d(q)]\ket{q_{L-1}}\ldots \ket{q_j} \otimes |\tilde{d}(q_{j-1}\ldots q_0, k, q_j\rangle \ket{1\ldots10\ldots 0}.
\end{gathered}
$$
After the fifth step, the state is 
$$
\begin{gathered}
\sum_{\tilde{h}(q)=j} (-1)^{q_0 \lor q_j} \i^{q_j} e^{(-1)^{q_0+1}\i\pi/4}\hat{f}[d(q)]\ket{q_{L-1}}\ldots \ket{q_j} 
\\ \otimes  e^{\i \cdot \pi 2^{-j}\tilde{d}(q_{j-1}\ldots q_0, k, q_j)}|\tilde{d}(q_{j-1}\ldots q_0, k, q_j\rangle \ket{1\ldots10\ldots 0},
\end{gathered}
$$
After the steps $6-8$, the state is 
$$
\begin{gathered}
\sum_{\tilde{h}(q)=j} (-1)^{q_0 \lor q_j} \i^{q_j} e^{(-1)^{q_0+1}\i\pi/4}\hat{f}[d(q)]\ket{q_{L-1}}\ldots \ket{q_j}  
\\ \otimes  e^{\i \cdot \pi 2^{-j}\tilde{d}(q_{j-1}\ldots q_0, k, q_j)}\left(2^{-j/2}\sum_{n=0}^{2^j-1} e^{\i2^{-j}n\cdot \tilde{d}(q_{j-1}\ldots q_0, j, q_j)} \ket{n} \right) \ket{0},
\end{gathered}
$$
if we denote $q_{L-1}\ldots q_j$ as $m$, and $q_{j-1} \ldots q_0$ as $\tilde{q}$ then we can re-write the final state as 
$$
(\star_1) = \sum_{m\,s.t.\, \mathrm{W}^j_m\in \Lambda_T} \sum_{n=0}^{2^j-1} \sum_{\tilde{q}=0}^{2^j-1} F^A_{m2^j + n, m2^j + \tilde{q}} \hat{f}[d(m2^j+\tilde{q})] |m2^j+n\rangle\ket{0} = (\star)
$$
However, $F^A_{m2^j+n, q} = 0$ for $q \not\in \{m2^j, \ldots, (m+1)2^j-1\}$, so 
$$
(\star) = \sum_{m\,s.t.\, \mathrm{W}^j_m\in \Lambda_T} \sum_{n=0}^{2^j-1} \sum_{q=0}^{2^L-1} F^A_{m2^j + n, q} \hat{f}[d(q)] |m2^j+n\rangle \ket{0}.
$$
The sum over all possible $j$'s gives the evolution claimed by the algorithm.
\end{proof}

%% file: Algorithms/Linear-Pauli.tex
\ifthenelse{\boolean{originalformat}}{
 \subsection{Efficient quantum circuits for $\tilde{G}^A$ (see section \ref{section:efficient-GA})}
}{
 \subsection{Efficient quantum circuits for $\tilde{G}^A$}
}
\label{appendix:alg-ga}

\begin{algorithm}[H]\caption{Operation $\tilde{G}^A$, using linear Pauli rotations}
\label{algo:wave-atoms-tilde-g-old}
\noindent Input: 
$$
\ket{\alpha} = \sum_{\mathrm{W}^j_m\in \Lambda_T} \smashoperator{\sum_{n=0}^{2^{j-1}-1}} \ket{m2^{j-1}+n}\left(\alpha_{m2^j+2n} \ket{0} + \alpha_{m2^j+2n+1} \ket{1} \right)
$$

\noindent Output: 
\begin{multline*}
\tilde{G}^A\ket{\alpha} = \sum_{\mathrm{W}^j_m\in \Lambda_T} \smashoperator{\sum_{n=0}^{2^{j-1}-1}} |m2^{j-1}+n\rangle \otimes 
\begin{pmatrix}
\cos(v(j, m, n)) & \i \sin(v(j, m, n))\\
\i \sin(v(j,m,n)) & \cos(v(j,m,n))
\end{pmatrix} \left(\alpha_{m2^j+2n} \ket{0} + \alpha_{m2^j+2n+1} \ket{1} \right).
\end{multline*}
\noindent Notation:
\begin{itemize}
\item Quantum register $\mathcal{R}_q$ of $L$ qubits contains the input state; $\mathcal{R}_q[i]$ corresponds to $i$-th qubit in register $\mathcal{R}_q$, $i=0,\ldots,L-1$
\item Quantum register $\mathcal{R}_h$ of $L-1$ qubits contains ancilla qubits in zero initial state with indices starting from~$1$, i.e., $\mathcal{R}_h[i]$ is initialized in the state $\ket{0}$, $i=1, \ldots, L-1$
\item Quantum register $\mathcal{R}_I$ of $1$ qubit contains an ancilla qubit in zero initial state.
\end{itemize}
\begin{enumerate}
    \item Apply $S$-gate to $\mathcal{R}_q[0]$.
    \item Compute $h_j$, $1 \le j < L$, by calling subroutines such that for basis state $\ket{q}$
    \begin{small}
    $$
    \underbrace{\ket{q}}_{\mathcal{R}_q} \underbrace{\ket{0}}_{\mathcal{R}_h}\underbrace{\ket{0}}_{\mathcal{R}_I} \mapsto \ket{q}\ket{h_1(q) \ldots h_{L-1}(q)}\ket{0},
    $$
    \end{small}
    in other words, qubit $\mathcal{R}_h[i]$ is in the state $\ket{h_i(q)}$.
    \item For each $j=1,\ldots, L-1$ controlled on $\mathcal{R}_h[j]=\ket{1}$, $\mathcal{R}_h[j+1]=\ket{0}$:
    \begin{enumerate}
        \item Compute $\mathrm{Comparator}_{j-1}(\mu_0(j, 0) + 1)$ applied to qubits $\mathcal{R}_q[1], \ldots, \mathcal{R}_q[j-1]$ with the result added to $\mathcal{R}_I$ ancilla qubit.
        \item Apply linear Pauli rotation (see Figure \ref{fig:linear-pauli}) with slope of $3\pi 2^{-j}$ and offset of $-\pi$ controlled on $\mathcal{R}_I=\ket{1}$ and $\mathcal{R}_q[j]=\ket{0}$.
        \item Apply linear Pauli rotation with slope of $3\pi 2^{-(j+1)}$ and offset of $-\tfrac{\pi}2$ controlled on $\mathcal{R}_I=\ket{0}$ and $\mathcal{R}_q[j]=\ket{0}$.
        \item Undo the computation of step 3a, returning the register $\mathcal{R}_I$ to the state $\ket{0}$. This is feasible because steps 3b-3c are limited to action on qubit $\mathcal{R}_q[0]$.
        \item Compute $\mathrm{Comparator}_{j-1}(\mu_0(j, 1) + 1)$ applied to qubits $\mathcal{R}_q[1], \ldots, \mathcal{R}_q[j-1]$ with the result added to $\mathcal{R}_I$ ancilla qubit.
        \item Apply linear Pauli rotation with slope of $3\pi 2^{-(j+1)}$ and offset of $-\tfrac{\pi}4$ controlled on  $\mathcal{R}_I=\ket{1}$ and $\mathcal{R}_q[j]=\ket{1}$.
        \item Apply linear Pauli rotation with slope of $3\pi 2^{-j}$ and offset of $-\tfrac{\pi}2$ controlled on that $\mathcal{R}_I=\ket{0}$ and $\mathcal{R}_q[j]=\ket{1}$.
        \item Undo the computation of step 3e, returning the register $\mathcal{R}_I$ to the state $\ket{0}$. This is feasible because steps 3f-3g are limited to action on qubit $\mathcal{R}_q[0]$.
    \end{enumerate}
    \item Undo the computation of step 2, returning the register $\mathcal{R}_h$ to the state $\ket{0}$. This is feasible because the action of step 3 on the register $\mathcal{R}_q$ is block-diagonal, i.e., it preserves each subspace $\text{span}\{\ket{q} \;|\; \tilde{h}(q)=j\}$.
    \item Apply inverse $S$-gate to $\mathcal{R}_q[0]$.
\end{enumerate}
\end{algorithm}

\begin{proof}[Proof of correctness of Algorithm \ref{algo:wave-atoms-tilde-g-old}]
For the proof, let us denote 
$$
R(\theta) = \begin{pmatrix}
\cos(\theta) & \sin(\theta)\\
-\sin(\theta) & \cos(\theta)
\end{pmatrix}.
$$
Let's consider the evolution within the subspace $\text{span}\{\ket{q} \;|\; \tilde{h}(q)=j\}$ for a fixed $j$, in other words,
$$
\sum_{m\,s.t.\, \mathrm{W}^j_m\in \Lambda_T} \smashoperator{\sum_{n=0}^{2^{j-1}-1}} |m2^{j-1}+n\rangle(\alpha_{m2^j + 2n} \ket{0} + \alpha_{m2^j+2n+1}\ket{1}) \ket{0}\ket{0}.
$$
After the first step, the state is
$$
\sum_{m\,s.t.\, \mathrm{W}^j_m\in \Lambda_T} \smashoperator{\sum_{n=0}^{2^{j-1}-1}} |m2^{j-1}+n\rangle(\alpha_{m2^j + 2n} \ket{0} + \i \alpha_{m2^j+2n+1}\ket{1})\ket{0}\ket{0}.
$$
After the second step, the state is
$$
\sum_{m\,s.t.\, \mathrm{W}^j_m\in \Lambda_T} \smashoperator{\sum_{n=0}^{2^{j-1}-1}} |m2^{j-1}+n\rangle(\alpha_{m2^j + 2n} \ket{0} + \i \alpha_{m2^j+2n+1}\ket{1})|\underbrace{1 \ldots 1}_j 0\ldots 0\rangle\ket{0}.
$$
After step 3a, the state is 
$$
\begin{gathered}
\sum_{m\,s.t.\, \mathrm{W}^j_m\in \Lambda_T} \smashoperator{\sum_{n=0}^{2^{j-1}-1}} |m2^{j-1}+n\rangle(\alpha_{m2^j + 2n} \ket{0} + \i \alpha_{m2^j+2n+1}\ket{1})|1 \ldots 1 0\ldots 0\rangle \\ \otimes \ket{I(n \ge \mu_0(j, 0) + 1)}.
\end{gathered}
$$
Note that $I(n \ge \mu_0(j,0) + 1) = I(n > 2^j / 3)$ as $n, j \in \mathbb{Z}$.
After step 3b, the state is 
\begin{align*}
\sum_{m\,s.t.\, \mathrm{W}^j_m\in \Lambda_T} &\smashoperator{\sum_{n=0}^{2^{j-1}-1}}|m2^{j-1}+n\rangle  
\\
&\otimes  \left( R\left( \tfrac{\pi}{2^{j+1}}(3n-2^j) I(m \text{ is even}) I(n>2^j/3)  \right) (\alpha_{m2^j + 2n}\ket{0} + \i \alpha_{m2^j+2n+1} \ket{1}) \right)  
\\
&\otimes |1 \ldots 1 0\ldots 0\rangle\ket{I(q_{j-1} \ldots q_1 \ge \mu_0(j, 0) + 1)}.
\end{align*}
After step 3c, the state is 
\begin{align*}
\sum_{m\,s.t.\, \mathrm{W}^j_m\in \Lambda_T} &\smashoperator{\sum_{n=0}^{2^{j-1}-1}}|m2^{j-1}+n\rangle \\
&\otimes \left( R\left( \tfrac{2^{I(n>2^j/3)}\pi }{2^{j+1}}(3n-2^j) I(m \text{ is even})  \right) (\alpha_{m2^j + 2n}\ket{0} + \i \alpha_{m2^j+2n+1} \ket{1}) \right) \\
&\otimes |1 \ldots 1 0\ldots 0\rangle\ket{I(q_{j-1} \ldots q_1 \ge \mu_0(j, 0) + 1)}.
\end{align*}
After step 3d, the state is 
\begin{align*}
\sum_{m\,s.t.\, \mathrm{W}^j_m\in \Lambda_T} &\smashoperator{\sum_{n=0}^{2^{j-1}-1}}|m2^{j-1}+n\rangle 
\\
&\otimes  \left( R\left( \tfrac{2^{I(n>2^j/3)}\pi }{2^{j+1}}(3n-2^j) I(m \text{ is even})  \right) (\alpha_{m2^j + 2n}\ket{0} + \i \alpha_{m2^j+2n+1} \ket{1}) \right)  
\\
&\otimes |1 \ldots 1 0\ldots 0\rangle\ket{0}.
\end{align*}
Steps 3a-3d actions are limited to even $m$'s. Steps 3e-3h are similar but applied to odd $m$'s. At the end of step 3, the state is 
\begin{align*}
\sum_{m\,s.t.\, \mathrm{W}^j_m\in \Lambda_T} &\smashoperator{\sum_{n=0}^{2^{j-1}-1}}|m2^{j-1}+n\rangle 
\\
&\otimes \Bigg( R\left( \tfrac{2^{I(n>2^j/3)}\pi }{2^{j+1}}(3n-2^j) I(m \text{ is even}) \right) R\left(\tfrac{2^{I(n<2^{j-1}/3)}\pi } {2^{j+2}}(3n-2^{j-1}) I(m \text{ is odd})  \right) \\ 
&(\alpha_{m2^j + 2n}\ket{0} + \i \alpha_{m2^j+2n+1} \ket{1}) \Bigg)  
\\
&\otimes |1 \ldots 1 0\ldots 0\rangle\ket{0},
\end{align*}
with the use of function $v(\cdot)$ we can re-write the state as 
$$
\sum_{m\,s.t.\, \mathrm{W}^j_m\in \Lambda_T} \smashoperator{\sum_{n=0}^{2^{j-1}-1}} |m2^{j-1}+n\rangle \otimes 
R(v(j, m, n)) (\alpha_{m2^j + 2n}\ket{0} + \i \alpha_{m2^j+2n+1} \ket{1})
 \otimes \ket{1\ldots 1 0\ldots 0}\ket{0}.
$$
After the fourth step, the state is 
$$
\sum_{m\,s.t.\, \mathrm{W}^j_m\in \Lambda_T} \smashoperator{\sum_{n=0}^{2^{j-1}-1}}|m2^{j-1}+n\rangle \otimes 
R(v(j, m, n)) (\alpha_{m2^j + 2n}\ket{0} + \i \alpha_{m2^j+2n+1} \ket{1})
\otimes \ket{0}\ket{0}.
$$
As 
$$
\alpha_{m2^j + 2n}\ket{0} + \i \alpha_{m2^j+2n+1} \ket{1} = S(\alpha_{m2^j + 2n}\ket{0} + \alpha_{m2^j+2n+1} \ket{1}),
$$
the final state is 
\begin{align*}
\sum_{m\,s.t.\, \mathrm{W}^j_m\in \Lambda_T} \smashoperator{\sum_{n=0}^{2^{j-1}-1}} |m2^{j-1}+n\rangle \otimes R^{(i)}(v(j,m,n)) (\alpha_{m2^j + 2n}\ket{0} + \alpha_{m2^j+2n+1}\ket{1}) \otimes \ket{0}\ket{0},
\end{align*}
where
$$
R^{(i)}(\theta) = \begin{pmatrix}
\cos(\theta) & \i \sin(\theta)\\
\i \sin(\theta) & \cos(\theta)
\end{pmatrix}.
$$
\end{proof}

%% file: Algorithms/R.tex
\ifthenelse{\boolean{originalformat}}{
 \subsection{Quantum circuits for $R$ (see section \ref{section:permutations})}
}{
 \subsection{Quantum circuits for $R$}
}
\label{appendix:alg-r}

\begin{algorithm}
\caption{Computing the functions $h^\rho_j(\cdot)$}\label{algo:h-rho}
\noindent Input: $\ket{\alpha}\ket{0}\ket{0}\ket{0} = \sum_{q=0}^{2^{L-1}-1} \alpha_q \ket{q}\ket{0}\ket{0}\ket{0}$.\\
Result: $\sum_{q=0}^{2^{L-1}-1} \alpha_q \ket{q}\ket{h_1^\rho(q),\ldots,h_{L-2}^\rho(q)}\ket{0}\ket{0}$.\\
Notation:\\
\indent\hspace{0.25cm}Quantum register $\mathcal{R}_q$ of $L-1$ qubits containing the input state $\ket{\alpha}$; $\mathcal{R}_q[i]$ corresponds to $i$-th qubit in register $\mathcal{R}_q$, $i=0, \ldots, L-2$;\\
\indent\hspace{0.25cm}Quantum register $\mathcal{R}'_h$ of $L-1$ qubits in initial state $\ket{0}$; $\mathcal{R}'_h[i]$ corresponds to $i$-th qubit in register $\mathcal{R}'_h$, $i=1, \ldots, L-1$;\\
\indent\hspace{0.25cm}Quantum register $\mathcal{R}_h$ of $L-1$ ancilla qubits in initial state $\ket{0}$; $\mathcal{R}_h[i]$ corresponds to $i$-th qubit in register $\mathcal{R}_h$, $i=1, \ldots, L-1$;\\
\indent\hspace{0.25cm}Quantum register $\mathcal{R}_I$ of $2$ ancilla qubits in initial state $\ket{0}$.
\begin{enumerate}
\item Compute $h_j$, $1 \le j < L$, by calling subroutines on the register $\mathcal{R}_q$ and ancilla qubit $R_I[0]$ as the least significant qubit,
    $$
    \underbrace{\ket{q}}_{\mathcal{R}_q} \underbrace{\ket{0}}_{\mathcal{R}'_h}\underbrace{\ket{0}}_{\mathcal{R}_h}\underbrace{\ket{0}}_{\mathcal{R}_I} \mapsto \ket{q}\ket{0}\ket{h_1(q), \ldots, h_{L-1}(q)}\ket{0}.
    $$
\item Apply $\mathrm{Compare}_{L-1}(\mu_0(j_l, 0)+1)$ to $\mathcal{R}_q$ with the result added to $\mathcal{R}_I[0]$;
\item Apply $\mathrm{Compare}_{L-1}(m_r2^{j_r-1}+\mu_0(j_r, m_r)+1)$ to $\mathcal{R}_q$ with the result added to $\mathcal{R}_I[0]$;
\item For each $j=2, \ldots, L-1$ controlled on that $\mathcal{R}_h[j] = \ket{1}$, $\mathcal{R}_h[j+1]=\ket{0}$ and $\mathcal{R}_I[0]=\ket{1}$:
\begin{enumerate}
    \item Apply $\mathrm{Compare}_j(\mu_0(j, 0)+1)$ to the first $j-1$ qubits in $\mathcal{R}_q$ controlled on that $\mathcal{R}_q[j-1]=\ket{0}$ and add the result to $\mathcal{R}_I[1]$;
    \item Apply $\mathrm{Compare}_j(\mu_0(j, 1)+1)$ to the first $j-1$ qubits in $\mathcal{R}_q$ controlled on that $\mathcal{R}_q[j-1]=\ket{1}$ and add the result to $\mathcal{R}_I[1]$;
    \item Apply $X$-gate to qubit $\mathcal{R}'_h[j]$ controlled on 
    $\mathcal{R}_I[1]=\ket{0}$ and $\mathcal{R}_q[j-1]=\ket{0}$;
    \item Apply $X$-gate to qubit $\mathcal{R}'_h[j]$ controlled on $\mathcal{R}_I[1]=\ket{1}$ and $\mathcal{R}_q[j-1]=\ket{1}$;
    \item Apply $X$-gate to qubit $\mathcal{R}'_h[j-1]$ controlled on $\mathcal{R}_I[1]=\ket{0}$ and $\mathcal{R}_q[j-1]=\ket{1}$;
    \item Apply $X$-gate to qubit $\mathcal{R}'_h[j-1]$ controlled on  $\mathcal{R}_I[1]=\ket{1}$, $\mathcal{R}_q[j-1]=\ket{0}$;
    \item Un-compute steps (b) and (a).
\end{enumerate}
\item Un-compute steps 3, 2, 1.
\end{enumerate}
\end{algorithm}
\begin{proof}[Proof of correctness of Algorithm \ref{algo:h-rho}]
Here $\mathcal{R}_I[0]=\ket{1}$ allows us to avoid permuting the border cases and computed in steps 2-3. $\mathcal{R}_I[1]=\ket{1}$ computed in step 4a-4b controls the condition $n > \mu(j, m)$. This condition depends on the parity of $m$ which is equal to $\mathcal{R}_q[j-1]$. Finally, steps 4c-4f implement the logic outlined in the main text.
\end{proof}